\documentclass[12pt,a4paper]{amsart} 
\usepackage{amsaddr}
\usepackage[T1]{fontenc}
\usepackage[margin=1.5cm]{geometry}
\usepackage[pdfusetitle]{hyperref}
\usepackage[normalem]{ulem}
\usepackage{wasysym}

\usepackage{amscd,amsmath,amssymb,amsfonts,xspace,mathrsfs,amsthm}
\usepackage{color, mathtools}
\usepackage[dvipsnames]{xcolor}

\usepackage{url}

\usepackage{latexsym}
\usepackage{graphicx}
\usepackage{dsfont}
\usepackage{longtable}
\usepackage{enumitem}

\usepackage[latin1]{inputenc}
\usepackage{multirow}

\newtheorem{theorem}{Theorem}
\newtheorem{proposition}{Proposition}
\newtheorem{corollary}{Corollary}
\newtheorem{lemma}{Lemma}
\newtheorem{conj}{Conjecture}

\newtheorem{definition}{Definition}
\newtheorem{remark}{Remark}

\usepackage{tikz-cd}
\usetikzlibrary{bending}
\usepackage{tikz,textcomp}
\usetikzlibrary{decorations.pathreplacing,calc,fadings,fit,shapes,arrows,positioning,chains,matrix}

\numberwithin{equation}{section}

\def\bea{\begin{eqnarray}}
\def\eea{\end{eqnarray}}
\def\be{\begin{equation}}
\def\ee{\end{equation}}
\def\ba{\begin{align}}
\def\ea{\end{align}}
\def\bse{\begin{subequations}}
\def\ese{\end{subequations}}

\newcommand{\nn}{\nonumber}

\def\det{\,{\rm det}\, }

\def\PTi{{\rm P}}

\def\Im{\,{\rm Im}\,}
\def\Re{\,{\rm Re}\,}

\def\PE{{\rm PE }}

\DeclareMathOperator{\Td}{Td}
\DeclareMathOperator{\ch}{ch}
\DeclareMathOperator{\rk}{rk}

\DeclareMathOperator{\Hom}{Hom}
\DeclareMathOperator{\Coh}{Coh}

\DeclareMathOperator{\Ext}{Ext}
\DeclareMathOperator{\Stab}{Stab}
\DeclareMathOperator{\Pic}{Pic}
\DeclareMathOperator{\Li}{Li}
\DeclareMathOperator{\Aut}{Aut}

\def\({\left(}
\def\){\right)}
\def\[{\left[}
\def\]{\right]}
\def\<{\left\langle}
\def\>{\right\rangle}
\def\hf{{1\over 2}}

\newcommand{\p}{\partial}

\renewcommand\v{\mathsf v}
\newcommand\w{\mathsf w}

\newcommand{\eps}{\epsilon}

\newcommand{\I}{\mathrm{i}}

\newcommand{\cA}{\mathcal{A}}

\newcommand{\cC}{\mathcal{C}}
\newcommand{\cD}{\mathcal{D}}
\newcommand{\cE}{\mathcal{E}}
\newcommand{\cF}{\mathcal{F}}

\newcommand{\cH}{\mathcal{H}}

\newcommand{\cL}{\mathcal{L}}
\newcommand{\cM}{\mathcal{M}}
\newcommand{\cN}{\mathcal{N}}
\newcommand{\cO}{\mathcal{O}}

\newcommand{\cS}{\mathcal{S}}
\newcommand{\cT}{\mathcal{T}}
\newcommand{\cU}{\mathcal{U}}

\newcommand{\cW}{\mathcal{W}}

\newcommand{\PP}{{\mathbb P}}

\newcommand{\Z}{{\mathbb Z}}

\newcommand{\IR}{\mathds{R}}
\newcommand{\IC}{\mathds{C}}
\newcommand{\IZ}{\mathds{Z}}
\newcommand{\IQ}{\mathds{Q}}

\newcommand{\IH}{\mathds{H}}

\newcommand{\IP}{\mathds{P}}

\def\scM{\mathscr{M}}

\def\tF{\tilde F}
\def\hq{\hat q}

\def\CY{\mathfrak{Y}}

\def\bOm{\overline{\Omega}}
\def\bOmPi{\lefteqn{\overline{\phantom{\Omega}}}\Omega^\Pi}

\def\hpol{h^{\rm (p)}}

\def\ths#1{\theta^{(#1)}}
\def\vths#1{\vartheta^{(#1)}}

\def\sf#1{{\color{green} #1}}

\def\mm{\ell_0}

\def\rmz{{\rm z}}
\def\brmz{\bar{\rm  z}}
\def\gmax{g_{\rm max}}
\def\gC{g_C}
\def\bOmH{\bOm_H}
\def\chiOD{\chi_{\cD}}

\newcommand{\q}{\mbox{q}}

\newcommand\PT{\operatorname{PT}}
\newcommand\DT{\operatorname{DT}}
\newcommand{\GLt}{\widetilde{GL^+}(2,\IR)}

\def\GW{{\rm GW}}
\def\GV{{\rm GV}}
\newcommand{\GVg}[2][Q]{{\GV}^{(#2)}_{#1}}

\title{
Quantum geometry, stability and modularity
}

\author{Sergei Alexandrov}
\address{Laboratoire Charles Coulomb (L2C), Universit\'e de Montpellier,
CNRS, \\ F-34095, Montpellier, France}
\email{sergey.alexandrov@umontpellier.fr}

\author{Soheyla Feyzbakhsh}
\address{ Department of Mathematics, Imperial College, London SW7 2AZ, United Kingdom}
\email{s.feyzbakhsh@imperial.ac.uk}

\author{Albrecht Klemm}
\address{Bethe Center for Theoretical Physics and 
Hausdorff Center for Mathematics,\\ Universit\"at Bonn, D-53115, Germany}
\email{aoklemm@th.physik.uni-bonn.de}

\author{Boris Pioline \ and \ Thorsten Schimannek }
\address{Sorbonne Universit\'e, CNRS, 
Laboratoire de Physique Th\'eorique et Hautes Energies, \\
Campus Pierre et Marie Curie, 4 place Jussieu, F-75005, Paris, France}
\email{pioline@lpthe.jussieu.fr, schimannek@lpthe.jussieu.fr}

\begin{document}
\setlength{\parskip}{0.2cm}

\begin{abstract}
By exploiting new mathematical relations between Pandharipande-Thomas (PT) invariants, 
closely related to Gopakumar-Vafa (GV) invariants, and rank 0 Donaldson-Thomas (DT) invariants counting D4-D2-D0 BPS bound states, we  rigorously compute the first few terms in the generating series of Abelian D4-D2-D0 indices for compact one-parameter Calabi-Yau threefolds of hypergeometric type. 
In all cases where  GV invariants can be computed to sufficiently high genus,
we find striking confirmation that the generating series is modular, and predict infinite series of Abelian D4-D2-D0 indices. Conversely, we use these results to provide new constraints for the direct integration method, which allows to compute GV invariants (and therefore the topological string partition function) to higher genus than hitherto possible. The  triangle of relations between GV/PT/DT invariants is powered by a new explicit formula relating PT and rank 0 DT invariants, which is proven in an Appendix by the second named author. 
As a corollary, we obtain rigorous Castelnuovo-type bounds for PT and GV invariants for CY threefolds with Picard rank one.

\end{abstract}

\maketitle

\newpage
\tableofcontents

\setlength{\parskip}{0.2cm}

\section{Introduction}

More than 25 years after Strominger and Vafa's celebrated breakthrough \cite{Strominger:1996kf}, 
the precision counting of BPS black hole microstates in string vacua with $\cN=2$ supersymmetry 
in 4 dimensions remains an outstanding challenge at the frontier of theoretical physics and mathematics. 
Unlike in cases with higher supersymmetry, the index $\Omega_z(\gamma)$ counting BPS states with fixed electromagnetic 
charge $\gamma$ has an intricate chamber structure with respect to the moduli $z$ specifying the internal manifold, 
while that moduli space is itself  subject to complicated quantum corrections. 
As a result, the indices  $\Omega_z(\gamma)$ are almost never known exactly.

For type IIA strings compactified on a Calabi-Yau (CY) threefold $\CY$, the proper mathematical framework 
involves the derived category of coherent sheaves $\cC=D^b\Coh \CY$, the associated space of 
Bridgeland stability conditions $\cS=\Stab\cC$ and the Donaldson-Thomas (DT) invariants $\Omega_\sigma(\gamma)$ 
counting semi-stable objects in $\cC$ with charge $\gamma$ for a stability condition $\sigma=(Z,\cA)\in\cS$,
where $Z$ is a central charge function and $\cA$ a certain Abelian subcategory of $\cC$ locally determined by $Z$. 
While physics (or rather mirror symmetry) selects a particular slice $\Pi\subset\cS$
where $Z$ is a computable
function of the (complexified) K\"ahler moduli $z\in\cM_K$, the DT invariants $\Omega_\sigma(\gamma)$ are in principle 
well-defined in the larger space $\cS$. In cases where $\Omega_\sigma(\gamma)$ can be
shown to vanish at some particular point $\sigma\in \cS$ (which need not belong to the physical slice $\Pi$), 
it then becomes possible to determine it elsewhere using the universal wall-crossing formulae of
\cite{ks,Joyce:2008pc}.

This strategy has been pursued in a recent series of mathematical 
works~\cite{Toda:2011aa, Feyzbakhsh:2020wvm,Feyzbakhsh:2021rcv,Feyzbakhsh:2021nds},
which culminated in explicit formulae \cite{Feyzbakhsh:2022ydn}
relating rank 0 DT invariants, counting D4-D2-D0 bound states,
to Pandharipande-Thomas (PT) invariants, counting D6-D2-D0 bound states with one unit of D6-brane charge. 
These rigorous results depend on a conjectural inequality
which lies at the heart of the construction of stability conditions on CY threefolds \cite{bayer2011bridgeland,bayer2016space}, 
and is widely believed to hold in general but
proven only in a handful of cases. PT invariants are in turn related to
Gopakumar-Vafa (GV) invariants entering the A-model topological string partition function $Z_{\rm top}$
on $\CY$ \cite{gw-dt}, and are in principle computable by integrating the holomorphic anomaly equations satisfied 
by $Z_{\rm top}$, a procedure sometimes called `direct integration' \cite{Bershadsky:1993ta,Huang:2006hq,Grimm:2007tm}. These relations between D4-D2-D0 indices and
topological strings are in the spirit of the OSV conjecture \cite{Ooguri:2004zv}, and in fact imply
a special case of the latter \cite{Denef:2007vg,Toda:2011aa,Feyzbakhsh:2022ydn}.

On the other hand, the fact that D4-D2-D0 bound states in type IIA string theory lift to 
M5-branes wrapped on a divisor $\cD\subset \CY$ times a circle strongly suggests that suitable 
generating series of rank 0 DT invariants should exhibit modular properties  \cite{Maldacena:1997de}. 
Specifically, in the simplest case of a single M5-brane wrapped on an ample divisor $\cD$, 
the corresponding series of rank 0 DT invariants, which we call Abelian D4-D2-D0 indices, should transform as a vector-valued modular form, arising from the theta-series decomposition of the elliptic genus of 
the $(0,4)$ superconformal field theory obtained by reducing
the M5-brane along $\cD$ \cite{deBoer:2006vg,Gaiotto:2006wm,Gaiotto:2007cd,Manschot:2007ha}.
More generally, for a reducible divisor  the generating series should 
transform as 
a vector-valued mock modular form of higher depth, with a fixed modular anomaly 
\cite{Alexandrov:2016tnf,Alexandrov:2017qhn,Alexandrov:2018lgp} (see \cite{Manschot:2010sxc,Alim:2010cf,Dabholkar:2012nd,Cheng:2017dlj} for related work).
Since the space of such vector-valued (mock) modular forms is finite-dimensional, 
this opens up the possibility of computing infinite families of D4-D2-D0 indices, 
provided the singular terms in the generating series (also known as polar terms)
can be determined independently. 

This approach was applied long ago for a few CY threefolds $\CY$
with $b_2(\CY)=1$ in \cite{Gaiotto:2006wm,Gaiotto:2007cd,Collinucci:2008ht,VanHerck:2009ww}. It was extended recently in \cite{Alexandrov:2022pgd} to the full list of 13 smooth complete intersections 
in weighted projected space (the so-called hypergeometric CY threefolds), see Table \ref{table1}. Unfortunately, the analysis in \cite{Alexandrov:2022pgd} 
was based on an educated guess for the coefficients of the polar terms, which reproduced earlier results in
\cite{Gaiotto:2006wm,Gaiotto:2007cd} and provided plausible answers for 5 additional models, 
but failed to produce a modular form for the last 3 models in Table \ref{table1}.
Although a strategy to compute non-Abelian D4-D2-D0 indices was spelled out, it was eventually inconclusive, again due to lack of control on the polar coefficients.

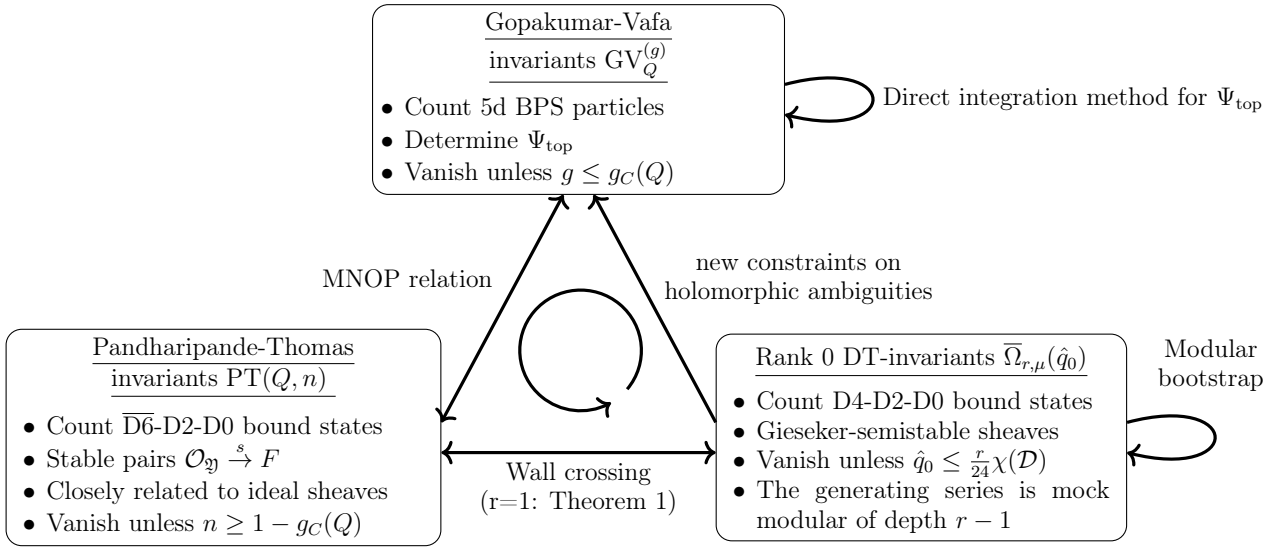
\begin{figure}[t!]
\centering
\begin{tikzpicture}[
scale=.8, every node/.style={scale=0.8},
dbox/.style = {draw, rectangle, align=center, rounded corners=0.5em,fill=white!20}
]
\node [dbox, minimum width=1cm, minimum height=3cm, align=center ] (GV)   {
	\underline{Gopakumar-Vafa}\\\underline{invariants $\GVg{g}$}\\[.2cm]%
\begin{minipage}{6.5cm}
\begin{itemize}[leftmargin=*]\setlength\itemsep{.1em}
\item Count 5d BPS particles
\item Determine $\Psi_{\rm top}$
\item Vanish unless $g\leq g_C(Q)$
\end{itemize}
\end{minipage}  
};  
\node [ minimum width=1cm, minimum height=1cm, align=center, below= of GV  ] (empty1)   {};  
\node [ minimum width=2cm, minimum height=1cm, align=center, below= of empty1  ] (empty2)   {};  

 \node[dbox, minimum width=1cm, minimum height=3cm, align=center, left=of empty2] (PT)   {
	 \underline{Pandharipande-Thomas}\\\underline{invariants $\text{PT}(Q,n)$ }\\[.2cm]
	 {
\begin{minipage}{6.6cm}
\begin{itemize}[leftmargin=*]\setlength\itemsep{.1em}
\item Count $\overline{\rm D6}$-D2-D0 bound states
\item Stable pairs $\cO_\CY\stackrel{s}{\to} F$
\item Closely related to ideal sheaves
\item Vanish unless $n\geq 1-g_C(Q)$
\end{itemize}\end{minipage}
	}
};  
 \node[dbox, minimum width=1cm, minimum height=3cm, align=center, right=of empty2] (BH)   {
	 \underline{Rank 0 DT-invariants $\bOm_{r,\mu}(\hq_0)$}\\[.2cm]
	 {
\begin{minipage}{6.2cm}
\begin{itemize}[leftmargin=*]\setlength\itemsep{.0em}
\item Count D4-D2-D0 bound states 
\item Gieseker-semistable sheaves 
\item Vanish unless $\hq_0\leq \frac{r}{24}\chi(\cD)$
\item The generating series is mock modular of depth $r-1$
\end{itemize}\end{minipage}
}
};
\draw[<->, very thick, align=center] ([shift={(0,-.25)}]PT.east) to node[align=center, below] {Wall crossing\\(r=1: Theorem 1)}([shift={(0,-.25)}]BH.west);
\draw[<->, very thick, align=center] ([shift={(-.25,0)}]GV.south) to node[align=center, left, shift={(0,.5)}] {MNOP relation}([shift={(0,.25)}]PT.east);
\draw[->, very thick, align=center] ([shift={(0,.25)}]BH.west) to node[align=center, right, shift={(0,.5)}] {new constraints on\\holomorphic ambiguities}([shift={(.25,0)}]GV.south);
\draw[->, very thick, align=center] ([shift={(0,.25)}]BH.east) to [loop right,looseness=5,min distance=20mm] node[align=center, right, shift={(-1.1,1.2)}] {Modular\\bootstrap}([shift={(0,-.25)}]BH.east);
\draw[->, very thick, align=center] ([shift={(0,.25)}]GV.east) to [loop right,looseness=5,min distance=20mm] node[align=center, right, shift={(0,0)}] {Direct integration method for $\Psi_{\rm top}$}([shift={(0,-.25)}]GV.east);

\draw[->, very thick, align=center] ([shift={(.8, .8)}]empty2) arc (-40:300:1) ;
\end{tikzpicture}
\caption{The triangle of relations between GV/PT/DT invariants \vspace{-0.01cm}}
\label{figflow}
\end{figure}

In this work, we revisit the analysis in \cite{Alexandrov:2022pgd} in light of 
the recent mathematical results in \cite{Feyzbakhsh:2022ydn}. More specifically,  
we exploit a new and powerful explicit formula \eqref{thmS11} relating PT and rank 0 DT invariants, which is proven by one of the authors in Appendix \ref{sec_appS} of this paper, and depicted by the horizontal arrow at the bottom of Figure \ref{figflow}. Among other applications, this formula allows to prove rigorous Castelnuovo-type bounds for PT and GV invariants, and determines the GV invariants $\GVg{g}$ for maximal genus $g=\gmax(Q)$, assuming some congruence condition on the degree $Q$. Along with various optimizations of the computer implementation, this allows us to push the direct integration method of \cite{Huang:2006hq} to high genus.
By combining the formula \eqref{thmS11} with these results for GV invariants, 
we are able to rigorously compute all polar terms and a large number of non-polar terms for most of the 13 hypergeometric CY threefolds, 
and find striking confirmations of the modularity of the
corresponding generating series (as well as supporting evidence  for
the validity of the BMT inequality in those models where it is not yet known to hold). 
Expanding these generating series to arbitrary order, we predict an infinite set of Abelian D4-D2-D0 indices.

Turning  the logic around and assuming that the generating series of Abelian D4-D2-D0 indices is indeed the one dictated by 
modularity, we predict infinite series of GV invariants $\GVg{g}$ lying at finite distance from the Castelnuovo bound $g=g_{C}(Q)$. This in turn provides additional 
boundary conditions for the direct integration method, which in principle allows us to push it beyond the maximal
genus (indicated as $g_{\rm integ}$ in Table \ref{table1}) at which the leading behaviour at special points in the moduli space and the Castelnuovo vanishing conditions no longer suffice to fix the holomorphic ambiguities. The maximal genus  attainable using these additional boundary conditions is indicated in the column $g_{\rm mod}$ in Table \ref{table1}. The updated data are available at~\cite{CYdata}.

More specifically, we find the generating series of D4-D2-D0 indices for 11 out of 13 models listed in Table \ref{table1}.
For 5 models, namely $X_{10}$, $X_{4,3}$, $X_{6,2}$, $X_{6,4}$ and $X_{4,2}$,
our results imply that the polar terms differ from the naive Ansatz of \cite{Alexandrov:2022pgd} (in particular,
the result for $X_{10}$ disagrees with \cite{Gaiotto:2007cd} but confirms the  proposal in \cite{VanHerck:2009ww}). In all these cases, we find spectacular confirmation that the generating series is modular.  
For the last 2 models in this Table, namely $X_{3,2,2}$ and $X_{2,2,2,2}$, we are not yet able to uniquely fix the generating series due to our limited knowledge of GV invariants for these models. 

The outline of this work is as follows. In \S\ref{sec_prim}, we give a rather extensive introduction 
to the main mathematical concepts which underlie this work, 
including the space of Bridgeland stability conditions on the derived
category of coherent sheaves $\cC=D^b\Coh\CY$ and the associated generalized DT invariants. 
We also introduce the family of weak stability conditions $\nu_{b,w}$, which plays a central role in relating
rank 0  DT invariants and PT invariants, and spell out the expected modular properties 
of generating series of Abelian D4-D2-D0 indices. In \S\ref{sec_GVdirect}, we recall the relation between PT invariants and GV invariants, and explain how the latter can be computed using the direct integration method. We further give a heuristic computation of GV invariants for maximal genus $g=\gmax(Q)$ and 
submaximal genus $g=\gmax(Q)-1$, which is confirmed in \S\ref{sec_wcrm1}  as a consequence of Theorem 1 in Appendix \S\ref{sec_appS}. In \S\ref{sec_rank0fromPT},
we explain the main results of Appendix \ref{sec_appS} in more physical terms,
starting in \S\ref{sec_wcr0} with Theorem \ref{thm.rk0} which expresses 
D4-D2-D0 indices as contributions of D6-$\overline{\rm D6}$-bound states, but whose
applicability is limited to the most polar terms, and continuing in \S\ref{sec_wcrm1} with Theorem \ref{thm-main}, 
which is less transparent physically but of much wider applicability.
In \S\ref{sec_test} we use Theorem \ref{thm-main} to compute D4-D2-D0 indices and test modularity in three representative models, namely $X_5$, $X_{10}$ and $X_{4,2}$, leaving the details of other models to Appendix \ref{sec_gen}. Finally, in \S\ref{sec-disc} we summarize our findings and discuss avenues for future research. Extensive tables 
of GV, PT and DT invariants computed in the course of this project are available in Mathematica-readable form at the website~\cite{CYdata}.

\begin{table} [t]
\begin{centering}
$$
\begin{array}{|l|r|r|r|r|r|r|r|r|r|r|r|r|}
\hline \CY  & \chi_\CY& \kappa  &c_{2}  & \chiOD     & n_1^p  & n_1^c & (a_i) 
& {\rm type} & g_{\rm integ} & g_{\rm mod} & g_{\rm avail}\\   \hline
X_5(1^5)   
& -200   &5   &  50 &   5   & 7 & 0 
& (\frac15,\frac25,\frac35,\frac45)
& F  & 53 & 69 & 64   \\
 X_6(1^4,2)  
 & -204& 3 & 42& 4  & 4 & 0 & (\frac16,\frac26,\frac46,\frac56)
 &F & 48 & 66 & 48  \\
X_8(1^4,4)  
&-296   &2              &  44  &   4 &    4 & 0  & (\frac18,\frac38,\frac58,\frac78) &F  & 60 & 84 & 64  \\
 X_{10}(1^3,2,5)  
 &  -288&      1 & 34& 3  & 2   & 0 & (\frac1{10},\frac3{10},\frac7{10},\frac9{10}) &F & 50 & 70 & 68 \\
 X_{4,3}(1^5,2)  
& -156   &6              &  48 &    5  &  9 & 0 &  (\frac14,\frac13,\frac23,\frac34) &F & 20 &24 & 24    \\
X_{6,4}(1^3,2^2,3) 
 & -156&         2 & 32& 3  &  3  & 0  &  (\frac16,\frac14,\frac34,\frac56) &F & 14 & 17 & 17    \\
 X_{3,3}(1^6)  
& -144   &9         &  54   &   6  & 14  & 1 &  (\frac13,\frac13,\frac23,\frac23) &K & 29 & 33 & 33 
\\
 X_{4,4}(1^4,2^2)  
 & -144&          4 & 40& 4   &  6  & 1 & (\frac14,\frac14,\frac34,\frac34) &K & 26 & 34 & 34   \\
X_{6,6}(1^2,2^2,3^2)  
&-120   &1      &  22 & 2 &  1  & 0 &  (\frac16,\frac16,\frac56,\frac56)
&K & 18 & 21 & 21   \\
X_{6,2}(1^5,3)  
&-256   &4             &  52  &    5  &  7 & 0 & (\frac16,\frac12,\frac12,\frac56) &C & 63 & 84 & 49   \\
 X_{4,2}(1^6)  
 & -176&   8 & 56& 6  &  15  & 1&  (\frac14,\frac12,\frac12,\frac34)
 &C & 50 & 64 & 50  \\
X_{3,2,2}(1^7) 
&-144   &12              &  60 &    7  &  21 & 1 &  (\frac13,\frac12,\frac12,\frac23)  &C & 14 & ? & 14  \\
 X_{2,2,2,2}(1^8) 
& -128&       16 & 64 & 8  & 33  & 3 &  (\frac12,\frac12,\frac12,\frac12)
& M & 17 & ? & 32  \\
\hline
\end{array}
$$
\caption{Relevant data for the 13 hypergeometric CY threefolds. The second to fifth columns indicate the Euler number of $\CY$, the  self-intersection $\kappa=H^3$ of the generator of $\Pic\CY$, the second Chern class $c_2=c_2(T\CY).H$ and the holomorphic Euler characteristic $\chiOD=\chi(\cO_\cD)$ of the primitive divisor $\cD$ dual to $H$ (not to be confused with its topological Euler characteristic $\chi(\cD)$). 
The columns $n_1^p$ and $n_1^c$ indicate the number of polar terms and modular constraints on the generating series of Abelian D4-D2-D0 invariants, taken from \cite{Alexandrov:2022pgd}.
The columns $(a_i)$ and "type"  indicate the local exponents in the Picard-Fuchs equation and the resulting degeneration type at $z=\infty$ in the notation of \cite{Joshi:2019nzi}.
The column $g_{\rm integ}$ and $g_{\rm mod}$  indicate the maximal genus 
for which GV invariants $\GVg{g}$ can be determined by the direct integration method, either using only the usual regularity conditions
and the expression \eqref{eqn:gvgmax} for GV invariants saturating the bound 
$g\leq g_{\rm max}(Q)$ for $Q=0 \mod \kappa$, 
or also including GV invariants predicted by the modular series of Abelian D4-D2-D0 indices.
The column $g_{\rm avail}$ indicates the genus up to which 
complete tables of GV invariants are currently known. For updates check~\cite{CYdata}.  
\label{table1}}
\end{centering}
\end{table}

\subsection*{Glossary of invariants}
For the reader's convenience we summarize the notations for the various types of enumerative invariants that appear in this work.
More details will be provided in the corresponding sections.

We generally denote by $\bOm_\bullet(\gamma)\in\mathbb{Q}$ the rational Donaldson-Thomas invariants counting $\bullet$-semistable objects of class $\gamma$ defined as in~\cite{Joyce:2008pc}, where $\bullet$ denotes a (weak) stability condition or a limit thereof, and by $\Omega_\bullet(\gamma)$ the (conjecturally integral) generalized Donaldson-Thomas invariants obtained from $\bOm_\bullet(\gamma)$ via the
`multicover formula' \eqref{eqn:gDTdef}. 
This applies to the following invariants:
\begin{itemize}
	\item $\bOm_\sigma$, with $\sigma$ a general (weak) stability condition on $\cC=D^b\Coh(\CY)$, introduced in~\S\ref{sec_DTdef};
	\item $\bOm_{b,w}=\bOm_{\nu_{b,w}}$, with $\nu_{b,w}$ the slope function~\eqref{noo} on the heart $\cA_b$;
	\item $\bOm_\infty=\lim\limits_{w\rightarrow+\infty}\bOm_{b,w}$, introduced above~\eqref{HilbP};
	\item $\bOm_H$ counting Gieseker-semistable sheaves with respect to an ample class $H$, defined below~\eqref{HilbP};
	\item $\bOmPi_{\rmz}$, the DT invariant along the $\Pi$-stability slice, defined in \S\ref{sec_pi}.
\end{itemize}
We deviate from this notation for the D4-D2-D0 index $\bOm_{r,\mu}(\hq_0)$ introduced in \S\ref{sec_modconj}, which determines the
rank 0 DT invariant $\bOmPi_{\rmz}(0,r,q_1,q_0)$
in the large volume attractor chamber. 
In the special case of CY threefolds with Picard group $\Pic\CY=H \IZ$, it coincides with the index
$\bOmH(\gamma)$, see  ~\eqref{defOmrmu}.
In \S\ref{sec_DTLV}, we also introduce lighter notations for rank $\pm 1$ DT invariants at large volume,
\begin{itemize}
    \item Donaldson-Thomas invariants ${\mathrm I}_{n,\beta}=\DT(\beta. H,n)$;
	\item Pandharipande-Thomas invariant ${\mathrm P}_{n,\beta}=\PT(\beta. H,n)$.
\end{itemize}
As explained  in \S\ref{sec_GVPT}, these invariants
are closely related to Gromov-Witten  invariants $\GW^{(g)}_Q\in\mathbb{Q}$ and Gopakumar-Vafa  invariants $\GVg{g}\in\mathbb{Z}$.

 \subsection*{Acknowledgments}
 The authors are grateful to  Pierre Descombes, Amir-Kian Kashani-Poor, Sheldon Katz, Bruno Le Floch, 
 Emmanuel Macr\`i, Richard Thomas  for useful discussions.
 SA and BP are especially grateful to Jan Manschot and Nava Gaddam for collaboration
 on the earlier work \cite{Alexandrov:2022pgd}. AK likes to thank Yongbin Ruan for discussions, Oliver Freyermuth and Andreas Wisskirchen for computer support 
 and Claude Duhr and Franziska Porkert for making computer resources  available. 
The research of BP and TS is supported by Agence Nationale de la Recherche under contract number ANR-21-CE31-0021. SF acknowledges the support of EPSRC postdoctoral fellowship EP/T018658/1.

\section{Preliminaries}
\label{sec_prim}

In this section, we recall the basic definitions of the mathematical structures which we use in this work, emphasizing their physical interpretation. In \S\ref{sec_dcoh} we introduce the derived category of coherent sheaves $\cC=D^b\Coh\CY$, which formalizes the notion of BPS states in type IIA string theory compactified on a Calabi-Yau threefold $\CY$. 
In \S\ref{sec_DTdef} we recall the definition of the space of Bridgeland stability conditions $\Stab\cC$ and the associated 
generalized Donaldson-Thomas invariants $\Omega_{\sigma}(\gamma)$, which are the mathematical counterpart of BPS indices.
In \S\ref{sec_stab} we review the mathematical construction of Bridgeland stability conditions in an open set around the large volume point. As an intermediate step, 
we introduce a two-parameter family of weak stability conditions defined by the central charge \eqref{Ztilt} which will play a central role in \S\ref{sec_rank0fromPT}.  
In \S\ref{sec_pi} we identify the physical slice of $\Pi$-stability conditions inside $\Stab\cC$. 
In \S\ref{sec_modconj}, we introduce the rank 0 DT invariants counting D4-D2-D0 bound states, and state the 
modular properties of generating series of these invariants predicted by string theory arguments, 
restricting to the Abelian case (one unit of D4-brane charge).
Finally, in \S\ref{sec_DTLV} we introduce the rank 1 DT and PT invariants,
$\DT(Q,n)$ and $\PT(Q,n)$, which count bound states with $\pm 1$ unit of D6-brane charge at large volume. 
Their relation to Gopakumar-Vafa invariants is deferred to \S\ref{sec_GVdirect}. 

After reading \S\ref{sec_dcoh} where notations for charge vectors are introduced, 
a reader uninterested in mathematical details may skip ahead to \S\ref{subsec-summary}, 
where we briefly summarize the necessary mathematical constructions. 
In the last two subsections we introduce
the main objects studied in this work, namely the D4-D2-D0 indices and the 
rank 1 DT and PT invariants.

\subsection{BPS branes and derived category of coherent sheaves}
\label{sec_dcoh}

 As explained in \cite{Douglas:2000ah,Douglas:2000gi,Aspinwall:2004jr}, BPS states in type IIA  string theory compactified on a Calabi-Yau (CY) threefold $\CY$ are identified with B-branes in the A-twisted topological sigma model on $\CY$. Mathematically, they are best understood as objects in the bounded derived category of coherent sheaves 
$\cC=D^b\Coh \CY$. Such an object is a bounded complex
\begin{align}
	E=\(\ldots \stackrel{d^{-2}}{\rightarrow} \cE^{-1}\stackrel{d^{-1}}{\rightarrow} \cE^0 \stackrel{d^0}{\rightarrow} \cE^1
	\stackrel{d^1}{\rightarrow} \dots\)\,,
	\label{eqcomplex}
\end{align}
 where at each place $k\in\IZ$, $\cE^k$ is a coherent sheaf on $\CY$ which vanishes for all but a finite set of indices $k$, and $d^k:\cE^k \to \cE^{k+1}$ 
a morphism such that $d^{k+1} d^k=0$ for all $k\in\IZ$. Up to quasi-isomorphisms (which preserve the cohomology of the complex and physically correspond  to irrelevant boundary deformations), the coherent sheaf $\cE^k$ can be assumed to be a vector bundle on $\CY$, and is physically interpreted as a stack of wrapped D6-branes for $k$ even, respectively anti-D6-branes for $k$ odd. The morphism $d^k$ is then interpreted as an open string tachyon field. More generally, the extension group $\Ext^n(E,E'):=\Hom(E,E'[n])$, where
$[n]$ is the translation functor mapping $E=(\cE^k,d^k)_{k\in\IZ} \mapsto E[n]=(\cE^{k-n},d^{k-n})_{k\in \IZ}$, is interpreted
physically as the space of open strings of ghost number $n$. 

Besides the grading by ghost number, the category $\cC$ is also graded by the numerical Grothendieck group $K(\cC)$, which plays the role of the lattice of electromagnetic charges. Using the Chern character map $E\mapsto \ch(E) = \sum_{k} (-1)^k \ch(\cE^k)$, $K(\cC)$ can be identified with the lattice $\Gamma\subset  H^{\rm even}(\CY,\IQ)$ spanned by vectors $\v=(\ch_0,\ch_1,\ch_2,\ch_3)$ satisfying the quantization conditions \cite[Theorem 4.19]{Joyce:2008pc}
\be
\begin{split}
\ch_0 \in H^0(\CY,\IZ), &\qquad 
\ch_1 \in H^2(\CY,\IZ), 
\\
\ch_2-\frac12 \ch_1^2 \in H^4(\CY,\IZ), 
&\qquad
\ch_3+\frac12 c_2(T\CY) \ch_1 \in H^6(\CY,\IZ).
\end{split}
\ee
The respective integer cohomology classes correspond physically to the D6, D4, D2 and D0 brane charges. The lattice $\Gamma$ is endowed with the integer skew-symmetric pairing 
\be
\label{dsz}
\langle \ch(E), \ch(E') \rangle := \int_{\CY} (\ch E')^\vee \ch (E)  \Td(T\CY)\, ,
\ee
where $\vee$ acts as $(-1)^p$ on a form of degree $2p$ and $\Td(T\CY)=1+\frac1{12} c_2(T\CY)$ is the Todd class of the tangent bundle. This pairing is
skew-symmetric due to Serre duality 
$\Ext^n(E,E')=\Ext^{3-n}(E',E)$, and integer valued by the Grothendieck-Riemann-Roch (GRR) theorem, which identifies it with 
the alternating sum of the dimensions 
\be
\label{GRR}
\chi(E',E):=
\sum_{n} (-1)^n \dim \Ext^n(E',E) = \langle \ch(E), \ch(E') \rangle\, .
\ee 
Physically, \eqref{dsz} is interpreted as the Dirac-Schwinger-Zwanziger pairing between electromagnetic charge vectors. 
It is useful to introduce the Mukai vector\footnote{Note that a different convention 
$\gamma(E)=\ch(E)^\vee \sqrt{\Td(T\CY)}$ also appears in the literature.}
\be
\label{Mukaimap}
\gamma(E)=\ch(E) \sqrt{\Td(T\CY)}\, ,
\ee 
such that the pairing \eqref{dsz} takes the Darboux form  $\int_{\CY}  
\gamma(E')^\vee \, \gamma(E)$. We shall abuse notation and denote it by 
$\langle \v,\v'\rangle$ or $\langle \gamma,\gamma'\rangle$ interchangeably. 
We note that both $\ch(E)$ and $\gamma(E)$ change sign under the translation functor $E\mapsto E[1]$, corresponding to CPT symmetry in physics, which maps D-branes to anti-D-branes. Instead, the transformation $\ch(E)\mapsto (\ch E)^\vee$ follows by taking the derived dual
$E\mapsto E^{\vee}$, which is the physical counterpart of a parity transformation. 

In this paper, unless mentioned otherwise, we always assume that $\CY$ 
is a smooth projective CY threefold $\CY$  with $b_2(\CY)=1$ and $H^2(\CY,\IZ)_{\rm tors}=0$.
This last property holds for any general complete intersection in weighted projective spaces  by a generalisation of Grothendieck-Lefschetz theorem proved in \cite[Theorem 1]{gl}, in particular for all models in Table 1.
We denote by $H$ the generator of 
$\Lambda\coloneqq H^2(\CY,\IZ)=H \IZ$. The lattice $\Lambda^*=H^4(\CY,\IZ)$ is then generated by $H^2/\kappa$ where
$\kappa=\int_{\CY} H^3$. Poincar\'e duality maps $H$ to a primitive divisor class
$[\cD]$ in $H_4(\CY,\IZ)$, where $\cD$ is an ample divisor with cubic self-intersection  $\kappa=[\cD]^3$, and $H^2/\kappa$ to a primitive curve class $[C]\in H_2(\CY,\IZ)$.
 
We identify the Chern character $\ch(E)$ with the vector of rational numbers 
\be
\label{defC0123}
[C_0,C_1,C_2,C_3](E) := \int_{\CY} [H^3 \ch_0(E)\,,\, H^2.\ch_1(E)\,,\, H.\ch_2(E)\,,\, \ch_3(E)] \in \IQ^4\,,
\ee 
such that $\ch=(C_0+C_1 H + C_2 H^2+ C_3 H^3)/\kappa$. 
Its components satisfy the quantization conditions
\be
C_0\in \kappa \IZ, 
\qquad 
C_1\in  \kappa \IZ,
\qquad
C_2 \in \IZ+\frac{C_1^2}{2 \kappa}, 
\qquad
C_3 \in \IZ - \frac{c_2}{12\kappa}\, C_1,
\label{defCvec}
\ee
where we use the shorthand notation $c_2:=H.c_2(T\CY)$. 
We also define the charge vector $\gamma(E)=(p^0,p^1,q_1,q_0)$ 
obtained by expanding the Mukai vector \eqref{Mukaimap} as in \cite[(4.8)]{Alexandrov:2010ca},
\be
\label{defMukai}
\gamma(E)
= p^0 + p^1 H -  \frac{q_1}{\kappa}\, H^2
+ \frac{q_0}{\kappa}\,H^3\, .
\ee 
The Chern and Mukai vectors are related by  
\bea
\label{Mukaibasis}
p^0 &=& \ch_0,
\quad
p^1 = \frac{1}{\kappa}\, H^2.\ch_1, 
\quad 
q_1 
= - H.\ch_2 - \frac{c_2}{24\kappa}\, H^3 \ch_0,
\quad 
q_0= \ch_3+\frac{c_2}{24\kappa}\, H^2.\ch_1,
\eea
such that 
\be
\label{quant}
p^0\in\IZ, 
\qquad 
p^1\in\IZ, 
\qquad 
q_1\in \IZ+\frac{\kappa}{2} \,(p^1)^2 - \frac{c_2}{24}\, p^0, 
\qquad 
q_0\in \IZ-\frac{c_2}{24}\, p^1.
\ee
In this basis, the Dirac pairing \eqref{GRR} takes the Darboux form 
\be
\langle \gamma, \gamma'\rangle = q_0 p'^0 + q_1 p'^1 - q'_1 p^1 - q'_0 p^0\, .
\ee
Under the action of the auto-equivalence $E\mapsto E( k):=E\otimes \cO_{\CY}( k H)$ with $ k\in\IZ$, the Chern character transforms as $\ch (E)\mapsto e^{ k H} \ch(E)$, while the components of the Mukai vector transform as 
\be
\label{specflowD6}
\begin{split}
p^0\mapsto p^0,
\qquad &
p^1\mapsto p^1+ k p^0,
\qquad
q_1 \mapsto q_1 - \kappa k\, p^1  -\frac{\kappa k^2}{2}\, p^0\, ,
\\
&\,q_0 \mapsto q_0 - k q_1 +\frac{\kappa  k^2}{2}\,  p^1 +
\frac{ \kappa  k^3}{6}\, p^0 \, .
\end{split}
\ee
We refer to this transformation as a spectral flow.

For later reference, we record the Mukai vectors for the primitive D6, D4, D2 and D0-branes, represented by the structure sheaves of the threefold $\CY$,
of the ample divisor $\cD$, of the curve $C$ and of
a point $x\in \CY$ Poincar\'e dual to $H^3/\kappa$,  
\be
\label{gammastruc}
\begin{split}
\gamma(\cO_\CY)=& \(1,0,-\frac{c_2}{24},0\),
\qquad 
\gamma(\cO_\cD)=\(0,1,\frac{\kappa}{2}, \frac{\kappa}{6} + \frac{c_2}{24}\),
\\
\gamma(\cO_C)=&\, (0,0,1,-1),
\qquad \qquad\
\gamma(\cO_x)=(0,0,0,1) .
\end{split}
\ee
It is immediate to check that the quantization conditions \eqref{quant} are obeyed, using the
fact that   $\chiOD\coloneqq \chi(\cO_{\cD})=\frac{\kappa}{6} + \frac{c_2}{12}$ \sf{$ $} is integer (and equal
to the arithmetic genus plus one).

\subsection{Bridgeland stability conditions and Donaldson-Thomas invariants}
\label{sec_DTdef}

Physically, BPS states are elements in the point particle spectrum  whose mass $M$ saturate the Bogomolnyi-Prasad-Sommerfeld bound $M\geq |Z(\gamma)|$, where $Z(\gamma)$ is a central generator in the 
super-Poincar\'e algebra, which depends linearly on the electromagnetic charge vector $\gamma$ and is otherwise a transcendental function of the complexified K\"ahler moduli $z\in \cM_K(\CY)$. The BPS index $\Omega_z(\gamma)$ counts the number of BPS states with charge $\gamma$, weighted with a sign $(-1)^{2J_3}$ where $J_3$ is the projection of the angular momentum along a fixed axis, such that 
$\Omega_z(\gamma)$ becomes robust under complex deformations of $\CY$.
Mathematically, this is formalized by introducing
the notion of $\Pi$-stability conditions, which are special cases of Bridgeland stability conditions\footnote{Stability conditions are defined on triangulated categories, which include the data of a translation functor $E\mapsto E[1]$ and a collection of distinguished triangles $A\to B\to C\to A[1]$ satisfying various axioms. The derived category of coherent sheaves is automatically endowed with a triangulated structure. For simplicity, we 
conflate distinguished triangles with short exact sequences $0\to A\to B\to C\to 0$.}, and the associated generalized Donaldson-Thomas invariants.

A Bridgeland stability condition consists of a pair $\sigma=(Z,\cA)$ satisfying the following axioms \cite{MR2373143}:
\begin{itemize}
\item[i)] $Z:\Gamma\to \IC$ is a linear map, known as the (holomorphic) central charge (we abuse notation and denote $Z(E)=Z(\gamma(E))$ for any $E\in\cC$);
\item[ii)]  $\cA$ is the heart of a bounded $t$-structure on $\cC$ (i.e. $\cA=\cD^{\leq 0}\cap \cD^{\geq 0}$ where $(\cD^{\leq 0},\cD^{> 0})$  is a pair of orthogonal subcategories of $\cC$ which are invariant under the left and right translation functors $[1]$ and $[-1]$, respectively),
in particular $\cA$ is
an Abelian subcategory of $\cC$;
\item[iii)] For any non-zero $E\in \cA$, the central charge  $Z(E)$ is contained in 
the Poincar\'e upper half-plane $\IH_B=\IH \cup (-\infty,0)$, i.e. 
 $Z(E)=\rho(E) e^{\I\pi \phi(E)}$ where $\rho(E)>0$
and $0<\phi(E)\leq 1$;
\item[iv)] {\it (Harder-Narasimhan property)} Every non-zero $E\in \cA$ admits a finite filtration $0 \subset E_0\subset E_1 \cdots \subset E_n =E$ by objects $E_i$ in $\cA$, such that each factor $F_i\coloneqq E_i/E_{i-1}$ is $\sigma$-semistable (as defined below)
and $\phi(F_1)>\phi(F_2)\cdots > \phi(F_n)$;
\item[v)] {\it (Support property)} There exists a constant $C>0$ such that, for all  $\sigma$-semistable objects $E\in \cA$, $\| \gamma(E) \| \leq C \, |Z(E)|$
where $\| \cdot \|$ is any fixed Euclidean norm on $\Gamma\otimes \IR$.
\end{itemize}
In the last two items above, an object $F\in \cA$ is called  $\sigma$-semistable if $\phi(F')\leq \phi(F)$
for every non-zero subobject $F'$ of $F$. More generally, an object $F\in \cC$ is called $\sigma$-semistable
if there exists $n\in\IZ$ such that $F[n]\in\cA$ and $F[n]$ is  $\sigma$-semistable in the previous sense. For most purposes in this paper, we shall only need the notion of {\it weak stability condition} (as defined in \cite[Appendix B]{bayer2016space}), which essentially amounts to relaxing the axiom iii) and allowing $\cA$ to contain objects with vanishing central charge. 

For any weak stability condition $\sigma$ (subject to certain technical conditions spelled out in ~\cite{Joyce:2008pc})
and any charge vector $\gamma\in \Gamma$, one defines the generalized Donaldson-Thomas invariant $\Omega_\sigma(\gamma)$ as follows. Let $\cM_\sigma(\gamma)$ be the moduli stack
of $\sigma$-semistable objects in $\cA$ with $\gamma(E)=\pm \gamma$, where the sign is chosen such that  $\pm Z(\gamma)\in \IH_B$. 
If $\gamma$ is primitive and $\sigma$ generic, $\Omega_\sigma(\gamma)$ can be defined as the weighted Euler number 
\be
\Omega_\sigma(\gamma) = \chi(\cM_\sigma(\gamma),\nu)
\coloneqq \sum_{m\in\IZ} m \chi(\nu^{-1}(m)),
\label{eqn:gDTdefPrim}
\ee
where $\nu:\cM_\sigma(\gamma)\rightarrow \IZ$ is Behrend's constructible function \cite{behrend2009donaldson}.\footnote{As explained 
e.g. in \cite[\S 2.3]{Bouchard:2016lfg}, the weight $\nu(p)$ can be interpreted  physically as the
dimension of the chiral ring of the superpotential whose critical locus determines the moduli space $\cM_\sigma(\gamma)$.}
In the simplest case when $\cM_\sigma(\gamma)$ is a smooth projective variety (up to the trivial $\IC^\times$ action), $\Omega_\sigma(\gamma)$ is equal to the topological Euler characteristic up to a sign, 
\be
\label{OmbToOm}
\Omega_\sigma(\gamma) = (-1)^{\dim_{\IC}\cM_\sigma(\gamma)}\, \chi(\cM_\sigma(\gamma)).
\ee
For non-primitive charge vectors, one first defines a rational invariant $\bOm_\sigma(\gamma)\in\IQ$ following~\cite{Joyce:2008pc}, and then sets
\be
\Omega_\sigma(\gamma)=\sum_{k|\gamma} \frac{\mu(k)}{k^2}\,  \bOm_\sigma(\gamma/k),
\label{eqn:gDTdef}
\ee
where $\mu(k)$ is the M\"obius function.\footnote{Recall that $\mu(k)=0$ if $k$ has repeated prime factors, otherwise $\mu(k)=(-1)^n$ with $n$ the number of prime factors.\label{fooMoebius}} While $\Omega_\sigma(\gamma)$ is manifestly integral when $\gamma$ is primitive, its integrality for general charge $\gamma$ and $\sigma$ generic remains conjectural. We shall often abuse notation and denote $\Omega_\sigma(\v)=\Omega_\sigma(\gamma)$
where $\v=\gamma/\sqrt{\Td(T\CY)}$ is the Chern character associated to the Mukai vector $\gamma$.

For a compact CY threefold, the space of Bridgeland stability conditions $\Stab(\cC)$ is hard to construct and poorly understood in general. Assuming that it is non-empty 
(as physics strongly suggests), 
one can show  \cite{MR2373143} that it is a complex manifold of dimension $\rk\Gamma$, such that 
the forgetful map $\Stab(\cC)\to \Hom(\Gamma,\IC)$ which sends $\sigma=(Z,\cA)\mapsto Z$ is a local homeomorphism. In other words, the heart $\cA$ is locally determined by the central charge function $Z$. In particular, the complex dimension $\rk\Gamma=b_{\rm even}(\CY)=2b_2(\CY)+2$ is larger than the dimension $b_2(\CY)$ of
K\"ahler moduli space $\cM_K(\CY)$, which is conjecturally embedded 
as a co-dimension $b_2(\CY)+2$ submanifold $\Pi\subset\Stab(\cC)$, as we discuss in \S\ref{sec_pi}.

Moreover, $\Stab(\cC)$ admits an action of $\GLt \times \Aut\cC$  
\cite[Lemma 8.2]{MR2373143},
where $\GLt$ is the universal cover of  the group of $2\times 2$ real matrices  with 
positive determinant and $\Aut\cC$ is the group of autoequivalences of $\cC$. 
The group $GL^+(2,\IR)$ acts on the central charge $Z$ via 
\be
\begin{pmatrix} \Re Z \\ \Im Z \end{pmatrix} \mapsto 
\begin{pmatrix} a & b \\ c & d \end{pmatrix} \begin{pmatrix} \Re Z \\ \Im Z \end{pmatrix},
\qquad ad-bc>0\,,
\ee
preserving the orientation on $\IR^2$, hence the phase ordering of the central charges and hence stability of objects. Its universal cover acts on the stability condition $(Z,\cA)$
by suitably tilting the heart $\cA$.  
By construction, the Donaldson-Thomas invariant $\Omega_\sigma(\gamma)$ is invariant under 
the action of $\GLt$ on $\sigma$,
and under the combined action of 
$\Aut(\cC)$ on $(\gamma,\sigma)$.

Importantly, being integer valued, the generalized DT invariants
$ \Omega_\sigma(\gamma) \in \IZ$ are locally constant on $\Stab\cC$, but they may  jump when some object $E\in \cA$
of charge $\gamma$ goes from being stable to unstable. This may happen when the central charge $Z(\gamma')$ of a subobject $E'\subset E$ of charge $\gamma'$ becomes aligned with  
$Z(\gamma)$, therefore along the
real-codimension one {\it wall of instability} (or marginal stability)
\be
\label{eqWall}
\cW(\gamma,\gamma') \coloneqq \{ \sigma=(Z,\cA) \in\Stab\cC : \Im(Z(\gamma') \overline{Z(\gamma)}) =0\}\,.
\ee
The discontinuity across $\cW(\gamma,\gamma')$ is  determined from the invariants on either side of the wall by the wall-crossing formulae of \cite{ks,Joyce:2008pc}. Physically, the jump in the BPS index is due to the appearance or disappearance of multi-centered black hole bound states \cite{Denef:2007vg}. Of course, this physical
interpretation only holds along the physical slice of $\Pi$-stability conditions.

\subsection{Stability conditions for one-modulus CY threefolds}
\label{sec_stab}

We now restrict again to compact CY threefolds with $b_2(\CY)=1$,
and explain a general construction of an open set of Bridgeland stability conditions around the large volume limit following \cite{bayer2011bridgeland,bayer2016space}. While the full construction is not needed for the rest of the paper, it allows us to introduce, as an intermediate step, a family of weak stability conditions \eqref{Ztilt} (called tilt-stability in \cite{bayer2011bridgeland,bayer2016space}) and a conjectural inequality \eqref{BMTorig}, which will play an essential role in  relating rank 1 and rank 0 DT invariants in \S\ref{sec_rank0fromPT}.

\subsubsection*{Parametrizing central charge functions modulo $\GLt$ action}

As explained in the previous subsection, the space of Bridgeland stability conditions is parametrized locally by the central charges
of the objects \eqref{gammastruc}, or equivalently by the  components $(X^0,X^1,F_1,F_0)\in\IC^4$
of the holomorphic central charge in the Mukai basis,
\be
\label{ZMukai}
Z(\gamma)  = q_0 X^0 + q_1 X^1 - p^1 F_1 - p^0 F_0.
\ee
Using the $\GLt$ action, we may restrict to the real four-dimensional slice with central charge
 \cite[\S 8]{bayer2016space}\footnote{We swap  $(a,b)$ and $(\alpha,\beta)$ 
compared to \cite{bayer2016space}, and rescale the imaginary part by the positive factor $a$.}
parametrized by $(a,b,\alpha,\beta)\in \IR^4$,
\be
\label{ZBMT}
Z_{a,b,\alpha,\beta}(\gamma) =
\left( - \ch_3^b + \beta \, \ch_2^b+ \alpha\, \ch_1^b \right) 
+ \I \left( a \,  \ch_2^b  -  \frac12\, a^3 \ch_0^b  \right),
\ee
where  $\ch^b_k (E)=\int_{\CY} e^{-b H}. H^{3-k}.\ch(E)$, or more explicitly 
\be
\begin{split}
\ch_0^b=C_0,
\qquad &
\ch_1^b = C_1 - b C_0, 
\qquad 
\ch_2^b =C_2-b C_1+\frac12\, b^2 C_0,
\\
& \ch_3^b=C_3-b  C_2+\frac12\, b^2 C_1-\frac16\, b^3 C_0.
\end{split}
\label{relCch}
\ee
This slice is invariant under the spectral flow transformation \eqref{specflowD6} provided
it is accompanied by a translation $b\mapsto b+ k$.
We note that under derived duality $\gamma\mapsto\gamma^\vee$ (see below \eqref{Mukaimap}) accompanied by a sign flip of $(b,\beta)$, the central charge \eqref{ZBMT} transforms into its complex conjugate, 
\be
\label{ZBMSdual}
Z_{a,b,\alpha,\beta}(\gamma^\vee) = 
- \overline{Z_{a,-b,\alpha,-\beta}(\gamma)}\, .
\ee
Upon setting 
\be
\label{LVslice}
a=\sqrt{\frac13\, t^2 - \frac{c_2}{12\kappa}}\, , 
\qquad 
\alpha=\frac12\, t^2 - \frac{c_2}{24\kappa}\, ,
\qquad 
\beta=0\, ,
\ee
the function \eqref{ZBMT} coincides with the large volume central charge\footnote{As discussed below \eqref{FLV}, this formula agrees with the physical central charge in the large volume $t\to\infty$, up to an $\cO(t^0)$ correction proportional to $\zeta(3)\chi_{\CY}$. Agreement up to $\cO(e^{-t})$ can be achieved by replacing $\sqrt{\Td(T\CY)}$ in $\gamma(E)$ \eqref{Mukaimap} by the
$\Gamma$-class \cite{galkin2016gamma}. }
\be
\label{ZLVint}
Z_{b,t}^{\rm LV}(E) =- \int_{\CY} e^{-  (b+\I t) H} \gamma(E)\, ,
\ee
up to rescaling of its imaginary part by $t/a$ using the $\GLt$ action.

In \cite{bayer2011bridgeland,bayer2016space}, a method to construct a heart $\cA_{a,b}$ (depending only on $a$ and $b$) is introduced so that the pair $(Z_{a,b,\alpha,\beta},\cA_{a,b})$ is a Bridgeland stability condition on $\cD^b(\CY)$
whenever the inequalities
\be
\label{alphabound}
a>0,
\qquad 
\alpha> \frac16\, a^2 + \frac12\, a |\beta| 
\ee
are satisfied. The second condition ensures that the central charge \eqref{ZBMT} never vanishes on objects $\cO_\CY(mH)$ with $m\in\IZ$. In particular, the region \eqref{alphabound} includes the large volume slice \eqref{LVslice} for $t^2>\frac{c_2}{4\kappa}$. As we review in the remainder of this subsection, the construction of \cite{bayer2011bridgeland,bayer2016space} proceeds in two steps,
\bea
&
\hspace{1.7cm}
\mbox{slope stability} \hspace{1.4cm} N_{b,a}\mbox{-stability} \hspace{1cm} \mbox{Bridgeland stability} &
\nn\\
& \displaystyle
\({\Coh\CY} \, , \, {-\ch_1^b+\I\ch_0^b }\)
\ {\mathop{\longrightarrow}\limits^{\rm tilt}}\
\({\cA_b}\, ,\, {Z_{b,a}}\)
\  {\mathop{\longrightarrow}\limits^{\rm tilt}}\
\({\cA_{b,a}}\, ,\, {Z_{a,b,\alpha,\beta}}\) &
\nn
\eea
Independently of its use for constructing Bridgeland stability conditions, the 
family of weak stability conditions $N_{b,a}$ appearing in the intermediate step
plays an essential role in relating rank 0 DT invariants to rank 1 DT invariants.

\subsubsection*{Step 1}

We first start with the Abelian category of coherent sheaves $\Coh\CY$ where for any $b \in \mathbb{R}$ we define the slope function  
\be
\label{defslope}
\mu_b(\cE) = \frac{\ch_1^b(\cE)}{\ch_0^b(\cE)}
\ee
for $\ch_0(\cE)\neq 0$, and $\mu_b(\cE)=+\infty$ otherwise. We say a coherent sheaf $\cE$ is slope semi-stable if $\mu_b(\cE')\leq \mu_b(\cE)$ for any subsheaf $\cE'\subset \cE$. We know that any slope semistable sheaf satisfies the classical Bogolomov-Gieseker inequality \cite[Theorem 3.2]{bayer2016space}:
\be
\label{BGineq}
\Delta_H(E) \coloneqq (\ch_1^b(E))^2 - 2 \ch_0^b(E) \ch_2^b(E) = C_1^2-2C_0C_2 \geq 0.
\ee
Following \cite{bayer2011bridgeland},
one defines \begin{itemize}
    \item 
$\cT_{b}\subset\Coh\CY$ as the subcategory generated by slope-semi-stable
sheaves $\cE$ with $\mu_b(\cE)>0$,
\item $\cF_b\subset\Coh\CY$ as the subcategory generated by slope-semi-stable
sheaves $\cE$ with $\mu_b(\cE)\leq 0$. 
\end{itemize}
Then $\cA_b \coloneqq \langle \cF_b[1], \cT_b\rangle$ is the heart of a bounded t-structure on $\cD^b(\CY)$ generated
by length two complexes of the form $E=(F\xrightarrow{d} T)$ with $\ker d \in \cT_b$ and cok$\, d \in \cF_b$. 
For objects in the heart $\cA_b$, we consider the central charge function 
\be
\label{Ztilt}
 Z_{b,a}(\gamma) =-a \ch_2^b   +  \frac12\, a^3 \ch_0^b  + \I \, a^2\ch_1^b \, .
\ee
Note that up to $\widetilde{GL}(2,\IR)$-action, it can be obtained by setting 
$\beta=0$ and  $\alpha=\infty$ in \eqref{ZBMT}, effectively
getting rid of the dependence on $\ch_3$.  The resulting 
pair $(Z_{b,a},\cA_b)$ satisfies the axioms (i,ii,iv,v)
in the previous subsection, but not iii), since the central charge of skyscraper sheaves vanishes. Nonetheless, it defines a family of weak stability conditions in the sense of \cite[Appendix B]{bayer2016space}.

For an object $E \in \cA_b$, we define\footnote{The ratio
\eqref{slopeBMT} agrees with $\sqrt3 N_{bH,tH}(E)$ in  \cite{bayer2011bridgeland,bayer2016space}, upon setting $t=a\sqrt3$.} 
\be
\label{slopeBMT}
N_{b,a}(E)
\coloneqq -\frac{\Re[Z_{b,a}(E)]}{\Im[Z_{b,a}(E)]} 
= \frac{   \ch_2^b(E)  -  \frac12\, a^2  \ch_0^b(E)}
{a\,\ch_1^b(E)  }\,,
\ee
with $N_{b,a}(E)=+\infty$ if $\ch_1^b(E)=0$. Then by definition, $E \in \cA_b$ is semistable with respect to the pair $(Z_{b,a},\cA_b)$ if and only if for any non-trivial subobject $F \subset E$ in $\cA_b$, we have $N_{b,a}(F) \leq N_{b,a}(E)$. By \cite[Theorem 3.5]{bayer2016space}, any such semistable object $E \in \cA_b$ satisfies the classical Bogomolov inequality \eqref{BGineq}. Moreover, it is conjectured in \cite[Conjecture 4.1]{bayer2016space} that it satisfies the following inequality involving the third Chern class $\ch_3(E)$: 
\be
a^2 \left[ (\ch^b_1)^2 -2 \ch_0^b \ch_2^b\right]
 + 4 (\ch^b_2)^2-6 \ch_1^b \ch_3^b 
\geq 0\,,
\label{BMTch}
\ee
which we refer to as the BMT inequality. Moreover \cite[Theorem 4.2]{bayer2016space} shows that the inequality \eqref{BMTch} is equivalent to the original Conjecture 1.3.1 in  \cite{bayer2011bridgeland}, which says that for any object $E \in \cA_b$ which is semistable with respect to the stability function 
$Z_{b,a}$ and satisfies $N_{b,a}(E)=0$, i.e. $\ch_2^b=\frac12\, a^2 \ch_0^b$, 
one has 
\be
\label{BMTorig}
\ch_3^b \leq \frac{a^2}{6}\, \ch_1^b. 
\ee

\subsubsection*{Step 2}

Similar to the construction of $\cA_b$ in the first step, one defines 
\begin{itemize}
\item $\cT_{b,a}\subset\cA_b$ as the subcategory generated by semi-stable objects in $\cA_b$ with $N_{b,a}(E)>0$, 
\item $\cF_{b,a}\subset\cA_b$ as the subcategory generated by semi-stable objects in $\cA_b$ with $N_{b,a}(E)\leq 0$. 
\end{itemize}
Then we define $\cA_{b,a}=\langle \cF_{b,a}[1], \cT_{b,a}\rangle$. By construction, $\Im Z_{a,b,\alpha,\beta}(E)\geq 0$ for any object $E\in \cA_{b,a}$. The conjectural inequality \eqref{BMTorig} further guarantees that $\Re Z_{a,b,\alpha,\beta}(E)<0$ whenever
$\Im Z_{a,b,\alpha,\beta}(E)=0$
\cite{bayer2016space}, which shows that the axioms of \S\ref{sec_DTdef} are indeed satisfied. This was in fact the original motivation for the conjectural BMT inequality.

\subsubsection*{Wall-crossing in the space of weak stability condition} 
To obtain the formula relating rank zero DT invariants to rank one DT invariants in Appendix \ref{sec_appS}, we 
shall apply the wall-crossing formula in the space of weak stability conditions $(Z_{b,a}, \cA_b)$, rather than in the space of Bridgeland stability conditions $\Stab(\cC)$, as walls are much easier to control. 

It will be convenient to rescale and shift the slope function $N_{b,a}$ \eqref{slopeBMT} into
\begin{equation}\label{scale}
\nu_{b,w}\coloneqq \ a N_{b,a}+b, \quad\text{where } w \coloneqq \frac12(a^2+b^2),
\end{equation}
for $w>b^2/2$. This is because the new slope
\begin{equation}\label{noo}
\nu_{b,w}(E)\ =\ \left\{\!\!\begin{array}{cc} 
\frac{C_2(E) - wC_0(E)}{C_1(E)-b C_0(E)}
 & \text{if }\ch_1^{b}(E)\ne 0, \\
+\infty & \text{if }\ch_1^{b}(E)=0 \end{array}\right.
\end{equation}
has a denominator that is linear in $b$ and numerator linear in $w$, so the walls of $\nu_{b,w}$-instability (which is by construction equivalent to $N_{b,a}$-instability) are line segments in the region 
\be
U \coloneqq  \,\left\{(b,w) \in \mathbb{R}^2 \colon w > \tfrac12b^2  \right\}
\ee
of the $(b,w)$ plane (see the green line in Fig. \ref{fig-curves}). We shall abuse notation and denote by $\bOm_{b,w}(\gamma)$ the rational DT invariant 
$\bOm_{\nu_{b,w}}(\gamma)$ counting $\nu_{b,w}$-stable objects of class $\gamma$. 

More precisely, the slope \eqref{noo} coincides for two objects $E$ and $E'$ of Chern character $C_i$ and $C'_i$ along the line
\be
(C_1 C'_0 - C'_1C_0) w
+ b ( C_0 C'_2 - C'_0 C_2) + ( C_2 C'_1 - C_1 C'_2)=0\,,
\ee
passing through the points $\varpi(\gamma)$ and $\varpi(\gamma')$ defined by 
\be
\label{defPi}
\varpi(\gamma) = \left( \frac{C_1}{C_0}, \frac{C_2}{C_0} \right).
\ee
Note that the points $\varpi(\gamma)$ lie outside the region $U$ when $E$ and $E'$ are 
$\nu_{b,w}$-semistable objects, due to the Bogomolov-Gieseker inequality \eqref{BGineq}.
In the original coordinates $(b,a)$, walls of $N_{b,a}$-instability are half-circles centered at $b=\frac{C_0 C'_2 - C'_0 C_2}{C_0 C'_1 - C'_0C_1}$ along the axis $a=0$, or vertical lines going through $b=\frac{C_1 C'_2 - C'_1 C_2}{C_0 C'_2 - C'_0C_2}$ when $C_1 C'_0 - C'_1C_0=0$.

\begin{figure}[t]
\begin{center}
\includegraphics[height=7cm]{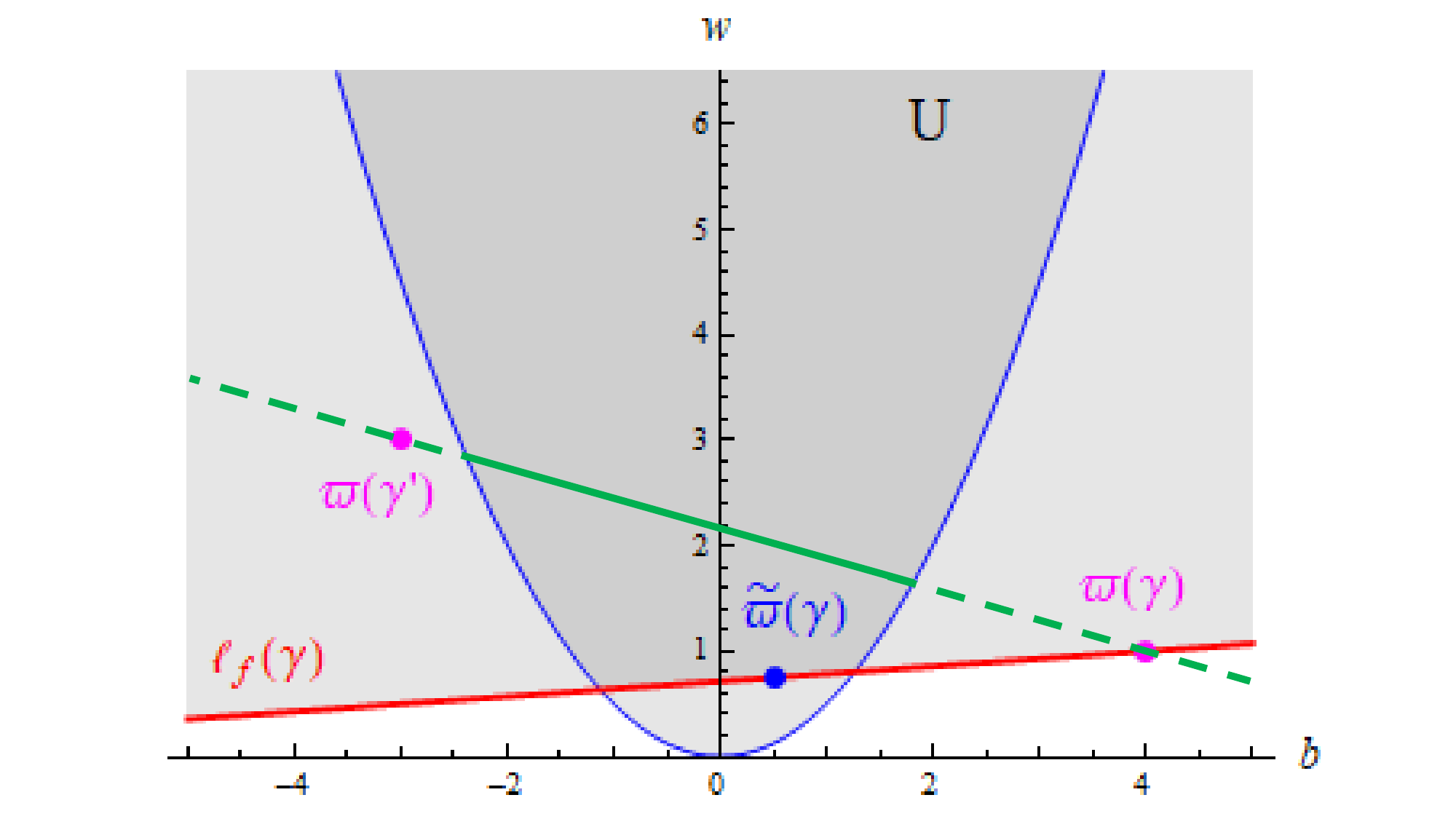}
\end{center}
\vspace{-0.5cm}
\caption{In the $(b,w)$-plane, walls of $\nu_{b,w}$-instability are straight lines 
between $\varpi(\gamma)$ and $\varpi(\gamma')$, where $\gamma$ and $\gamma'$ are the charges of the destabilizing objects. The BMT inequality \eqref{BMTineq} is saturated along the red line going through $\varpi(\gamma)$ and $\widetilde\varpi(\gamma)$. The parabola $w=\hf b^2$ is shown in blue.}
\label{fig-curves}
\end{figure}

\subsubsection*{Wall and chamber structure}
For any fixed class $\gamma$ with $C_0 \neq 0$, or $C_0=0$ and $C_1 \neq 0$, there exists a set of lines $\{\ell_i\}_{i \in I}$ in $\mathbb{R}^2$ \cite[Proposition 4.1]{feyz-noether-lef} such that the segments $\ell_i\cap U$ (called `walls') are locally finite and satisfy
\begin{enumerate}
    \item If $C_0 \neq 0$, then all lines $\ell_i$ pass through $\varpi(\gamma)$, and if $C_0=0$ then all lines $\ell_i$ are parallel of slope $\frac{C_2}{C_1}$.
    \item The $\nu_{b,w}$-semistability of any object $E \in \cC$ of class $\gamma$ is unchanged as $(b,w)$ varies within any connected component (called a ``\emph{chamber}") of $U \setminus \bigcup_{i \in I}\ell_i$. 
    \item For any wall $\ell_i\cap U$, there is an object $E \in \cC$ of class $\gamma$ which is strictly $\nu_{b,w}$-semistable for all $(b,w) \in \ell_i\cap U$. 
\end{enumerate}
The DT invariant $\bOm_{b,w^+}(\gamma)$  at a point just above $\ell_i$ is determined from the invariant $\bOm_{b,w^{-}}(\gamma)$ at a point
just below $\ell_i$ by the wall-crossing formula of \cite{Joyce:2008pc}.
Note that with this definition, the DT invariant $\bOm_{b,w}(\gamma)$ is not necessarily discontinuous across the wall.

\subsubsection*{Tilt-stability and Gieseker stability}

Since the number of walls for fixed charge $\gamma$ which are crossed as $w\to+\infty$ is finite \cite[Proposition 1.4]{Feyzbakhsh:2021nds}, the index $\bOm_{b,w}(\gamma)$ reaches a fixed value as $w\to +\infty$. For $p^0=0$,
there is no vertical wall, so this value is independent of $b$, and 
we denote it by $\bOm_{\infty}(\gamma)$.  For $p^0\neq 0$, the index
may jump across the vertical wall at $b=\frac{C_1}{C_0}$ given by the vanishing of the slope \eqref{defslope}. We denote by $\bOm_{\infty}(\gamma)$ the
limit of the index $\bOm_{b,w}(\gamma)$ as $w\to +\infty$ on the side $b<\frac{C_1}{C_0}$ for positive rank $p^0> 0$, or on the side $b>\frac{C_1}{C_0}$ for negative rank $p^0<0$. 

For non-negative rank $p^0\geq 0$ and  $\gamma$ primitive, it turns out that $\bOm_{\infty}(\gamma)$ agrees with the weighted Euler number 
$\chi(\cM_{\rm tilt}(\gamma),\nu)$ 
of the moduli space $\cM_{\rm tilt}(\gamma)$ of tilt-semi-stable sheaves 
of charge $\gamma$ \cite[Lemma 2.4]{Feyzbakhsh:2022ydn}. Here, tilt-stability is a variant of
Gieseker semi-stability defined as follows: let $P_E( k)$ be the Hilbert polynomial 
\be
\label{HilbP}
\begin{split}
P_E( k) \coloneqq &\, \chi(\cO(- k H), E) = \int_{\CY} e^{ k H} \ch E \Td(T\CY)
\\
=&\,  \frac{\kappa p^0}{6}  k^3 + \frac{\kappa p^1}{2}  k^2  - 
\left(q_1 + \frac{c_2}{24} p^0 \right)  k + 
\left(q_0 - \frac{c_2}{24} p^1 \right),
\end{split}
\ee
and $p_E( k)=P_E( k)/a_E$ the associated monic 
Hilbert polynomial, 
with $a_E$ the coefficient of the highest degree term in $ k$.
Gieseker-(semi)stability for a coherent sheaf $E$ is the requirement that for
all exact sequences $0\to E'\to E\to E''\to0$ of coherent sheaves, we have $\deg p_{E'}>\deg p_{E''}$,
or $\deg p_{E'}=\deg p_{E''}$ and $p_{E'}( k)<(\leq) p_{E''}( k)$ for $ k\gg 1$.
We denote by $\bOmH(\gamma)$ the rational index counting Gieseker-semistable sheaves with class $\gamma$, defined as in \cite{Joyce:2008pc}. 
Tilt-stability is defined in the same way, but discarding the constant term of the Hilbert polynomial before dividing by its top coefficient as before.
However, for threefolds with $\Pic\CY=H\IZ$ and two-dimensional class (i.e. $p^0=0$, $p^1\neq 0$),
the index $\bOm_{\infty}(\gamma)$ counting tilt-semistable objects coincides
with the index $\bOmH(\gamma)$ counting Gieseker-semistable sheaves~\cite[Lemma 5.2]{Feyzbakhsh:2022ydn}. In \S\ref{sec_rank0fromPT}, we shall present explicit formulae 
relating $\bOmH(\gamma)$ for rank 0 charges (counting D4-D2-D0 bound states) 
and rank $\pm 1$ charges (counting D6-D2-D0 bound states), which follow
by a sequence of wall-crossings from an empty chamber provided by
the conjectural BMT inequality \eqref{BMTch}.

\subsubsection*{Conjectural BMT inequality.} 

\begin{figure}[t]
\begin{center}
\includegraphics[height=7cm]{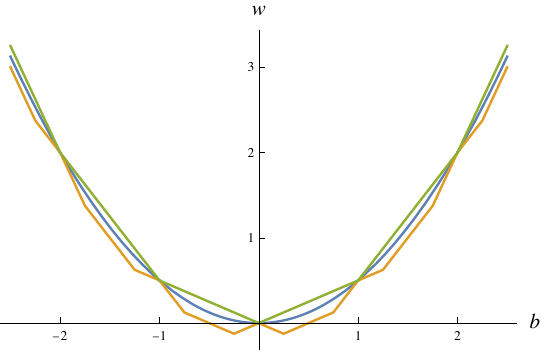}
\end{center}
\vspace{-0.4cm}
\caption{For the quintic threefold $X_5$, the stronger Bogolomov-Gieseker inequality established in \cite[Theorem 1.1]{li2019stability} implies that for any $\nu_{b,w}$-semistable object $E$, the point $\varpi(E)$ lies below the orange curve $w=G(b)$, where
$G(b)=-\frac12|b|$ for $|b|<\frac14$, $G(b)=\frac12|b|-\frac14$ for $\frac14<|b|<\frac34$, $G(b)=\frac32|b|-1$ for $\frac34<|b|<1$ and
$G(b)=G(b-\lfloor b\rfloor)+ \lfloor b\rfloor b - \frac12 \lfloor b\rfloor^2$ when
$b$ lies outside the interval $(-1,1)$. Moreover, the BMT inequality is  known to hold in the region above the green curve given in \eqref{liregion} \cite[Theorem 1.2]{li2019stability}. These two curves intersect the blue curve $w=\frac12 b^2$ for integer values of~$b$.}
\label{fig-Li}
\end{figure}

In the plane parametrized by $(b,w=\frac12(a^2+b^2))$, the BMT inequality \eqref{BMTch}  implies the linear inequality
\begin{equation}\label{BMTineq}
    L_{b,w}(\gamma) \coloneqq (C_1^2-2C_0C_2)w+(3C_0C_3-C_1C_2)b+(2C_2^2-3C_1C_3) \geq 0\,,
\end{equation}
whenever there exists a $\nu_{b,w}$-semistable object $E \in \cD^b(\CY)$ of class $\gamma$. From \eqref{BGineq}, the coefficient of $w$ in the above equation is $\Delta_H(E) \geq 0$. If $\Delta_H(E) > 0$, the inequality \eqref{BMTineq} says that $E$ can be $\nu_{b,w}$-semistable only for points $(b,w) \in U$ above the line $\ell_f(\gamma)$ defined by the equation $L_{b,w}(\gamma)=0$ (see the red line in Fig. \ref{fig-curves}). This line passes through the points $\varpi(\gamma)$ defined in \eqref{defPi} and 
\be
\label{defPit}
\widetilde\varpi(\gamma)=\left( \frac{2C_2}{C_1},\frac{3C_3}{C_1}\right).
\ee 

The conjectural BMT inequality \eqref{BMTineq} has now been proved for the quintic threefold $X_5$ and for a degree $(4,2)$ complete intersection $X_{4,2}$ in $\PP^5$ when $(b,w)$ satisfy  \cite{li2019stability,liu2021stability} 
\be
\label{liregion}
w -\frac12\, b^2 > \frac12\, [b] (1-[b])\, ,
\qquad
[b]\coloneqq b-\lfloor b \rfloor\, .
\ee
Moreover, a slightly weaker version of \eqref{BMTineq} is proved for the sextic and octic CY threefolds, $X_6$ and $X_8$, in the same restricted region \eqref{liregion} \cite{koseki2020stability}. The proofs of the BMT inequality for these models rely on a strengthening of the classical Bogolomov-Gieseker inequality \eqref{BGineq}, i.e. the existence of a function $G \colon \mathbb{R} \rightarrow \mathbb{R}$ such that any slope-semistable sheaf $E$ satisfies $\frac{\ch_2(E).H}{\ch_0(E)H^3} \leq G\left(\frac{\ch_1(E).H^2}{\ch_0(E)H^3}\right)$ and $G(b) \leq \frac{b^2}{2}$ for all $b \in \mathbb{R}$. When such a function is available, one can enlarge the space of weak stability conditions $U$ to $U_G \coloneqq \{ (b,w) \in \mathbb{R}^2 \colon w > G(b)\}$, see Figure \ref{fig-Li} for the quintic threefold. 
The existence of such a function and the status of the BMT inequality
for the other hypergeometric models in Table \ref{table1} remains open at the time of writing.

\subsection{K\"ahler moduli and $\Pi$-stability}
\label{sec_pi}

While the DT invariants $\Omega_\sigma(\gamma)$ are mathematically well-defined throughout the space of Bridgeland stability conditions $\Stab\cC$ (away from walls of marginal stability), they only
acquire physical meaning along a particular complex one-dimensional slice $\Pi\subset\Stab\cC$
where the central charge $Z(\gamma)$ coincides with the physical central charge $Z_{\rmz}(\gamma)$ determined by the complexified K\"ahler structure on $\CY$, or equivalently by the complex
structure parametrized by $z$ of the mirror family $\hat\CY$. 
On the mirror side, the central charge is given by the period integral 
\be
Z_z(\gamma) = \Pi_z(\hat\gamma) = \int_{\hat\gamma} \Omega_{3,0}\,,
\ee
of the holomorphic 3-form on the cycle $\hat\gamma\in H_3(\hat\CY,\IZ)$ dual to $\gamma$. 

We shall restrict to CY threefolds
obtained as a smooth complete intersection of degree $(d_1,\dots,d_n)$ in weighted projective space $\IP^{n+3}_{w_1,\dots,w_{n+4}}$. 
There are 13 such threefolds $\CY$, whose basic topological data are tabulated in Table \ref{table1}. 
In particular, we note that $\sum_j d_j=\sum_i w_i$ by the CY condition, and $\kappa=H^3=\prod_i d_i / \prod_j w_j$. The mirror threefold $\hat\CY$ can be obtained, for example, by applying the general construction
of \cite{Batyrev:1994pg}. For all these models, 
the periods  satisfy a Picard-Fuchs equation of hypergeometric type, 
\be
\cL \, \Pi_z(\hat\gamma) = \left[ (z\partial_z)^4 - \mu^{-1} z  \prod_{k=1}^4 (z\partial_z+a_k) \right] \Pi_z(\hat\gamma)=0\, ,
\ee
where $\mu=\prod w_i^{w_i}/\prod d_j^{d_j}$, the `local exponents' $a_k$ satisfy $\sum_k a_k=2$ and are ordered in increasing order for definiteness. 
The equation has singularities at $z=0$, $z=\mu$ and at $z=\infty$, such that the 
K\"ahler moduli space of $\CY$ (or complex structure moduli space of $\hat\CY$) consists of the punctured sphere $\cM_K(\CY)=\IP\backslash\{0,\mu,\infty\}$. 
The two singularities at $z=0$ and $z=\mu$ are universal, and correspond to the 
large volume limit and conifold point, respectively. Following \cite{Joshi:2019nzi}, we denote these two types of degeneration by $M$ (for maximal unipotent monodromy) and $C$ (for conifold). The type of degeneration at $z=\infty$ depends on the local exponents $a_k$, and may be of type $F$ (when all $a_k$ are distinct, corresponding to a monodromy of finite order), $C$ (when $a_2=a_3$), $K$ (when $a_1=a_2$ and $a_3=a_4$), or $M$ (when all $a_k$'s coincide). Degenerations of type $K$ and $M$ occur at infinite distance with respect to the special K\"ahler metric on $\cM_K$, while degenerations of type $F$ and $C$ occur at finite distance.
Under a type $C$ degeneration, the conformal field theory on $\CY$ becomes singular, due to a brane becoming massless, while a type $F$ degeneration leads to a regular 
CFT, often with a Gepner-model type description.
The regulator $\rho$, that will be relevant for the direct integration of the holomorphic anomaly equations discussed in Section~\ref{sec_direct}, is defined to be the smallest denominator among the local exponents at $z=\infty$.
The exponents and the type of the singularity at $z=\infty$ are indicated in Table~\ref{table1}.

To construct a basis of solutions adapted to the maximal unipotent monodromy at $z=0$ (corresponding to the large volume limit on the mirror),
we apply the Frobenius method. For $\eps\in\IC$ let
\be
\Pi(z,\eps) = \sum_{k=0}^{\infty}
\frac{\prod_{j=1}^n \Gamma( d_j(k+\eps)+1)}
{\prod_{i=1}^{n+4} \Gamma( w_i(k+\eps)+1)} \, z^{k+\eps}.
\ee
Using the identity 
\be
\frac{\prod_{j=1}^n \left( d_j \prod_{\ell=1}^{d_j-1} (d_j k+\ell) \right) }
{\prod_{i=1}^{n+4} \left( w_i \prod_{\ell=1}^{w_i-1} (w_i k+\ell) \right) }
= \mu^{-1} (k+a_1) (k+a_2) (k+a_3) (k+a_4) ,
\ee
one easily checks that 
\be
\cL \Pi(z,\eps) = 
\eps^4 z^\eps\,  \frac{\prod_j \Gamma( d_j \eps+1)}
{\prod_i \Gamma( w_i \eps+1)}\, .
\ee
Thus, the first three terms $\Pi_{0\leq p\leq 3}(z)$ in the Taylor expansion around $\eps=0$
\be
\Pi(z,\eps) =\sum_{p=0}^{\infty} \Pi_p(z) (2\pi\I\eps)^p\,,
\ee
are annihilated by $\cL$. In the Mukai basis \eqref{ZMukai}, the coefficients  $(X^0,X^1,F_1,F_0)$
are given by  
\be
\begin{split}
X^0=\Pi_0(z), & 
\qquad 
F_0 = \kappa \Pi_3(z) + \frac{c_2}{24}\, \Pi_1(z),
\\
X^1=\Pi_1(z), &
\qquad 
F_1=-\kappa \Pi_2(z) - \frac{c_2}{24}\, \Pi_0(z).
\end{split}
\ee
We define the flat coordinate $\rmz=b+\I t = X^1/X^0$, such that $\rmz\mapsto \rmz+1$ under
monodromy $z\mapsto e^{2\pi \I} z$ around $z=0$. The components can be integrated to a prepotential 
$F(\rmz)$ such that 
\be
F_1/X^0 = \partial_{\rm z}F(\rmz)\, ,
\qquad 
F_0/X^0 = 2 F(\rmz) - \rmz \partial_{\rmz} F(\rmz)\, .
\ee
In the large volume limit $t\to\infty$, the prepotential has an asymptotic expansion
\be
\label{FLV}
F(\rmz) = -\frac{\kappa}{6}\, \rmz^3 +  \frac{\zeta(3) \chi_{\CY}}{2(2\pi\I)^3} -\frac{1}{(2\pi\I)^3} 
\sum_{Q=1}^{\infty} \GVg{0} \Li_3\left( e^{2\pi\I Q \rmz}\right),
\ee
where $\GVg{0}$ are the genus-zero Gopakumar-Vafa invariants. Keeping only the
leading cubic term in \eqref{FLV} and fixing the K\"ahler gauge $X^0=-1$
in \eqref{ZMukai}, one arrives at 
\be
\begin{split}
Z^{LV}_{b,t}(\gamma) =&\,   \frac{\kappa}{6}\, \rmz^3 p^0 - \frac{\kappa}{2}\, p^1 \rmz^2 - q_1 \rmz - q_0
\\
=&\, \left( - \ch_3^b + \left(\frac12\, t^2-\frac{c_2}{24\kappa}\right)   \ch_1^b \right) + \I t \left(    \ch_2^b - \left(\frac16\, t^2-\frac{c_2}{24\kappa}\right) \ch_0^b   \right).
\end{split}
\label{ZLV}
\ee
which reproduces \eqref{ZLVint} and coincides with the standard slice \eqref{ZBMT} upon making the identifications in \eqref{LVslice}.

Taking into account subleading corrections, it is necessary to apply a $\GLt$ 
transformation in
order to reach the form  \eqref{ZBMT}. The resulting values of $a,b,\alpha,\beta$ can be computed
by equating the products $x_{ij}=\Im(Z_i \overline {Z_j})$, $0\leq i<j \leq 3$ where  $Z_i$
is the central charge for the Chern vector defined by $C_k=\delta^k_i$.
Indeed, these quantities are invariant up to
scale under  $\widetilde{GL(2,\IR)^+}$ and satisfy the quadratic constraint
$x_{01} x_{23} + x_{02} x_{31} + x_{03} x_{12}=0$, so give the desired 4 local real coordinates.
In this way, one finds
\bea
\begin{split}
a(z)=&\,\frac{\sqrt{(\Im \tF_1)^2 - 2 \kappa \Im\rmz \Im \tF_0}}{\kappa \sqrt{\Im\rmz}}\, ,
\qquad
b(z) = -\frac{\Im \tF_1}{\kappa \Im\rmz}\, ,
\\
\alpha(z)=&\, -\frac{ \Im(\brmz\tF_1)}{\kappa \Im\rmz}
-\frac{(\Im \tF_1)^2}{2 \kappa ^2  (\Im\rmz)^2}\, ,
 \\
\beta(z) =&\, \frac{6   \kappa ^2   (\Im\rmz)^2
\Im(\brmz \tF_{0}) 
+6 \kappa{\Im\rmz}  \Im \tF_{1}  \Im (\tF_{0}-\brmz \tF_{1})
-4 (\Im \tF_{1})^3} 
{3  \kappa  {\Im\rmz} \left(2  \kappa {\Im\rmz} \Im \tF_{0}-(\Im \tF_{1})^2\right)} \, ,
\end{split}
\eea
where we denoted\footnote{Note that the signs are such that shift by $\frac{c_2}{24}$ cannot be absorbed into a linear shift of $F$!} $\tF_0=2F-\rmz \partial_{\rmz} F + \frac{c_2}{24}\rmz$, $\tF_1=\partial_{\rmz}F - \frac{c_2}{24}$. 
In the region where $\alpha> \frac16 a^2 + \frac12 a |\beta|$, the heart $\cA(z)=\cA_{b(z),a(z)}$ is given by the double-tilt construction explained in the previous subsection. 
The construction of the heart on the full physical slice $\Pi$, including the vicinity of the singularities at $z=\mu$ and $z=\infty$, remains a challenging open problem. Assuming that this problem has been solved, we denote by $\bOmPi_{\rmz}(\gamma)$ the generalized DT invariant along the physical slice.
Fortunately, the relation between rank 0
and higher rank DT invariants at large volume can be derived using only the family of weak stability  conditions $\nu_{b,w}$, assuming that the BMT inequality holds.

\subsection{Interlude -- summary}
\label{subsec-summary}

Let us briefly summarize the previous subsections. First, we introduced the space of Bridgeland stability conditions $\Stab\cC$ on the derived category $\cC$ of coherent sheaves, which is the appropriate mathematical framework for BPS branes in type IIA string theory. A stability condition is a pair $(Z,\cA)$ of a central charge function $Z$ and a heart $\cA\subset \cC$ which determines which constituents may bind into stable objects. For one-modulus CY threefolds, after dividing out by the $\GLt$ action (which preserves the phase ordering of central charges), the space of stability conditions effectively has  real dimension 4. Assuming the BMT conjecture \eqref{BMTorig}, we outlined the construction of an open set $\cU$ in $\Stab\cC$ 
parametrized by $(a,\alpha,b,\beta)\in\IR^4$ subject to the inequalities \eqref{alphabound}.
The physical subspace of $\Pi$-stability conditions provides a real-codimension two slice $\Pi$ in this open set, determined by the prepotential $F(\rm z)$  (see Fig. \ref{fig-stab}). For $\rmz\to\I\infty$, this slice asymptotically coincides with the large volume slice \eqref{ZLV} or equivalently \eqref{ZLVint}. On the boundary of $\cU$, there is also an important family of weak stability conditions with central charge \eqref{Ztilt} parametrized by $(b,a)\in \IR\times \IR^+$, which is obtained from \eqref{ZLV} by omitting the contributions 
proportional to the D0-brane charge $\ch_3^b$ and to the second Chern class $c_2(T\CY)$, and setting $t=a\sqrt3$. This family, interchangeably called $N_{b,a}$-stability or $\nu_{b,w}$-stability with $w=\frac12(a^2+b^2)$, serves as a key intermediate step for the construction of the
heart $\cA$ throughout the open set $\cU$, and is the subject of the BMT conjecture \eqref{BMTorig}, which constrains the existence of $\nu_{b,w}$-semi-stable objects for small $w$. We denote  by $\bOm_{b,w}(\gamma)$ the rational Donaldson-Thomas invariant counting $\nu_{b,w}$-semistable objects of class $\gamma$ in the heart $\cA_b$,
and by 
$\bOmH(\gamma)$ the rational Donaldson-Thomas invariant counting $H$-Gieseker-semi-stable
sheaves of class $\gamma$ defined following \cite{Joyce:2008pc}. 
In the next two subsections, we spell out the modularity properties predicted by string theory for  DT invariants $\bOmH(\gamma)$ in the rank 0 case,
and the relation with ordinary DT invariants and PT invariants in the rank 
$\pm 1$ case.

\begin{figure}[t]
\begin{center}
\includegraphics[height=8cm]{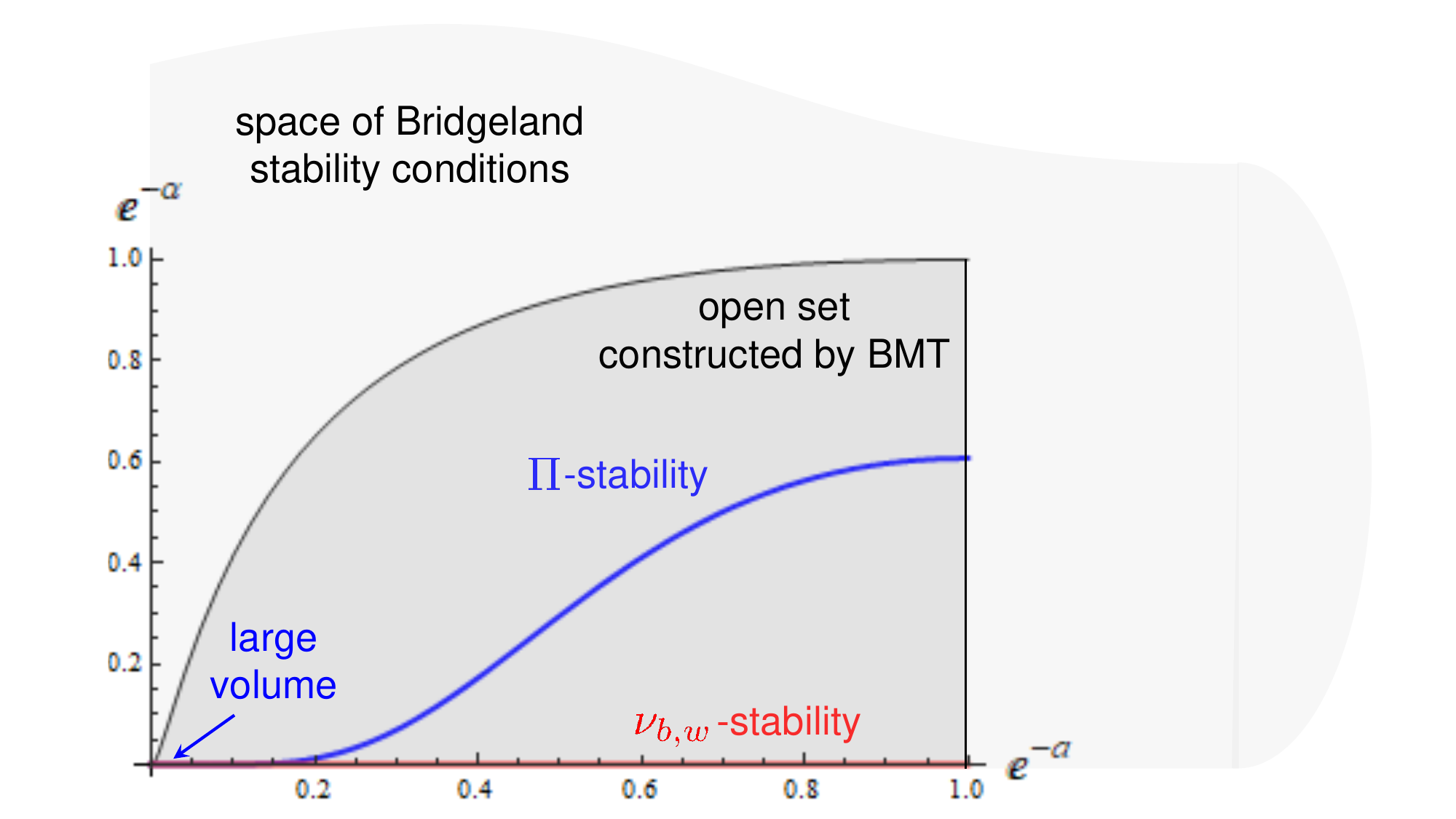}
\end{center}
\vspace{-0.5cm}
\caption{Section of the space of Bridgeland stability conditions by the plane $\beta=0$, $b=\mbox{const}$, 
drawn in coordinates $(e^{-a},e^{-\alpha})$. The boundary of the set 
constructed in \cite{bayer2011bridgeland,bayer2016space} corresponds to the inequalities \eqref{alphabound}. 
The red line is the slice of weak $\nu_{b,w}$-stability conditions with the central charge \eqref{Ztilt} and 
the blue line represents the slice of $\Pi$-stability conditions parametrized  by 
the complexified K\"ahler structure of $\CY$. The large volume limit corresponds to the region near the origin 
where the two slices approach each other.}
\label{fig-stab}
\end{figure}

\subsection{D4-D2-D0 indices and modularity conjecture}
\label{sec_modconj}

Here we consider the case $p^0=0$, $p^1:=r>0$, corresponding to D4-D2-D0 bound states.
As explained below \eqref{HilbP}, for a fixed charge $\gamma$,  the index $\bOm_{b,w}(\gamma)$ reaches a finite value $\bOm_{\infty}(\gamma)$ as $w\to +\infty$, which turns out to coincide with the index 
$\bOmH(\gamma)$ counting Gieseker-semi-stable sheaves.
For CY threefolds with Picard rank one, this index also agrees with the `large volume attractor index' (also called MSW index in \cite{Alexandrov:2012au,Alexandrov:2016tnf,Alexandrov:2017qhn,Alexandrov:2018lgp})\footnote{In general, the MSW index is defined as the value of $\bOmPi_\rmz(\gamma)$
in the asymptotic direction $\rmz^a=-\kappa^{ab} q_b+\I \lambda p^a$ with 
$\lambda\gg 1$, where $\kappa^{ab}$ is the inverse of the matrix $\kappa_{ab}=\kappa_{abc}p^c$. When $b_2(\CY)=1$, the distinction between 
Gieseker index and MSW index becomes irrelevant.}
\be
\label{lvolatt}
\bOm_{\infty}(0,r,q_1,q_0) = \lim_{\lambda\to +\infty}
\bOmPi_{-\frac{q_1}{\kappa r}+\I \lambda  r}(0,r,q_1,q_0)\,,
\ee
where $\bOmPi_{\rmz=b+\I t}(\gamma)$ denotes the DT invariant along the $\Pi$-stability slice. The index $\Omega_{\infty}(0,r,q_1,q_0)$ 
is preserved 
under spectral flow \eqref{specflowD6} with $k\in\IZ$, which leaves 
the D4-brane charge $r$ invariant, as well
as the reduced D0-brane charge
\be
\label{defqhat}
\hq_0 :=
q_0 -\frac{q_1^2}{2\kappa r}\,,
\ee
and the class of $\mu:=q_1 - \frac12\, \kappa r^2$ in $\Lambda^*/\Lambda=\IZ/(\kappa r \IZ)$. Accordingly, we denote 
\be
\label{defOmrmu}
\bOm_{r,\mu}(\hq_0)=\bOm_\infty(\gamma)=\bOmH(\gamma)\, .
\ee 
Note that for fixed $r\geq 1,\mu\in\IZ$, the argument $\hq_0$ is such that the combination
\be
\label{nfromhq0}
n = \frac{\chi(\cD_r)}{24} -\frac{\mu^2}{2\kappa r} - \frac{r \mu}{2} - \hq_0\,,
\ee
is an integer. Here $\chi(\cD_r)$ is the topological Euler characteristic of the divisor $\cD_r$ Poincar\'e dual to $rH$, \cite[(3.8)]{Maldacena:1997de}\footnote{We urge the reader not to confuse $\chi(\cD_r)$ with the holomorphic Euler characteristic 
$\chi_{\cD_r}$ defined in \eqref{defL0}.} 
\be
\chi(\cD_r) = \kappa r^3+c_{2}r.
\label{defchiD}
\ee

It follows from
derived duality $E\mapsto E^\vee$, which acts on the Chern vector by $\ch(E)\mapsto \ch(E)^\vee$ as defined below \eqref{dsz}, that the index 
\eqref{defOmrmu} is invariant under 
$\mu\mapsto -\mu$, on top of the periodicity $\mu\mapsto\mu+\kappa r$ (note however that the integer \eqref{nfromhq0} is not invariant under these symmetries). 
Furthermore, 
$\Omega_{r,\mu}(\hq_0)$ vanishes unless  the reduced charge $\hq_0$ 
is bounded from above by \cite[Corollary 3.3]{Toda:2011aa}
\be
\hq_0 \leq \hq_0^{\rm max}=\frac{1}{24}\, \chi(\cD_r)\, .
\label{qmax}
\ee
Upon identifying $\hat q_0=\frac{c_L}{24}-L_0$ with $c_L=\chi(\cD_r)$,
this coincides with the unitarity bound $L_0\geq 0$ in the two-dimensional $(0,4)$ superconformal field theory obtained by reducing the worldvolume theory of an M5-brane wrapped
on $\cD_r$ \cite{Maldacena:1997de}.

Since the reduced D0-brane charge is bounded from above for fixed D4-brane charge $r>0$ and D2-brane charge $q_1=\mu+\frac12\, \kappa r^2$, one 
can define the generating series of rational invariants
\be
h_{r,\mu}(\tau) =\sum_{\hq_0 \leq \hq_0^{\rm max}}
\bOm_{r,\mu}(\hq_0)\,\q^{-\hq_0 }\, , 
\qquad 
\q=e^{2\pi\I \tau},
\label{defhDT}
\ee
Since $\mu$ takes values in  $\IZ/(\kappa r \IZ)$, \eqref{defhDT}
defines a vector with $\kappa r$ entries (half of which being redundant due to the symmetry under $\mu\mapsto -\mu$). For $r=1$, the case of interest in this paper, the charge vector is primitive and therefore the rational DT invariant $\bOm_{1,\mu}(\hq_0)$ coincides with the integer DT invariant
$\Omega_{1,\mu}(\hq_0)$, defined by replacing $\bOm_\infty(\gamma)$
by $\Omega_\infty(\gamma)$ in \eqref{defOmrmu}.

By exploiting the constraints of S-duality in string theory, it has been argued that the generating series $h_{r,\mu}$  must possess specific modular properties 
under $SL(2,\IZ)$ transformations of the parameter $\tau$
\cite{Alexandrov:2012au,Alexandrov:2016tnf,Alexandrov:2018lgp,Alexandrov:2019rth}.
More precisely, $h_{r,\mu}$ should transform  as a weakly holomorphic vector valued mock modular form 
of depth $r-1$, with a specific modular anomaly.
In this paper, we restrict ourselves to the simplest Abelian case $r=1$, and
refer to the invariants $\Omega_{1,\mu}(\hq_0)$ as  Abelian D4-D2-D0 indices. In this situation, the modular anomaly disappears and $h_{1,\mu}$ must transform as a vector-valued modular form 
of weight $-3/2$ in the Weil representation attached to the lattice 
$\IZ[\kappa]$ with quadratic form 
$m\mapsto \kappa m^2$. Equivalently, it
must transform with the following matrices under $T:\tau\mapsto \tau+1$ and $S:\tau\mapsto -1/\tau$
\cite[Eq.(2.10)]{Alexandrov:2019rth} (see also \cite{Gaiotto:2006wm, deBoer:2006vg, Denef:2007vg, Manschot:2007ha})
\be
\begin{split}
M_{\mu \nu}(T)=&\,
e^{\frac{\pi\I}{\kappa}(\mu+\frac{\kappa}{2} )^2 +\frac{\pi\I}{12}\, c_{2} }
\,\delta_{\mu\nu},
\\
M_{\mu\nu}(S)=&\,
\frac{(-1)^{\chiOD}}{\sqrt{\I \kappa}}\,
e^{-\frac{2\pi\I}{\kappa}\,\mu \nu},
\end{split}
\label{multsys-h2}
\ee
where $\chiOD=\frac{\kappa}{6}+\frac{c_2}{12}$ (see below \eqref{gammastruc}).
We denote by $\scM_1(\CY)$ the space of weakly holomorphic vector-valued modular forms with these transformation properties under $SL(2,\IZ)$.

It is well known that any  weakly holomorphic vector-valued modular form of negative weight $w<0$ is completely determined 
only by its 'polar coefficients', i.e. the terms in its Fourier expansion that become singular in the limit $\tau\to\I\infty$. Such terms correspond to the terms with $\hq_0>0$ in \eqref{defhDT}. Once the polar terms are known, the full modular form can then be determined, for example by constructing the Poincar\'e-Rademacher sum (see e.g. \cite{Cheng:2012qc}).
It is important however, that the dimension of the space of modular forms can be strictly smaller
than the number of polar terms, which means that the polar coefficients must satisfy certain linear constraints,
which are related to the existence of cusp forms in dual weight $2-w$ \cite{Bantay:2007zz,Manschot:2007ha,Manschot:2008zb}. In Table \ref{table1},  we list the number of polar terms (denoted by $n_1^p$) and constraints (denoted by $n_1^c$) for the 13 smooth hypergeometric threefolds computed in \cite{Alexandrov:2022pgd}, such that the dimension of  $\scM_1(\CY)$ is given by $n_1^p-n_1^c$.

\subsection{Rank 1 DT invariants and stable pair invariants}
\label{sec_DTLV}

We now turn to the case $p^0=\pm 1$, as the corresponding invariants will turn out to provide the information needed to compute the polar coefficients in the generating series of D4-D2-D0 Abelian indices.

For $p^0=1$ and $p^1=0$, the index $\Omega_{\infty}(1,0,q_1,q_0)$ reduces to the invariant originally defined in \cite{MR1818182}, counting ideal sheaves $E$ with $\ch(E)=1-\beta- n H^3$, 
with $\beta=(q_1+\frac{c_2}{24}) H^2/\kappa$ (identified by Poincar\'e duality with a class in $H_2(\CY,\IZ)$) and $n=-q_0\in\IZ$. Equivalently, it counts dimension-one subschemes $C\subset \CY$ with class $[C]=\beta$ and holomorphic Euler number 
$\chi(\cO_C)=n$. The moduli space of such subschemes is projective and admits a perfect symmetric obstruction theory (see e.g. \cite{Pandharipande:2011jz} and references therein). We denote the corresponding DT invariant by 
\be
\label{noteDT}
{\mathrm I}_{n,\beta}=\DT(\beta. H,n)
=\Omega_\infty\left(1,0,-\beta.H-\frac{c_2}{24},-n\right)
,
\ee 
where the first notation is standard in the mathematics literature and the second was used in \cite{Alexandrov:2022pgd}. The case $p^0=1, p^1\neq 0$ can be reached by tensoring by a line bundle $\cL$ on $\CY$, or equivalently using the spectral flow \eqref{specflowD6}.
As a result the index $\Omega_\infty(\gamma)=\DT(Q_+,n_+)$ depends only on the invariant 
combinations
\be
\label{defQn}
Q_+ = q_1 + \frac12\, \kappa (p^1)^2 + \frac{c_{2}}{24}\, ,
\qquad
n_+ = - q_0 - p^1 q_1 - \frac13\, \kappa (p^1)^3\, ,
\ee
which both take integral values as a consequence of the quantization conditions \eqref{quant} and the integrality of  
the arithmetic genus of the divisor class $\cD_r$ with $r=p^1$,
\be
\label{defL0}
\chi_{\cD_r} \coloneqq \chi(\cO_{\cD_r}) 
= \frac16\, \kappa r^3 + \frac{1}{12}\, c_{2} r \, .
\ee

For $p^0=-1$ and $p^1=0$, the index $\Omega_{\infty}(-1,0,q_1,q_0)$ instead counts
stable pairs \cite[\S 3]{Toda:2011aa}
(more precisely, derived dual of stable pairs) 
$E=(\cO_\CY \stackrel{s}{\rightarrow} F)^\vee$, where $F$ is a pure one-dimensional sheaf with $[F]=\beta=-q_1 H^2/\kappa$ and $\chi(F)=n=-q_0$,
and $s$ is a section  with zero-dimensional cokernel. The Chern character for this object is $\ch(E)=-1+\beta-n H^3$.
As shown by Pandharipande and Thomas~\cite{pandharipande2009curve}, the moduli space of stable pairs is also projective and admits a perfect symmetric obstruction theory. We denote the corresponding PT invariant by 
\be 
\label{notePT}
{\mathrm P}_{n,\beta}=\PT(\beta. H,n)
=\Omega_\infty\left(-1,0,\beta.H+\frac{c_2}{24},-n\right)
\,,
\ee
where the first notation is standard in the mathematics literature and the second is similar to the one used for DT invariants.
The case $p^0=-1, p^1\neq 0$ can again be reached by tensoring by a line bundle $\cL$ on $\CY$, so that the 
index $\Omega_\infty(\gamma)=\PT(Q_-,n_-)$ depends only on the invariant 
combinations
\be
\label{defQn2}
Q_- = -q_1 + \frac12\, \kappa (p^1)^2 + \frac{c_{2}}{24}\, ,
\qquad
n_- = - q_0 + p^1 q_1 - \frac13\, \kappa (p^1)^3\, .
\ee

As shown in Appendix \ref{sec_appS}, Theorem 2, 
the invariants $\DT(Q,n)$ and $\PT(Q,n)$ vanish unless
\be
\label{CastPT}
Q\geq 0 \quad \mbox{and}\quad n\geq - \left\lfloor \frac{Q^2}{2\kappa} + \frac{Q}{2}\right\rfloor.
\ee
where the first
condition follows from the Bogolomov-Gieseker inequality \eqref{BGineq}.
In fact, in the range $0\leq Q<\kappa$, the BMT inequality \eqref{BMTch} implies  the slightly stronger bound
\be
\label{CastPTs}
n \geq -\left\lfloor \frac{2Q^2}{3\kappa} + \frac{Q}{3}\right\rfloor .
\ee
Given these lower bounds on $Q$ and $n$
we can define the generating series
\be
\label{defZDT}
\begin{split}
Z_{DT} (y,\q)=&\, \sum_{Q,n} \DT(Q,n)\,
y^Q\,  \q^{n}, 
\\ 
Z_{PT} (y,\q)=&\, \sum_{Q,n} \PT(Q,n)\,
y^Q\,  \q^{n}.
\end{split}
\ee
In terms of these formal series,  the DT/PT relation conjectured
in \cite{pandharipande2009curve} and proven in \cite{toda2010curve,bridgeland2011hall} takes the simple form
\be
Z_{DT} (y,\q)= M(-\q)^{\chi_{\scriptstyle\CY}} \, Z_{PT} (y,\q),
\label{eqn:DTPTrelation}
\ee
where $M(\q)=\prod_{k>0}(1-\q^k)^{-k}$ is the Mac-Mahon function. 
In \S\ref{sec_direct}, we shall explain how PT invariants, hence also DT invariants, can be computed from the knowledge of the topological string partition function.

\section{Gopakumar-Vafa invariants and direct integration}
\label{sec_GVdirect}

In this section, we recall how Gopakumar-Vafa (GV) invariants can be 
determined by integrating the holomorphic anomaly equations satisfied by the topological string partition function. 
Physically, GV invariants were introduced as  multiplicities of five-dimensional BPS states that arise from M2-branes wrapping curves in a CY threefold~\cite{Gopakumar:1998ii,Gopakumar:1998jq}.
We shall not go into the details of the mathematical definition of GV invariants but instead refer to~\cite{Maulik2018} and for an introduction to~\cite{Huang:2020dbh}.

\subsection{PT/GV relation}
\label{sec_GVPT}
As explained in \cite{Bershadsky:1993cx}, the A-twisted topological string associates to any CY threefold $\CY$
an infinite family of genus $g$ topological string free energies $\mathcal{F}^{(g)}(\rmz,\bar{\rmz})$, which depend
on the K\"ahler moduli $\rmz$ in a non-holomorphic fashion.
In the `holomorphic limit' $\bar{\rmz}\to -\I\infty$, and in an appropriate K\"ahler gauge, $\mathcal{F}^{(g)}(\rmz,\bar{\rmz})$ reduces to the generating series of Gromov-Witten invariants,
\begin{align}
    F^{(g)}(\rmz)\equiv \lim\limits_{\bar{\rmz}\rightarrow-\I\infty}(X^0)^{2g-2}\mathcal{F}^{(g)}(\rmz,\bar{\rmz})=\sum\limits_{Q=1}^\infty \GW^{(g)}_Q e^{2\pi \I Q \rmz}\,,
    \label{eqn:ztopgwex}
\end{align}
where $\GW^{(g)}_Q\in\IQ$ depends only on the symplectic structure of $\CY$.  For $g=0,1$, there are additional 
polynomial terms in $\rmz$ which we have dropped for brevity.
For $g=0$, $F^{(0)}$ coincides (up to an overall factor $-1/(2\pi\I)^3$) with the worldsheet instanton contribution to the tree-level prepotential~\eqref{FLV}.
The instanton part of the topological string partition function is then given by
\be
\label{PsitopF}
\Psi_{\rm top}(\rmz,\lambda) = \exp\left( \sum_{g\geq 0} \lambda^{2g-2} F^{(g)}(\rmz) \right),
\ee
According to \cite{Gopakumar:1998ii,Gopakumar:1998jq}, 
the Gromov-Witten invariants $\GW^{(g)}_Q$ can be traded
for new invariants  $\GVg{g}$ defined by equating
\be
\label{eqn:gvex}
\log \Psi_{\rm top}(\rmz,\lambda)=\sum\limits_{g=0}^\infty\sum\limits_{k=1}^\infty\sum\limits_{Q=1}^\infty\frac{\GVg{g}}{k}\left(2\sin\frac{k\lambda}{2}\right)^{2g-2}e^{2\pi \I k Q\rmz}\,.
\ee
The integrality of the Gopakumar-Vafa invariants $\GVg{g}$ defined by \eqref{eqn:gvex} was shown in \cite{IonelParker:2013}\footnote{As discussed in \S\ref{sec_GVmax}, an independent definition of GV invariants which makes integrality manifest was proposed 
in~\cite{Maulik2018}, but its compatibility with \eqref{eqn:gvex} remains conjectural.}. 
More recently, it was shown
in~\cite{gvfinite} that for fixed degree $Q$, there is only a finite number of non-vanishing 
invariants $\GVg{g}$. We shall denote by $\gmax(Q)$ the maximal genus $g$ such that $\GVg{g}\neq 0$.

It was conjectured in~\cite{gw-dt,gw-dt2} that the topological string partition function is related to the generating series of PT invariants defined in~\eqref{defZDT} via
\be
\label{PsitopPT}
\Psi_{\rm top}(\rmz,\lambda) =  M(-e^{\I\lambda})^{\hf\,\chi_{\scriptstyle{\CY}}}
Z_{PT} \left(e^{2\pi\I \rmz/\lambda},e^{\I\lambda} \right) .
\ee
The corresponding relation to the partition function $Z_{DT}$, which follows by using~\eqref{eqn:DTPTrelation}, was motivated physically in~\cite{Iqbal:2003ds} and a derivation in M-theory was given in~\cite{Dijkgraaf:2006um}. The MNOP conjecture \eqref{PsitopPT} is known to hold for non-compact toric CY threefolds~\cite{gw-dt,gw-dt2}, and for complete intersections in products of projective spaces~\cite{Pandharipande:2012re}. We shall assume that it continues to hold for complete intersections in weighted projective spaces.

The MNOP relation~\eqref{PsitopPT} can be expressed in a product form such that the PT invariants are related to the GV invariants by the following PT/GV relation~\cite{gw-dt,gw-dt2}
\be
\begin{split}
\label{PTGV}
Z_{PT}(y,\q)=&\,
\prod_{Q>0}\prod_{k>0} \left(1-(-\q)^k y^Q\right)^{k \GVg{0}}
\\
&\times
\prod_{Q>0}\prod_{g=1}^{\gmax(Q)}
\prod_{\ell=0}^{2g-2}
\left(1- (-\q)^{g-\ell-1} y^Q
\right)^{(-1)^{g+\ell} {\scriptsize \begin{pmatrix} 2g-2 \\ \ell \end{pmatrix}}
\GVg{g}} \, .
\end{split}
\ee
After some elementary algebra, one can rewrite \eqref{PTGV} as a plethystic exponential~\cite{Pandharipande:2011jz}
\be
\label{PTGVpleth}
Z_{PT}(y,\q)= \PE\[
\sum_{Q>0} \sum_{g=0}^{\gmax(Q)} (-1)^{g+1} \GVg{g} 
 \(1-x\)^{2g-2} x^{(1-g)} y^Q\](-\q,y),
\ee
where 
\be
\PE[f](x,y)=\exp\(\sum_{k=1}^\infty\frac{1}{k}\, f(x^k,y^k)\).
\ee
Conversely, GV invariants may be expressed in terms of PT invariants 
by taking the plethystic logarithm~(see e.g.~\cite{Cadogan1971TheMF}),
\be
\begin{split}
&\sum_{Q>0}
\sum_{g=0}^{g_{\rm max}(Q)}
\GVg{g} (1+\q)^{2g-2} \q^{1-g} y^Q
\\
=&\,
\sum_{k=1}^\infty \frac{\mu(k)}{k}\log\left[ 1+\sum_{Q>0,m} (-1)^m (-\q)^{km}  \PT(Q,m)\, y^{kQ} \right],
\end{split}
\ee
where $\mu(k)$ is the M\"obius function (see Footnote \ref{fooMoebius}), 
and expanding in powers of $\q$ and $y$ on either side.
The plethystic representation of the MNOP relation turns out to be computationally much more efficient than the original formula \eqref{PTGV}.

It easily follows from this relation and the bound \eqref{CastPT} that for any $Q>0$, the maximal genus $\gmax(Q)$ is bounded from above by 
\be
\label{gCast}
\gmax(Q) \leq g_C(Q) := \left\lfloor \frac{Q^2}{2\kappa} + \frac{Q}{2}\right\rfloor+1 \, .
\ee
As in \eqref{CastPTs}, the bound is strengthened to 
$\gmax(Q) \leq\left\lfloor \frac{2Q^2}{3\kappa} + \frac{Q}{3}\right\rfloor+1$ when $0<Q<\kappa$. 
We refer to 
\eqref{gCast} as the Castelnuovo bound, in reference to Castelnuovo's work on the maximal arithmetic genus of irreducible curves in projective space \cite{Castelnuovo89} (see e.g. \cite{Griffiths1994} for a more recent account). In this work, we have obtained \eqref{gCast} using rather different methods pioneered in \cite{macri-genus-bound} (see Appendix \ref{sec_appII}).  
We note that the bound \eqref{gCast}  was established recently for the quintic threefold in \cite{Liu:2022agh}. 

It is worth noting that for $m$ and $g$ sufficiently close to the Castelnuovo bound, the relation between PT and GV invariants becomes linear,
\be
\label{PTapprox}
\PT(Q,m) = \sum_{g=1}^{g_{\rm max}(Q)} \binom{2g-2}{g-1-m} \,\GVg{g}\, .
\ee
The exact range of validity of this relation depends on $\CY$, but it is easy to check that it holds true for $m=m_{\min}(Q)+\delta$ with 
$m_{\min}(Q)\coloneqq 1-g_{\rm max}(Q)$ and 
$\delta=0,1$: 
\bea
\PT(Q,m_{\min}(Q)) &=& \GVg{\gmax(Q)}\, ,
\label{PTmax} 
\\
\PT(Q,m_{\min}(Q)+1) &=& \GVg{\gmax(Q)-1} + (2\gmax(Q)-2)  \GVg{\gmax(Q)}\, . 
\label{PTsubmax}
\eea
In particular, $m_{\min}(Q)$ is the minimal value of $m$ such that $\PT(Q,m)\neq 0$, and satisfies
\be
\label{mCast}
m_{\min}(Q) \geq m_C(Q) := - \left\lfloor \frac{Q^2}{2\kappa} + \frac{Q}{2}\right\rfloor .
\ee

\subsection{Direct integration method for computing GV invariants}
\label{sec_direct}

As shown in \cite{Bershadsky:1993ta,Bershadsky:1993cx}, the topological string free energies $\mathcal{F}^{(g)}(\rmz,\bar{\rmz})$ satisfy the holomorphic anomaly equations
\begin{gather}
\label{F1hae}
\frac{\partial}{\partial {\brmz}} \frac{\partial}{\partial \rmz}\mathcal{F}^{(1)}=\frac12\, C_{\brmz}^{\rmz\rmz}C_{\rmz\rmz\rmz}+\left(\frac{\chi_\CY}{24}-1\right)G_{\brmz\rmz}\,,
\\
\label{Fghae}
\frac{\partial}{\partial {\brmz}} \mathcal{F}^{(g)} = 
\frac12\, \bar C_{\brmz}^{\, \rmz\rmz} \left( D^2 \cF^{(g-1)} + \sum\limits_{h=1}^{g-1} D \cF^{(g-h)} D \cF^{(h)} \right),
\quad
\text{for }g\ge 2\,,
\end{gather}
where $C_{\rmz\rmz\rmz}=\partial^3_{\rm z}F^{(0)}(\rmz)$ is the so-called Yukawa coupling, 
$\bar C_{\brmz\brmz\brmz}$ is its complex conjugate, and indices are raised using the K\"ahler metric $G_{\rmz\brmz}=\partial_{\rmz}\partial_{\brmz}K$ with
$K=-\log(\brmz \partial_{\rmz} F - \rmz \partial_{\brmz} \bar F)$.
In the flat coordinates $\rmz,\brmz$, the Christoffel symbols $\Gamma^\rmz_{\rmz\rmz}=G^{\brmz\rmz}\partial_{\rmz}G_{\brmz\rmz}$ vanish.
Denoting the Hodge line bundle with connection $K_\rmz=\partial_\rmz K$ on the moduli space by $\mathcal{L}$, the free energies $\mathcal{F}^{(g)}$ are sections of $\mathcal{L}^{2-2g}$ and the covariant derivative acting on a section of $\mathcal{L}^n$ takes the form $D=\partial_{\rmz}+n K_{\rmz}$.
Given the amplitudes $\mathcal{F}^{(h)}(\rmz,\brmz)$ for $h<g$, the equations
\eqref{Fghae} determine $\mathcal{F}^{(g)}(\rmz,\brmz)$ up to a holomorphic ambiguity $f^{(g)}(\rmz)$.

The non-holomorphic dependence of the free energies can be absorbed in a set of
`propagators'
$S^{\rmz\rmz},S^{\rmz}, S$, satisfying \cite{Bershadsky:1993cx}
\be
\partial_{\brmz} S^{\rmz\rmz} = \bar C_{\brmz}^{\, zz}\,,
\qquad
\partial_{\brmz} S^{\rmz} = G_{\rmz\brmz} S^{\rmz\rmz}\,,
\qquad
\partial_{\brmz} S = G_{\rmz\brmz} S^{\rmz}\,.
\ee
More precisely, $\mathcal{F}^{(g)}(\rmz,\brmz)$ is an inhomogeneous polynomial of degree $3g-3$ in 
$K_\rmz,S^{\rmz\rmz},S^{\rmz}, S$ (of respective degrees $1, 1,2,3$)
with holomorphic coefficients~\cite{Yamaguchi:2004bt,Alim:2007qj,Hosono:2008np}.
It turns out that the dependence on the connection $K_\rmz$ can also be absorbed by introducing the shifted propagators~\cite{Alim:2007qj,Hosono:2008np}
\begin{align}
    \begin{split}
        \tilde{S}^{\rmz\rmz}=S^{\rmz\rmz}\,,\qquad \tilde{S}^\rmz=S^\rmz-S^{\rmz\rmz}K_\rmz\,,\qquad \tilde{S}=S-S^\rmz K_\rmz+\frac12\, S^{\rmz\rmz}K_\rmz K_\rmz\,.
    \end{split}
\end{align}
Up to a holomorphic ambiguity $f^{(1)}(\rmz)$, the equation~\eqref{F1hae} can then be integrated to obtain
\begin{align}
    \partial_\rmz \mathcal{F}^{(1)}=\frac12\, C_{\rmz\rmz\rmz}\tilde{S}^{\rmz\rmz}-\left(\frac{\chi_\CY}{24}-1\right)K_\rmz+f^{(1)}(\rmz)\,,
\end{align}
and the holomorphic anomaly equations~\eqref{Fghae} for $g\ge 2$ can be rewritten as
\begin{align}
    \frac{\p \mathcal{F}^{(g)}}{\p \tilde{S}^{\rmz\rmz}}-K_\rmz\,\frac{\p \mathcal{F}^{(g)}}{\p \tilde{S}^{\rmz}}+\frac12\,K_\rmz K_\rmz\frac{\p \mathcal{F}^{(g)}}{\p \tilde{S}}=\frac12\, D^2 \mathcal{F}^{(g-1)}+\frac12 \sum\limits_{h=1}^{g-1}D \mathcal{F}^{(g-h)} D \mathcal{F}^{(h)}\,.
    \label{eqn:holanshift}
\end{align}
The holomorphic limit is obtained by replacing $K_\rmz,\tilde{S}^{\rmz\rmz},\tilde{S}^\rmz,\tilde{S}$ with the corresponding holomorphic limits $\rm{K}_\rmz,\rm{S}^{\rmz\rmz},\rm{S}^\rmz,\rm{S}$ and $\mathcal{F}^{(g)}(\rmz,\brmz)$ by $F^{(g)}(\rmz)$.
Since the dependence of the free energies on $K_\rmz$ is absorbed in the shifted propagators, the equations~\eqref{eqn:holanshift} can be integrated by collecting the powers of $K_\rmz$ on the right-hand side, and identifying them with the corresponding derivatives on the left-hand side.

In terms of the algebraic coordinate $z$, special geometry implies
\begin{align}
    \Gamma^z_{zz}=2K_z-C_{zzz}\tilde{S}^{zz}+s^z_{zz}\,,
    \label{eqn:sgeorel}
\end{align}
and the propagators are partially determined by the BCOV ring~\cite{Alim:2007qj}
\begin{align}
\begin{split}
    \partial_z \tilde{S}^{zz}=&\,C_{zzz}\tilde{S}^{zz}\tilde{S}^{zz}+2\tilde{S}^z-2s^z_{zz}\tilde{S}^{zz}+h^{zz}_z\,,
    \\
    \partial_z \tilde{S}^z=&\,C_{zzz}\tilde{S}^{zz}\tilde{S}^z+2 \tilde{S}-s^z_{zz}\tilde{S}^z-h_{zz}\tilde{S}^{zz}+h^z_z\,,
    \\
    \partial_z \tilde{S}=&\,\frac12\, C_{zzz}\tilde{S}^z\tilde{S}^z-h_{zz}\tilde{S}^z+h_z\,,
    \\
    \partial_z K_z=&\,K_zK_z-C_{zzz}\tilde{S}^{zz}+s^z_{zz}K_z-C_{zzz}\tilde{S}^z+h_{zz}\,,
    \label{eqn:bcovring}
\end{split}
\end{align}
up to another set of holomorphic (propagator) ambiguities $s^z_{zz},h^{zz}_z,h^z_z,h_z,h_{zz}$.
There is no canonical way to fix these ambiguities and different choices lead to a different functional dependence of the free energies on the propagators~\cite{Hosono:2008np,Alim:2008kp}.
For the 13 hypergeometric families, it turns out that there are always solutions of the form
\begin{align}
    s^z_{zz}=\frac{1}{z}\,\tilde{s}^z_{zz}\,,\qquad h^{zz}_z=z\tilde{h}^{zz}_{z}\,,\qquad h^z_z=0\,,\qquad h_z=\frac{1}{z}\,\tilde{h}_z\,,\qquad h_{zz}=\frac{1}{z^2}\,\tilde{h}_{zz}\,,
    \label{eqn:propambans}
\end{align}
with $\tilde{s}^z_{zz},\tilde{h}^{zz}_z,\tilde{h}_z,\tilde{h}_{zz}\in\mathbb{Q}$. Such solutions are determined by a polynomial equation in $s^z_{zz}$, and we pick the root such that $\tilde{h}_z$ has the maximal possible value.

The free energies in terms of the propagators can be obtained by integrating~\eqref{eqn:holanshift}.
If necessary, the full non-holomorphic dependence can then be restored by inserting the corresponding expressions for the propagators.
However, to obtain the enumerative invariants that are encoded in the free energies we only need to consider the holomorphic limit.
Using ${\rm K}_z=-\partial_z\log(X^0)$ and the Ansatz~\eqref{eqn:propambans}, the BCOV ring~\eqref{eqn:bcovring} can be used to calculate the holomorphic limits of the propagators.
Before carrying out the direct integration procedure, it remains to discuss how the holomorphic ambiguities $f^{(g)}(\rmz)$ that arise at each genus from the integration of \eqref{eqn:holanshift} can be fixed.

Let us first discuss the solution at genus $g=1$ for the hypergeometric families.
Combining~\eqref{F1hae} with~\eqref{eqn:sgeorel} and using the behaviour in the large volume limit~\cite{Bershadsky:1993ta} and at the conifold point~\cite{Vafa:1995ta}, the holomorphic anomaly equation for the genus one free energy can be integrated to obtain
\begin{align}
    \mathcal{F}^{(1)}=-\frac12\left(4-\frac{\chi_\CY}{12}\right)K-\frac12\log\det G_{\brmz \rmz}-\frac{1}{24}\left(12+c_2
    \right)\log z-\frac{1}{12}\log\Delta\,,
\end{align}
where $\Delta$ is the discriminant polynomial and $c_2$ is the numerical second Chern class defined below \eqref{defCvec}. 

At genus $g\ge 2$, the ambiguity can be written as a rational function in terms of the algebraic coordinates. It follows from the analysis in~\cite{Yamaguchi:2004bt,Huang:2006hq}, see~\cite{MR3965409} 
for a review of B-model techniques, that
\begin{align}
\label{defKg}
	f^{(g)}(z)=\frac{1}{\Delta^{2g-2}}\sum\limits_{k=0}^{2g-2}f_k z^k+\sum\limits_{k=1}^{N(g)} f'_k z^k\,,
	\qquad 
	N(g)=\left\lfloor \frac{2(g-1)}{\rho}\right\rfloor,
\end{align}
where the `regulator' $\rho$ is the smallest denominator among the local exponents $a_i$, and $f_k,f'_k$ are rational coefficients.
To fix the coefficients, one can use known enumerative invariants --- for example due to Castelnuovo like vanishing --- together with the behaviour of the free energies around special points in the moduli space.

The generic constant map contribution of the free energies in the large volume limit $z=0$ and the so-called gap condition at the conifold point $z=\mu$ can be used to fix all of the $f_k,\,k=0,\ldots,2g-2$.
On the other hand, for most of the hypergeometric families the current knowledge about the behaviour at $z=\infty$ is limited to the degree of regularity of the free energies at this point and, as discussed in~\cite{Huang:2006hq}, determines the upper limit $N(g)$ of the second sum in~\eqref{defKg}.
Additional constraints are currently only understood for $X_{6,2}$, $X_{4,2}$ and most recently also for $X_{2,2,2,2}$~\cite{Katz:2022lyl}.
For $X_{6,2}$, $X_{4,2}$ the point at infinity is of conifold type and the expansion of the free energies around this point satisfy an additional `small gap'.
This imposes an additional $\lfloor 2g/\rho'\rfloor$ constraints, with $\rho'$ for the two geometries respectively given by $3$ and $4$.
On the other hand, the point at infinity in the moduli space of $X_{2,2,2,2}$ corresponds to a `non-commutative resolution' of a singular degeneration of $X_8$.
The corresponding free energies encode certain $\mathbb{Z}_2$-refined GV-invariants that also exhibit a Castelnuovo-like vanishing~\cite{Katz:2022lyl}.

Using the Castelnuovo bound~\eqref{gCast},
together with the closed expression~\eqref{eqn:gvgmax} for the invariants that saturate the bound and, in the case of $X_{6,2}$, $X_{4,2}$, additional conditions at infinity,
the coefficients of the holomorphic ambiguity can be completely fixed as long as
\begin{align}
\label{Kgmax}
	N(g)\le \left\lfloor -\frac{\kappa}{2}+\frac12\sqrt{\kappa\bigl( 8(g-1)+\kappa\bigr)}\right\rfloor+\left\{\begin{array}{cl}
	\lfloor \frac{2g}{3}\rfloor&\text{ for }X_{6,2}\\
	\lfloor \frac{g}{2}\rfloor&\text{ for }X_{4,2}\\
	0&\text{ else}
	\end{array}\right.\,.
\end{align}
Moreover, as discussed in~\cite{Katz:2022lyl}, for $X_{2,2,2,2}$ taking into account the additional Castelnuovo bound for the $\mathbb{Z}_2$-refined GV invariants at the point at infinity determines the holomorphic ambiguity up to genus $32$.

We denote by $g_{\rm integ}$ the maximal value of $g$ for which the previously discussed boundary conditions are sufficient to fix the holomorphic ambiguity\footnote{Ignoring the floor functions in \eqref{defKg} and \eqref{Kgmax}, and absorbing the correction term in \eqref{Kgmax} into an effective regulator  $\rho=6$ for $X_{6,2}$ or  $\rho=4$  for $X_{4,2}$, one finds a rule-of-thumb estimate $g_{\rm integ}\simeq \frac12\kappa\rho(\rho-1)+1$.\label{foothumb}}, and tabulate its values for the various hypergeometric models in 
Table \ref{table1}. Due to computational limitations, we have not yet reached this genus for all models. In \S\ref{sec_test}, we shall see that the knowledge of D4-D2-D0 invariants can be used to push the direct integration method to even higher genus, denoted by $g_{\rm mod}$ in Table \ref{table1}.

\subsection{GV invariants at maximal and submaximal genus}
\label{sec_GVmax}

Although the definition of GV invariants via Gromov-Witten invariants presented in \S\ref{sec_GVPT} makes it clear that
they are robust under complex structure deformations of $\CY$, its main drawback is that integrality of the resulting invariants is not manifest. In~\cite{Maulik2018} an 
alternative definition using moduli of stable sheaves was proposed, inspired by the 
geometric picture developed in~\cite{Katz:1999xq} (and earlier attempts in \cite{kiem2012categorification,Hosono:2001gf}).
While the mathematical definition in~\cite{Maulik2018} is quite involved (see \cite{Huang:2020dbh} for a  review aimed at physicists),
the  approach of~\cite{Katz:1999xq}  can be used to  calculate GV invariants near the Castelnuovo bound, 
at least heuristically. The results \eqref{eqn:gvgmax} and \eqref{eqn:gvgmaxm1} will be justified rigorously in \S\ref{sec_wcrm1} by combining Theorem \ref{thm-Cast} with the MNOP conjecture.

Motivated by the interpretation in terms of bound states of D2-D0 branes in Type IIA string theory, one considers one dimensional (semi-)stable sheaves supported on a curve of class $\beta\in H_2(\CY,\IZ)$.
The invariants are conjecturally independent of the D0-brane charge~\cite{Toda:2017sdo,Maulik2018}, which can therefore be taken to be $1$, such that semi-stability implies stability.
For a fixed curve class $\beta$, the corresponding moduli space of stable sheaves $\widehat{\mathcal{M}}_\beta$ is fibered over the Chow variety $\mathcal{M}_\beta=\text{Chow}(\beta)$, which parametrizes effective curves $\cC\subset \CY$ with $[\cC]=\beta$.
If a point $p\in\widehat{\mathcal{M}}_\beta$ projects to a smooth curve of genus\footnote{In this section we abuse notation and denote $\GVg[\beta]{g}=\GVg[\beta.H]{g}$ and $\gmax(\beta)=\gmax(\beta.H)$.}
$\gmax(\beta)$, the corresponding fiber is the  Jacobian torus $T^{2\gmax(\beta)}$ of $\cC$.

It was argued in~\cite{Katz:1999xq} that the little group $Spin(4)=SU(2)_L\times SU(2)_R$ in the five-dimensional effective theory arising from M-theory compactified on $\CY$, should be identified with the product of the Lefschetz actions on the cohomology of the fiber and base of $\widehat{\mathcal{M}}_\beta$. As a result, 
in cases where $\widehat{\mathcal{M}}_\beta$ is smooth, the genus zero GV invariants can  be defined as $\GVg[\beta]{0}=(-1)^{\text{dim}_{\mathbb{C}}\,\widehat{\mathcal{M}}_\beta}\chi(\widehat{\mathcal{M}}_\beta)$. More generally, genus zero GV invariants are related to generalized Donaldson-Thomas invariants via $\GVg[\beta]{0}=\Omega_H(0,0,\beta,1)$, where $H$ is any ample divisor on $\CY$~\cite{Katz:2006gn,Joyce:2008pc}.

On the other hand, if the Chow variety $\mathcal{M}_\beta$ itself is smooth, one finds that for maximal genus $g=\gmax(\beta):=\gmax(\beta.H)$,
\begin{align}
    \GVg[\beta]{\gmax(\beta)} =(-1)^{\text{dim}_{\mathbb{C}}\,\mathcal{M}_\beta}\chi(\mathcal{M}_\beta)\,.
\end{align}
Invoking a localization argument motivated by~\cite{Yau:1995mv}, the authors of~\cite{Katz:1999xq} propose further geometric expressions for the invariants at genera close to $\gmax(\beta)$.
In particular, in favorable cases
\begin{align}
\GVg[\beta]{\gmax(\beta)-1}=\,(-1)^{\text{dim}_{\mathbb{C}}\,\mathcal{M}_\beta+1}\Bigl[\chi(\mathcal{C}_\beta)+(2\gmax(\beta)-2)\chi(\mathcal{M}_\beta)\Bigr]\, ,
    \label{eqn:kkvgv2}
\end{align}
where $\mathcal{C}_\beta\subset \CY\times \mathcal{M}_\beta$ is the the universal curve.
We observe that these relations agree with ~\eqref{PTmax},~\eqref{PTsubmax}, after identifying~\cite{Pandharipande:2007qu,Choi:2012jz}
\begin{align}
    \begin{split}
    \PT(\beta,m_{\min}(\beta))=&\,
    (-1)^{\text{dim}_{\mathbb{C}}\,\mathcal{M}_\beta}\chi(\mathcal{M}_\beta)\,,
    \\
    \PT(\beta,m_{\min}(\beta)+1)=&\, (-1)^{\text{dim}_{\mathbb{C}}\,\mathcal{M}_\beta+1}\chi(\mathcal{C}_\beta)\,.
     \end{split}
\end{align}
We shall now apply these relations to determine the GV invariants
$\GVg{\gmax(Q)}$ and $\GVg{\gmax(Q)-1}$ for degree $Q=\kappa d$ with $d\in\mathbb{N}$ for the 13 hypergeometric CY threefolds.

Let us consider a CY threefold $\CY$ obtained as a complete intersection of $n$ generic hypersurfaces of respective degrees $(d_1,\dots,d_n)$ in weighted projective space $W:= \IP^{n+3}_{w_1,\dots,w_{n+4}}$.
Curves of degree $Q=\kappa d$ on $\CY$ are obtained by intersecting $\CY$ with two additional hypersurfaces of respective degrees 1 and $d$.
Using the adjunction formula, one can check that a generic curve of this type has the maximal possible genus $g=\gC(\kappa d)$.

We can identify the restriction of the linear subspace to $\CY$ with the ample divisor $\cD$.
By Bertini's theorem, $\cD$ is smooth if the restriction of $\mathcal{O}_{W}(1)$ to $\CY$ is base-point free.\footnote{Recall that the base locus of a bundle consists of the points where all sections vanish simultaneously.}
This is generically the case if the number
\begin{align}
    \chiOD=\#\{\,i\,|\,w_i=1\,,\,\, i=1,\ldots,n+4\}\,,
\end{align}
of weights that are equal to one is strictly greater than three.
The equality with $\chi(\cO_D)$ can be verified using the Hirzebruch-Riemann-Roch (HRR) theorem.
Then a generic $(1,d)$-curve is smooth if the restriction of $\mathcal{O}_{W}(d)$ to $\cD\subset \CY$ is basepoint free as well, which is automatically implied.
Comparing with Table~\ref{table1}, we see that $\cD$ is singular for $X_{10}$, $X_{6,4}$ and $X_{6,6}$.

For $d\ge 2$, the moduli space $\mathcal{M}(1,d)$ of complete intersection curves of degree $(1,d)$ is a projective bundle with fibers $\mathbb{P}\left[h^0\left(\cD,\mathcal{O}_{\cD}(d)\right)-1\right]$ over $\mathbb{P}\left[\chiOD-1\right]$, using $\mathbb{P}[k]:=\mathbb{P}^k$.
Using again HRR or generating functions, we further calculate  that
\begin{align}
    h^0\left(\cD,\mathcal{O}_{\cD}(d)\right)=\left\{
    \begin{array}{cl}
    \chi_{\cD}-1& d=1,
    \\
    \frac12 \kappa d(d-1)+\chi_{\cD}& d\ge 2.
    \end{array}\right.
\end{align}
Assuming that every smooth curve of degree $Q=\kappa d$ is a complete intersection, we conclude that
\begin{align}
	\GVg[\kappa d]{\gC(\kappa d)}=\left\{
		\begin{array}{cl}
			\frac12\chi_{\cD}(\chi_{\cD}-1)&d=1,
			\\
			(-1)^{\frac{1}{2}\kappa d(d-1)}\chi_{\cD}\, h^0\left(\cD,\mathcal{O}_{\cD}(d)\right)&d\ge 2,
		\end{array}\right.
		\label{eqn:gvgmax}
\end{align}
where for $d=1$ we took into account the fact that the two linear sections play a symmetric role.

To calculate $\GVg[\kappa d]{\gC(\kappa d)-1}$, we first note that the universal curve $\mathcal{C}_{\kappa d}$ is fibered over $\CY$, with the fiber over a point $p$ being the subset $\mathcal{M}_p\subset\mathcal{M}$ of curves that intersect $p$ 
(to avoid cluttering, we now suppress the subscript denoting the curve class).
If the point $p$ is sufficiently generic, which is always true if $\chi_\cD>3$, we obtain one condition on the linear section as well as on the degree $d$ section, such that
\begin{align}
	\chi(\mathcal{M}_p)=\left\{
		\begin{array}{cl}
			\frac12(\chi_{\cD}-1)(\chi_{\cD}-2)&d=1,
			\\
			(\chi_{\cD}-1)\bigl(\chi_{\cD}-1+\frac{1}{2}\kappa d(d-1)\bigr) & d\ge 2.
		\end{array}\right.
\end{align}
For the nine hypergeometric cases with a smooth divisor $\cD$, using~\eqref{eqn:kkvgv2} together with $\chi(\mathcal{C})=\chi(\mathcal{M}_p)\times\chi({\CY})$ one then arrives at
\begin{align}
	\GVg[\kappa d]{\gC(\kappa d)-1}=\left\{
		\begin{array}{cl}
			\frac12J_1\bigl(h^0\(\cD,\mathcal{O}_{\cD}(d)\)-1\bigr)+(\gC(\kappa d)-1)(\chi_{\cD}-1)\chi_{\cD} &d=1,
			\\
			J_1\bigl( h^0\left(\cD,\mathcal{O}_{\cD}(d)\right)-1\bigr)+(2\gC(\kappa d)-2)\chi_{\cD}\, h^0\left(\cD,\mathcal{O}_{\cD}(d)\right)&d\ge 2,
		\end{array}\right.
		\label{eqn:gvgmaxm1}
\end{align}
with $J_1=\chi_\CY (\chi_{\cD}-1)$.

In Section~\ref{sec_wcrm1} we shall derive the expressions~\eqref{eqn:gvgmax} and~\eqref{eqn:gvgmaxm1} using the relation between PT invariants and rank 0 DT invariants.
In particular, we shall find that~\eqref{eqn:gvgmax} holds also for $X_{10}$, $X_{6,4}$, $X_{6,6}$, and~\eqref{eqn:gvgmaxm1} holds for those geometries if  $\kappa d\ge 4$ and $J_1$ is defined in terms of a particular rank 0 DT invariant counting D4-D0 bound states (see below~\eqref{eqn:PTJ1}).

\section{D4-D2-D0 indices from GV invariants}
\label{sec_rank0fromPT}

In this section, we explain how to compute the Abelian D4-D2-D0 indices
$\Omega_{1,\mu}(\hq_0)$ introduced in \S\ref{sec_modconj} in terms of the Gopakumar-Vafa invariants $\GVg{g}$ 
determined by the topological string partition function. The strategy is to combine the 
relation between rank 0 DT invariants and PT invariants, investigated in the series of mathematical papers~\cite{Toda:2011aa,Feyzbakhsh:2020wvm,Feyzbakhsh:2021rcv,Feyzbakhsh:2021nds,Feyzbakhsh:2022ydn}, 
with the PT/GV relation explained in \S\ref{sec_GVPT}. 
Unfortunately, the explicit formulae stated in Thm 1.1 and Thm 1.2 of \cite{Feyzbakhsh:2022ydn} 
are not yet sufficient for our purposes. 
In Appendix \ref{sec_appS}, one of the authors proves a generalization of both theorems
which we present in the following two subsections using more physics-friendly notations.
The first theorem has a close relationship to the physical picture based on D6-$\overline{\rm D6}$ 
bound states advocated in earlier works on D4-D2-D0 indices~\cite{Gaiotto:2006wm,Gaiotto:2007cd,Denef:2007vg,Collinucci:2008ht,VanHerck:2009ww,Gaddam:2016xum,Alexandrov:2022pgd}, but turns out to be much less powerful than the second theorem which is
fully explicit and allows to compute a large number of Abelian D4-D2-D0 indices. 
It also implies the Castelnuovo bounds on PT and GV invariants,
as we explain in \S\ref{sec_wcrm1}.

\subsection{Wall-crossing for rank 0 class}
\label{sec_wcr0}

In Theorem \ref{thm.rk0} from Appendix \ref{sec_appS}, a slightly stronger version of \cite[Thm 1.1]{Feyzbakhsh:2022ydn} is established by studying the walls of $\nu_{b,w}$-instability for rank 0 classes. 
In this subsection we reformulate this result
by restricting to CY threefolds with $b_2(\CY)=1$ and vanishing torsion
$H^2(\CY,\IZ)_{\rm tors}=0$, and translating to the notations of \S\ref{sec_modconj} and \S\ref{sec_DTLV}. 
To this end, we 
identify in Eqs. \eqref{defQSoh}-\eqref{cc.2}
\be
\begin{split}
D=&\,  rH,
\qquad \qquad
\beta. H=-\mu-\frac12\,\kappa r^2, 
\qquad 
m=n-\frac16\,\kappa r^3=- q_0 + \frac{r c_2}{24}\,,
\\
D_i=&\, (-1)^{i} r_i H,
\qquad\quad
\beta_i.H=Q_i, 
\qquad \qquad
n_i=(-1)^{i} m_i, 
\qquad
i=1,2.
\end{split}
\ee
Under these identifications, we obtain that, provided the reduced D0-brane charge $\hq_0$ \eqref{defqhat} lies in the range
\be
\label{condFeyz}
0 \leq \frac{\chi(\cD_r)}{24} -\hq_0 < \frac{\kappa r}{12} \min\(\frac{r^2}{2}-\frac18\, ,\, r-\hf\),
\ee
the rank 0 DT invariant \eqref{defOmrmu} can be expressed as  
\be
\bOm_{r,\mu}(\hq_0) = \sum_{r_i,Q_i,n_i} 
(-1)^{\gamma_{12}}\,\gamma_{12} \,
\PT(Q_1,n_1) \, \DT(Q_2,n_2)\,,
\label{conseqTh1}
\ee
where (using the notation \eqref{defL0})  
\be
\gamma_{12}
=r(Q_1+Q_2)+n_1+n_2-\chi_{\cD_r}\,,
\ee
and the sum runs over integers $r_i$, $Q_i$ and  $n_i$ restricted to satisfy
\be
\label{v12Q}
\begin{split}
r_1+r_2=&\, r\, , 
\\
Q_2-Q_1=&\, \mu+\kappa rr_2 \, ,
\\
n_1+n_2=&\, n-r_1 Q_1-r_2 Q_2-\frac{\kappa}{2}\, rr_1r_2\, ,
\end{split}
\ee
and
\be
\begin{split}
&\qquad
\left| r_i-\sqrt{\frac{r^2}{4}-\frac{6}{r\kappa} \(\frac{\chi(\cD_r)}{24} -\hq_0\)}+(-1)^i\(\frac{r}{2}+\frac{\mu}{r\kappa}\)\right|<1\, ,
\\
& 
0\leq  Q_i \leq \frac{3}{r}\(\frac{\chi(\cD_r)}{24} -\hq_0\)
+\frac{1}{2\kappa}\(\frac{\mu}{r}+\frac{\kappa}{2}\(r+2(-1)^i r_i\)\)^2 - \frac{\kappa r}{8}\, , 
\\
&\qquad\qquad\qquad
 n_i \geq -\frac23\, Q_i\( \frac{Q_i}{\kappa}+\frac{1}{2}\).
\end{split}
\label{ineqDTPT}
\ee

Physically, the  r.h.s. of 
\eqref{conseqTh1} can be interpreted as contributions of two-centered bound states of an anti-D6-brane
bound to $(Q_1,n_1)$ D2-D0 branes, carrying index $\PT(Q_1,n_1)$, and a D6-brane bound to $(Q_2,n_2)$ D2-D0 branes, carrying index $\DT(Q_2,n_2)$, with the D4-brane charge arising from the fluxes $r_i$ on either side.

Unfortunately, the condition
\eqref{condFeyz} is so restrictive that the theorem can only apply, at best, to the most polar term 
in each component of the modular vector $h_{1,\mu}$. In particular, for $\mu=0$ it is valid only for 
$\hq_0=\frac{\chi(\cD_r)}{24}$ where only $Q_i=n_i=0$ contribute, leading to
\be
\label{eqFeyzus}
\bOm_{r,0}\left( \frac{\chi(\cD_r)}{24} \right) =  (-1)^{1+\chi_{\cD_r}} \chi_{\cD_r}\,,
\ee
where $\chi_{\cD_r}$ was defined in \eqref{defL0}.
In practice however, it was observed in \cite[\S D]{Alexandrov:2022pgd}
that the formula \eqref{conseqTh1} predicts the correct polar terms in many examples with $r=1$, provided one restricts the sum only to $Q_1=n_1=0$. Using $\PT(0,0)=1$, one arrives at the naive Ansatz for polar coefficients in \cite[(5.20)]{Alexandrov:2022pgd},
\be
\label{naive}
\bOm_{r,\mu}(\hq_0) = (-1)^{r \mu+n + \chi_{\cD_r}} (r \mu + n-\chi_{\cD_r})\, \DT(\mu,n)\,,
\ee
where $n$ is the integer defined in \eqref{nfromhq0}. 
The physical intuition for this Ansatz was that D4-D2-D0 branes at large volume arise as bound states of a D6-brane with D2-brane charge $\mu$ and D0-brane charge $n$, and an anti-D6-brane carrying $-r$ units of D4-brane flux.  Unfortunately, it appears difficult to relax
the condition \eqref{condFeyz}, and to justify physically or mathematically the truncation
to terms with $Q_1=n_1=0$, which appears to work in many cases.

\subsection{Wall-crossing for rank $-1$ class}
\label{sec_wcrm1}
In \cite[Thm 1.2]{Feyzbakhsh:2022ydn}, one of the authors of the present work obtained a different formula relating rank 0 and rank 1 DT invariants, which is valid only for CY threefolds with $\Pic\CY=\IZ$ (hence $b_2(\CY)=1$ and 
$H^2(\CY,\IZ)_{\rm tors}=0$). The formula follows by studying the possible walls for objects of rank $-1$ class 
\be
\label{vrankm1}
\v_k=\v-e^{-kH} = \[-1, D+kH, \beta-\tfrac12\, k^2H^2, -m+\tfrac16\, k^3H^3\]\,,
\ee 
in the space of weak stability conditions for $k\gg 1$, and applies for arbitrary $D\in H^2(\CY,\IZ)$, Poincar\'e dual to an arbitrary divisor class. Unfortunately, an explicit lower bound on $k$ was not provided. 
In Appendix \ref{sec_appS}, restricting to the case of primitive divisor, which is sufficient for computing Abelian D4-D2-D0 invariants,
a more general formula is derived that does not require taking $k$ large. 
Below, we rephrase Theorem \ref{thm-main} from Appendix \ref{sec_appS} using the same notations as in the previous subsection, and explain how to use it to compute Abelian D4-D2-D0 invariants from the knowledge of GV invariants. 

\subsubsection*{Main result}
Let us fix $(Q,m)\in \IZ_+\times \IZ$, and 
define the function $f:\IR^+\to \IR$ by
\be
f(x)\ \coloneqq \ \left\{\!\!\!\begin{array}{cc} x+ \frac{1}{2} & \text{if $0 <x < 1$,} 
\\
\vspace{.1 cm}
\sqrt{2x+\frac{1}{4}} & \text{if $1 \leq x < \frac{15}{8}$}\, ,
\\
\vspace{.1 cm}
\frac{2}{3}x+ \frac{3}{4} & \text{if $\frac{15}{8} \leq x< \frac{9}{4}$}\, ,
\\
\vspace{.1 cm}
\frac{1}{3}x+ \frac{3}{2} & \text{if $\frac{9}{4} \leq x< 3$,}
\\
\vspace{.1 cm}
\frac{1}{2}x+ 1 & \text{if $3 \leq x$.}
\end{array}\right.
\end{equation}
Note that this function is uniformly bounded by $\frac12(x+1)\leq f(x)\leq \frac34(x+1)$ (see Figure \ref{fig2plots}, left). 	
Theorem \ref{thm-main} then shows that, whenever $x>0$ and $f(x)<\alpha$,
with $x,\alpha$ defined by\footnote{The ratio $\alpha$ is unrelated to the parameter in \eqref{ZBMT}. Instead, the variables $(x,\alpha)$ are the coefficients of the line $\ell_f(\v)$ defined by  $L_{b,w}(\v)=2\kappa Q (w-\alpha b+x)=0$ in \eqref{BMTineq}
for the class $\v=[-1,0,Q H^2/\kappa,-m]$.} 
\be
\label{defxa}
x=\frac{Q}{\kappa}\,, 
\qquad 
\alpha=-\frac{3m}{2Q}\,,
\ee
the stable pair invariant $\PT(Q,m)$ can be expressed in terms of invariants 
$\PT(Q',m')$ with $Q'<Q$ {\it and} Abelian D4-D2-D0 invariants $\Omega_{1,\mu}(\hq_0)$. 
More precisely,\footnote{Translating the formula \eqref{pt-thm} to the notations of this section, 
one finds that the index of the Abelian invariant should be $\mu=Q'-Q-\kappa$. 
Then we used the invariance of $\Omega_{1,\mu}(\hq_0)$ under shifts of $\mu$ by $\kappa$ and the flip of the sign to get \eqref{thmS11}.}
\be
\label{thmS11}
\PT(Q,m) =   \sum_{(Q', m')} (-1)^{\chi(Q', m')}
\chi(Q', m') \, \PT(Q',m') \, \Omega_{1,Q-Q'}(\hat q'_0)\,,
\ee 	
where on the right-hand side
\be
\begin{split}
\chi(Q', m') =&\,
m -m'+ Q+Q'-\chiOD \, ,
\\
\hat q'_0 =&\, m'-m -\frac{1}{2\kappa}\(Q'-Q \)^2-\hf\, (Q+Q')+\frac{\chi(\cD)}{24}\,,
\end{split}
\label{defq0hatp}
\ee
with $\chiOD$ and $\chi(\cD)$ defined in \eqref{defL0} and \eqref{defchiD}, respectively.
The sum runs over pairs of integers $(Q', m')$ such that
\bea
\label{Qpbound}
0 \leq & Q' & \leq Q + \kappa \left(\frac12-\alpha\right),
\\
\label{mpbound}
-\frac{Q'^2}{2\kappa} - \frac{Q'}{2}
\leq& m' &\leq  m+\frac{1 }{2\kappa}\, (Q-Q')^2 
+ \frac{1}{2}\,(Q+Q') \, .
\eea
Note that the lower bound on $Q'$ simply follows from vanishing of PT invariants for negative degrees, while
the upper bound on $m'$ similarly corresponds to vanishing of Abelian invariants for charges spoiling \eqref{qmax}.
On the other hand, the upper bound on $Q'$ implies that $Q' < Q$, since $\alpha>f(x)>1/2$ which shows that \eqref{thmS11} has a recursive nature.

\begin{figure}[h]
\begin{center}
\includegraphics[height=5.5cm]{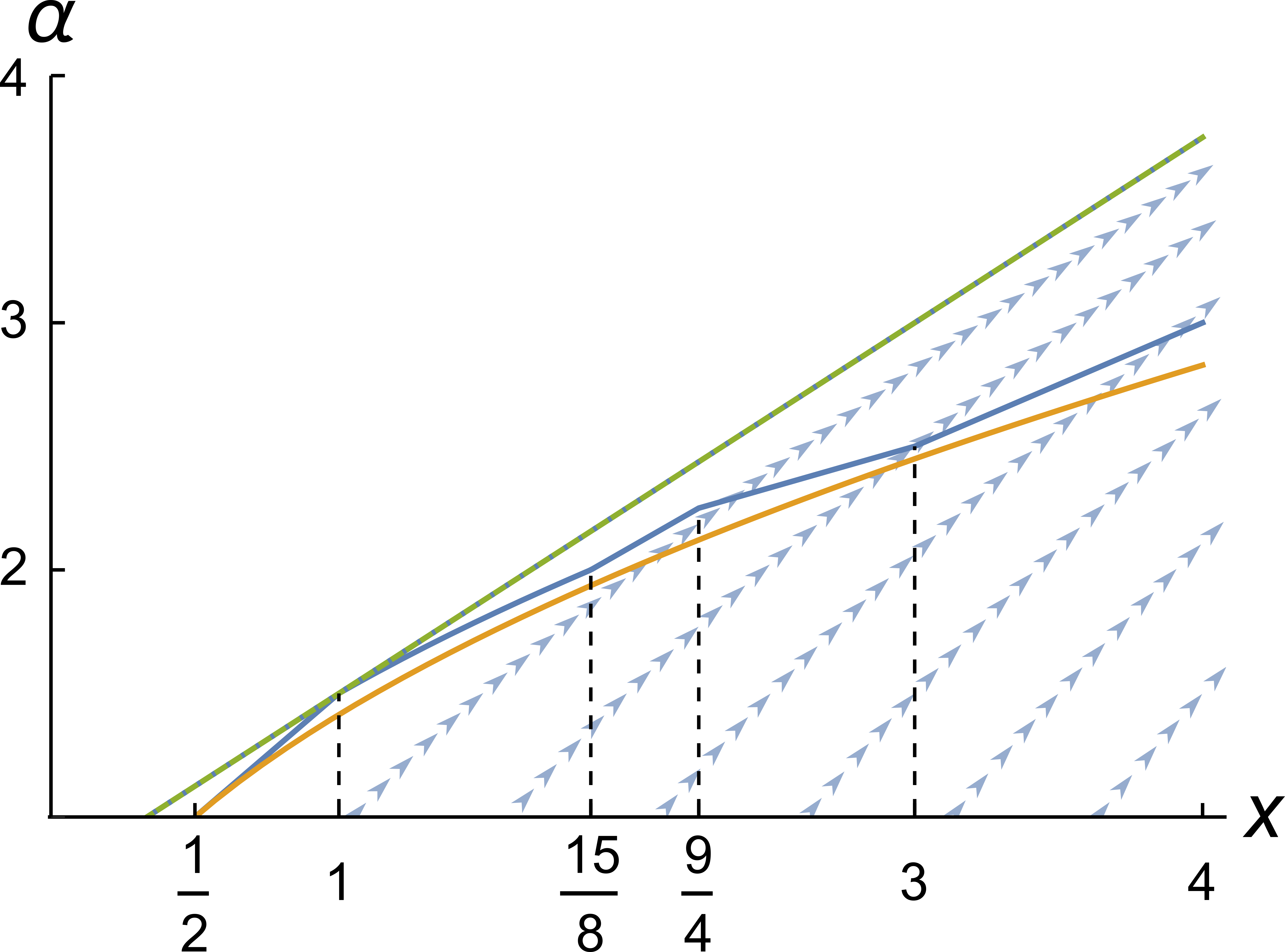}\hspace*{1cm}
\includegraphics[height=5.5cm]{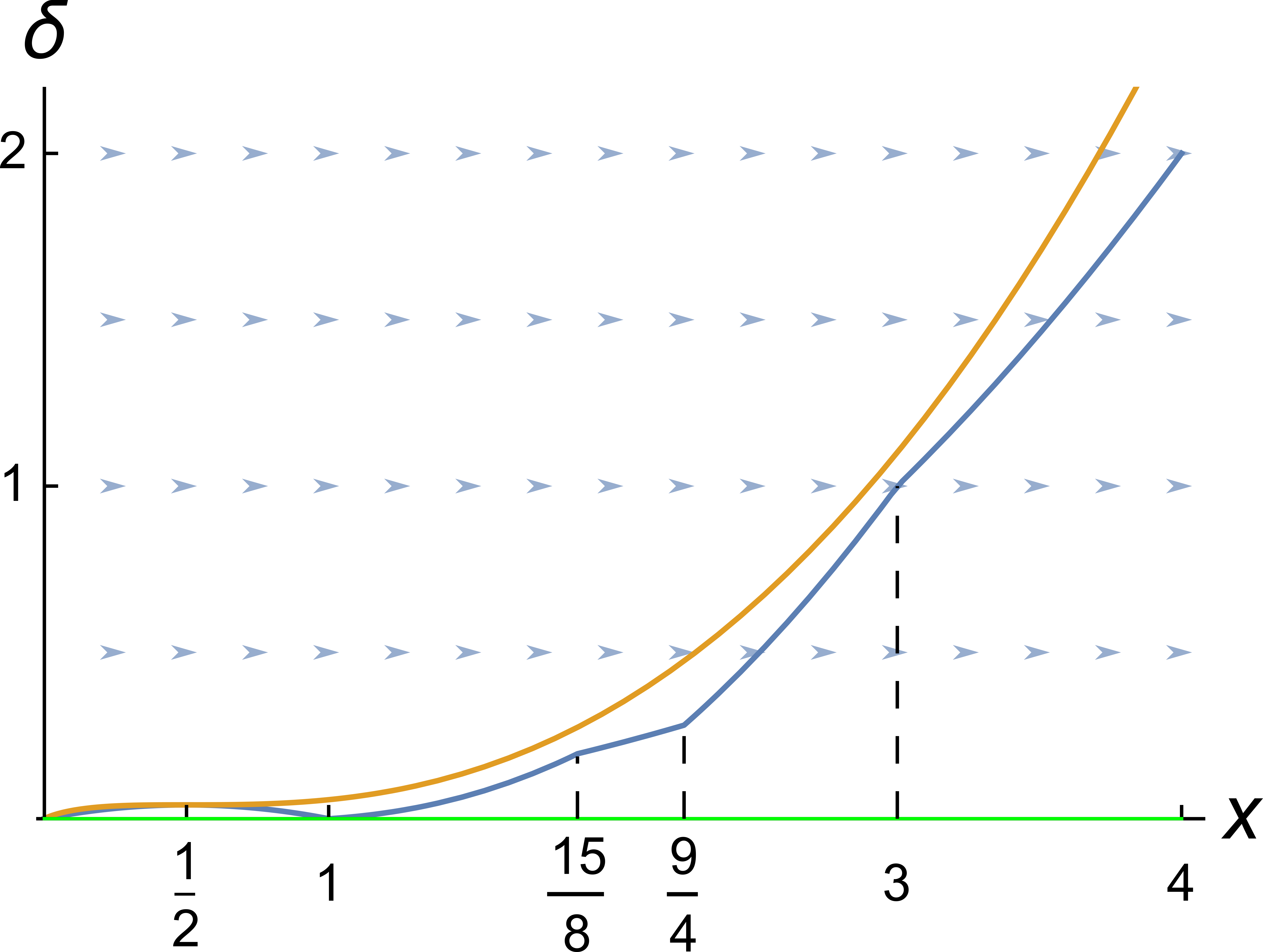}
\end{center}
\caption{Left: The blue line represents the curve $f(x)$, above which Theorem \ref{thm-main} applies. The green line is the Castelnuovo line $\alpha=\frac34(x+1)$, above which PT and DT invariants vanish. Below the orange line $\alpha=\sqrt{2x}$, the BMT line \eqref{BMTineq} does not intersect the parabola $w=\frac12 b^2$ in the $(b,w)$ plane, so the argument in \S\ref{sec_proof} fails.
The dotted lines, oriented to the right, indicate the trajectories induced by spectral flow $(Q,m)\mapsto(Q+\kappa k,m-Q k-\frac12\kappa  k( k+1))$. Right: Same diagram in the $(x,\delta)$ plane, where $\delta=\frac{m}{\kappa}+\frac12 x(x+1)$. In these coordinates, Theorem \ref{thm-main} applies when $(x,\delta)$ lies below the blue line $\delta=-\frac23 x f(x)+\frac12 x(x+1)$. The Castelnuovo line,  below which PT and DT invariants vanish, is the horizontal axis while the trajectories induced by spectral flow are now horizontal lines, oriented to the right.  }
\label{fig2plots}
\end{figure}

Mathematically, the equality \eqref{thmS11} follows by collecting the contributions from all walls for the Chern vector $\v_0=(-1,0,\beta,-m)$, between an empty chamber provided by the BMT inequality and the large volume limit
$w\to\infty$ where the index $\Omega_{b,w}(\v)$ coincides with $\PT(Q,m)$.
Schematically, the formula \eqref{thmS11} says that anti-D6-brane bound to $(Q,m)$ D2-D0-branes 
arises from bound states of anti-D6-branes bound to $(Q',m')$ D2-D0-branes and carrying $-1$ unit of D4-brane flux, and D4-brane bound to $(Q-Q',m-m')$ D2-D0-branes.
The relation \eqref{thmS11} in principle gives a recursive way of computing the PT invariants if Abelian D4-D2-D0 invariants are known, with the caveat that the terms $(Q',m')$ contributing to the sum may not satisfy the condition
$f(x')<\alpha'$. 

A crucial observation is that the term $(Q',m')=(0,0)$ with $\PT(0,0)=1$ always contributes to the sum \eqref{thmS11}, so one may invert this relation to extract the Abelian D4-D2-D0 invariant $\Omega_{1,Q}(\hq_0)$,
where $m$ should now be seen as a function of $Q$ and $\hq_0$ obtained by setting 
$Q'=0,m'=0$ in \eqref{defq0hatp}, 
\be
m(Q,\hq_0) = \frac{\chi(\cD)}{24}-\hq_0 - \frac{Q^2}{2\kappa} - \frac{Q}{2} \, .
\label{relmhatq}
\ee
As above, the resulting formula may be used to recursively compute Abelian D4-D2-D0 invariants in terms of PT invariants, with the same caveat.

In practice, however, the condition $f(x)<\alpha$ is typically not satisfied for the charges of interest.
Indeed, to compute the generating functions \eqref{defhDT},  we are interested in $Q\in [0,\kappa/2]$
and it is easy to see that for such small $Q$, in the best case, the condition is satisfied only 
for D0-brane charges very close to the bound \eqref{qmax}.
Fortunately, we can always use the spectral flow invariance to make $Q$ large enough so that the condition becomes satisfied.
Indeed, for $Q\ge 3\kappa$ the condition $f(x)<\alpha$ can be rewritten as 
\be
\frac{\chi(\cD)}{24} - \hq_0 < 
\frac{Q^2}{6\kappa} - \frac{Q}{6} \, ,
\label{condQq0}
\ee
and is clearly satisfied if $Q$ is sufficiently large. 

Thus, we arrive at the following recipe. 
To compute $\Omega_{1,\mu}(\hq_0)$, 
let us choose $ k\in \IZ_+$ such that 
\be
\label{condfeps}
f(Q_k)<-\frac{3m_k}{2Q_k},
\qquad \mbox{where}\ \
m_k=m(Q_k,\hq_0),
\quad
Q_k=\mu+\kappa  k.
\ee
Then the Abelian index is given by the following formula
\be
\label{thmS11inv}
\Omega_{1,\mu}(\hq_0) = \frac{(-1)^{m_k+Q_k-\chiOD}}{m_k+Q_k-\chiOD}
\left[ \PT(Q_k,m_k) 
-  \!\!\sum_{(Q',m')\neq (0,0)} \!\!
(-1)^{\chi(Q', m')}\chi(Q', m') \, \PT(Q',m') \, 
\Omega_{1,Q_k-Q'}(\hat q'_0) \right],
\ee
where one should apply \eqref{defq0hatp}-\eqref{mpbound} with $(Q,m)$ replaced by $(Q_k,m_k)$. For practical computations, it is of course convenient to choose the minimal possible value of $ k$ satisfying \eqref{condfeps}, because PT invariants are usually known  for small degrees $Q$ only. 

Before we proceed in the next section to apply this result to the CY threefolds listed in Table \ref{table1},
we spell out two important consequences of the formula \eqref{thmS11}, which are also proven in Appendix \ref{sec_appS} (see \S\ref{sec_appII} and \S\ref{sec_appI}).

\subsubsection*{Castelnuovo bound}

As a consequence of the wall structure established for the proof of Theorem \ref{thm-main}, and using induction on $Q$, one obtains a Castelnuovo-type inequality for PT invariants: namely, for any $(Q,m)\in \IZ_+\times \IZ$, 
$\PT(Q,m)=0$ unless
\be
\label{mCastC}
m \geq 
-\cC(Q), \qquad  \cC(Q)\coloneqq 
\begin{cases}
 \left\lfloor \frac{Q^2}{2\kappa}+ \frac{Q}{2} \right\rfloor & Q\geq \kappa, \\[3mm]
  \left\lfloor \frac{2Q^2}{3\kappa}+ \frac{Q}{3} \right\rfloor & 0<Q<\kappa.
\end{cases}
\ee
As a result, we can replace the lower bound in \eqref{mpbound} by $-\cC(Q') \leq m'$, as stated in Appendix \ref{sec_appS}.
By the DT/PT relation, \eqref{mCastC} implies the same statement for the DT invariant
$\DT(Q,m)$, while the PT/GV relation implies the Castelnuovo bound for GV invariants in \eqref{gCast}. Note that in terms of $(x,\alpha)$ defined in \eqref{defxa}, the bounds in \eqref{mCastC} take a universal form independent of $\kappa$:
\be
\alpha\leq 
\begin{cases}
\frac34\, (x+1) & x\geq 1, 
\\
x+\frac12 & x\leq 1. 
\end{cases}
\ee

Since \eqref{thmS11} provides a way to compute PT invariants in the range $\frac12 x+1 < \alpha \leq \frac34(x+1)$ (assuming that $Q=\kappa x$ is large, for the sake of argument) in terms of PT invariants of lower degree, it follows that for fixed degree $Q$, the number of unknown GV invariants is effectively reduced from $Q^2/(2\kappa)$ to $Q^2/(3\kappa)$. Fixing instead the genus $g$, the number of constraints on holomorphic ambiguities from known GV invariants now grows as $\sqrt{3\kappa g}$, rather than $\sqrt{2\kappa g}$, therefore allowing to fix them up to genus $g\apprle \frac34 \kappa\rho^2$ rather than $g\apprle \frac12\kappa\rho^2$ (see footnote \ref{foothumb}). Thus, we expect that the additional
constraints from \eqref{thmS11} will allow to push the direct integration method 
to genus $g_{\rm mod} \simeq \frac32 g_{\rm integ}$, i.e. a factor $3/2$ higher than the maximal genus predicted by \eqref{Kgmax}. Unfortunately, this reasoning overlooks the complicated relation between PT and GV invariants, and in practice the gain in genus will be slighter smaller (see the last column in Table \ref{table1}).

Returning to the prescription \eqref{condfeps}, we note that the distance away from the Castelnuovo bound \eqref{mCastC} is independent of $ k$,
\be
m(Q_k) + \frac{Q_k^2}{2\kappa} + \frac{Q_k}{2} = \hq_0-\frac{\chi(\cD)}{24}\, .
\ee
In Figure \ref{fig2plots} (right), we represent the region of validity of Thm 1 in the plane $(x,\delta)$ where $\delta=\frac{m}{\kappa}
+\frac{Q^2}{2\kappa^2} + \frac{Q}{\kappa}$, where spectral flow acts by horizontal translations $x\mapsto x- k$, keeping $\delta$ fixed. This makes it clear that Thm 1 is always valid for $ k\geq  k_0$ large enough. Experimentally, we shall see in \S\ref{sec_test} and \S\ref{sec_gen} that the formula  \eqref{thmS11inv} often gives the correct result for $ k= k_0-1$ or (less often) $ k= k_0-2$, even though the assumptions
of Thm 1 are no longer satisfied, see in particular Figure \ref{figX5}
for the quintic threefold.

\subsubsection*{Optimal case}
\label{sec:optimalcase}

The formula \eqref{thmS11} becomes particularly simple in cases where 
the sum over 
$(Q',m')\neq (0,0)$ becomes empty. This occurs provided $(Q',m')=(0,0)$ is the only solution to \eqref{Qpbound}, \eqref{mpbound} --- in this case we call the pair $(Q,m)$ optimal.  A sufficient set of conditions is that 
\be
\label{optimalsufficient}
\alpha>f(x),
\qquad 
\Psi\left(x,1/\kappa,\alpha\right)<0, 
\qquad 
\Psi\left(x,x-\alpha+\tfrac12,\alpha\right)<0,
\ee
where 
\be
\label{defPsi}
\Psi(x,x',\alpha) :=  \frac12\,(x-x')^2 + \frac12\,(x+x') -\frac23\,\alpha x + \frac12 \,x'^2 + \frac12\, x'
\ee
is the difference between the upper and lower bounds in \eqref{mpbound},
after expressing the result in terms of $x=Q/\kappa$ and $x'=Q'/\kappa$ and rescaling by $\kappa$. The values
$x'=1/\kappa$ and $x'=x-\alpha+\frac12$ correspond to the minimal and maximal values $Q'=1$ and $Q'=\kappa(\frac12-\alpha)$ to be ruled out. 
The condition $\Psi\left(x,1/\kappa,\alpha\right)<0$ can be equivalently written as 
\be
\label{Psicond1}
\frac34(x+1) - \frac{3(x-1)}{2x\kappa} + \frac{3}{2\kappa^2 x}  < \alpha .
\ee
This shows that the condition \eqref{Psicond1} can only be satisfied 
close to the Castelnuovo bound. The last condition in \eqref{optimalsufficient} turns out to be implied by the condition $\alpha>f(x)$ when $x>\frac14(5+\sqrt7)\simeq 1.91$. When the conditions in \eqref{optimalsufficient} are obeyed (or more generally when $(Q,m)$ is optimal), \eqref{thmS11} simply reduces to 
\be
\label{ptOm1}
\PT(Q,m)= (-1)^{m+Q-\chiOD} \Bigl(m+Q-\chiOD\Bigr)\,\Omega_{1,Q}(\hq_0) \, .
\ee

In fact, using the invariance of $\Omega_{1,\mu}(\hq_0)$ under spectral flow as above, 
one can always choose the spectral flow parameter $ k$ large enough such that 
$(Q_k,m_k)$ is optimal. In particular, this implies that the ratio
\be
\label{ratiostab}
\Omega_{1,Q}(\hq_0) = 
(-1)^{m_k+Q_k-\chiOD} \frac{\PT(Q_k,m_k)}{m_k+Q_k-\chiOD}\,,
\ee
must stabilize to a constant value for $ k$ larger than a suitable $ k_1\geq  k_0$.

It is also possible to use these relations to derive general formulae for GV invariants near the Castelnuovo bound. 
Let us choose, for example, $Q=m=0$. Then using \eqref{ptOm1} and \eqref{eqFeyzus} with $r=1$, we find that the optimality condition is satisfied for any $ k\geq 2$, leading to\footnote{This result reduces to the first part of Theorem \ref{thm-quintic} for the quintic upon setting
$\mu=\kappa k, \kappa=\chi_\cD=5$.}
\be
\PT\left(\kappa  k, -\tfrac12 \,\kappa  k( k+1) \right) = 
(-1)^{\frac12 \kappa  k( k-1)} \left(\chiOD + \tfrac12\, \kappa  k( k-1) \right)
\chiOD\, .
\ee
Since the second argument on the left-hand side is equal to 
$1-g_C(\kappa k)$, one can use the relation \eqref{PTmax} to obtain the GV invariant for $Q=\kappa  k$ and maximal genus,
\be
\GVg[\kappa  k]{\gmax(\kappa  k)} = (-1)^{\hf\kappa  k( k-1)} \(\chiOD + \tfrac12\,\kappa  k( k-1)\)\chiOD\, .
\ee
This reproduces the result which was obtained by heuristic arguments in \eqref{eqn:gvgmax} for $d= k\geq 2$. 
	
Similarly, choosing $Q=0$ and $m=1$, we find\footnote{By a case-by-case analysis, one checks that the optimality conditions are verified when  $ k\geq 2$ for $X_{4,3},X_{3,3},X_{4,2},X_{3,2,2},X_{2,2,2,2}$, 
when $ k\geq 3$ for $X_8,X_{4,4},X_{6,2}$, and when  $ k\geq 4$
for $X_{10}$ and $X_{6,6}$.}
\be
\PT\left(\kappa  k, -\tfrac12 \,\kappa  k( k+1) +1 \right) = 
(-1)^{\frac12 \kappa  k( k-1)+1} \left(\chiOD + \tfrac12 \,\kappa  k( k-1) -1 \right) J_1\, ,
\label{eqn:PTJ1}
\ee
where $J_1:=(-1)^{\chiOD+1}\,\bOm_{1,0}\left(\frac{\chi(\cD)}{24}-1\right)$. 
Using \eqref{PTsubmax}, we conclude that the GV invariant for $Q=\kappa k$ and submaximal genus is given by
\be
\begin{split}
\GVg[\kappa  k]{\gmax(\kappa  k)-1} =&\, (-1)^{\frac12 \kappa  k( k-1)}  
\[-\hf\, \kappa^2  k^4\chiOD
-\hf\, \kappa  k^2\(2 \chiOD^2-\kappa\chiOD  +J_1\)
+\kappa  k \(\frac{J_1}{2}-\chiOD^2\)+J_1(1-\chiOD ) \] .
\end{split}
\ee
This reproduces the result obtained by heuristic arguments in \eqref{eqn:gvgmaxm1}, although the constant $J_1$ is not determined by the present computation. In the examples in \S\ref{sec_test} and \S\ref{sec_gen}, we shall see that $J_1=\chi_\CY (\chi_{\cD}-1)$ when the divisor $\cD$ is smooth, which is the case when $\chiOD\geq 4$, but that it may otherwise differ from this value.

\section{Testing the modularity of rank 0 DT invariants}
\label{sec_test}

In this section, we apply the results explained in \S\ref{sec_wcrm1} to determine the generating series of Abelian D4-D2-D0 indices for several examples of one-parameter threefolds, including $X_5$ (the quintic in $\IP^4$), $X_{10}$ (the decantic in weighted projective space 
$\IP^4_{5,2,1,1,1}$ and $X_{4,2}$ (a complete intersection of degree $(4,2)$ in $\IP^5$). In particular, we rigorously compute the polar coefficients and a large number of non-polar coefficients, and confirm the modularity property predicted by string theory. 
For $X_5$ our results coincide with those in \cite{Gaiotto:2006wm},
for $X_{10}$ we confirm the proposal in \cite{VanHerck:2009ww}
(which deviates from the original computation in \cite{Gaiotto:2007cd}), 
while for $X_{4,2}$ we determine the generating series that was previously unknown.
In Appendix \ref{sec_gen}, we give similar results for all other hypergeometric models, except for $X_{3,2,2}$ and $X_{2,2,2,2}$ for which our current knowledge of GV invariants is still insufficient to determine (or just guess) the polar terms.

\subsection{Basis of vector-valued modular forms}
\label{sec_basis}

As explained in \S\ref{sec_modconj}, string theory predicts
that the generating series \eqref{defhDT} of Abelian D4-D2-D0 indices (for brevity we drop the rank index $1$)
\be
h_\mu(\tau) =\sum_{\hq_0 \leq \frac{\chi(\cD)}{24}}
\Omega_{1,\mu}(\hq_0)\,\q^{-\hq_0 }
\label{defhDT1}
\ee
should behave under $SL(2,\IZ)$ transformations as a vector-valued modular form of weight $-3/2$, transforming in the Weil representation of the lattice $\IZ[\kappa]$. The space  $\scM_1(\CY)$
of such functions has dimension $n_1^p-n_1^c$, where
$n_1^p$ is the number of polar coefficients, corresponding to terms with negative power $\hq_0>0$ in \eqref{defhDT1},
and $n_1^c$ is the number of linear relations which these coefficients must satisfy, in order for a modular form to exist (the numbers $n_1^p$ and $n_1^c$ are listed in Table \ref{table1}). 

An overcomplete basis of $\scM_1(\CY)$
can be constructed as follows \cite{Alexandrov:2022pgd}.
We define the theta series
\be
\label{theta}
\vths{\kappa}_{\mu}(\tau)=
\begin{cases}
(-1)^{\mu+\frac{\kappa}{2}}\, \sum_{{k}\in \IZ+\frac{\mu}{\kappa}+\frac{\kappa}{2}}
\q^{\tfrac{\kappa}{2}\,k^2}  & \kappa\ \mbox{even,}
\\
 -\I \kappa
\sum_{{k}\in \IZ+\frac{\mu}{\kappa}+\frac{\kappa}{2}} 
(-1)^{\kappa k}\,k\, \q^{\tfrac{\kappa}{2}\,k^2}  & \kappa\ \mbox{odd.}
\end{cases} 
\ee
They satisfy 
\be
\label{thsym}
\vths{\kappa}_{\mu}(\tau)= \vths{\kappa}_{\mu+\kappa}(\tau)= \vths{\kappa}_{-\mu}(\tau)
\ee
and transform under $\tau\mapsto \frac{a\tau+b}{c\tau+d}$
as vector-valued modular forms of weight 1/2 and 3/2, respectively.
For $\kappa=1$, we note that $\vths{1}_{\mu}(\tau)=\eta^3$ where $\eta(\tau)$ is the Dedekind theta function. More generally, for $\mu=0$ one has
\be
\vths{\kappa}_{0}(\tau)=
\begin{cases}
2 (-1)^{\frac{\kappa}{2}} \displaystyle{\frac{\eta(2\kappa\tau)^2}{\eta(\kappa\tau)}} & \kappa \mbox{ even},
\\
\kappa  (-1)^{\frac{\kappa-1}{2}} \eta(\kappa\tau)^3
& \kappa\mbox{ odd}.
\end{cases}
\ee
We claim that any element of $\scM_1(\CY)$ is a linear combination of the form
\be
h_\mu(\tau) = 
\sum_{\ell=0}^{\mm}
\sum_{k=0}^{k_\ell} a_{\ell,k}\, 
E_4^{\lfloor w_\ell/4 \rfloor-\eps_\ell-3k}(\tau)\, E_6^{2k+\eps_\ell}(\tau)\, 
\frac{D^\ell \ths{\kappa}_{\mu}(\tau)}{\eta^{4\kappa + c_2}(\tau)}\, ,
\label{decomp-modform}
\ee
where 
$E_4(\tau)$ and $E_6(\tau)$ are the standard Eisenstein series, and $D^{\ell}$ is the iterated Serre derivative\footnote{Rather than the standard iterated Serre derivative, one can just as well use its improved version introduced in  \cite[Eq (35)]{rodriguez1993square} or Rankin-Cohen brackets. Unfortunately this does not lead to smaller denominators in the resulting coefficients $a_{\ell,k}$.},
acting on holomorphic modular forms of weight $w$
through  $D_w = \q \partial_{\q} - \frac{w}{12} E_2$, where $E_2$ is the normalized quasi-modular Eisenstein series. Finally, the integers 
$k_\ell, \eps_\ell, w_\ell$ in \eqref{decomp-modform} are given by
\be
\begin{split}
& k_\ell=\lfloor w_\ell/12 \rfloor-\delta^{(12)}_{w_\ell-2},
\qquad
\eps_\ell=\delta^{(2)}_{w_\ell/2-1}, 
\\
&
w_\ell=2\kappa+\hf\, c_2-3-2\ell+\delta^{(2)}_\kappa,
\end{split}
\ee 
where $\delta^{(n)}_x$ is equal to 1 if $x=0\!\mod n$ 
and 0 
otherwise, and $\mm$ should be chosen sufficiently large so that
$\sum_{\ell=0}^{\mm}(k_\ell+1)$ is not smaller than the 
dimension of the space $\scM_1(\CY)$.
The coefficients $a_{\ell,k}$ are not unique in general (since the basis is overcomplete), but the modular form $h_\mu(\tau)$ is uniquely fixed by providing $n_1^p-n_1^c$ of its Fourier coefficients (for example the polar coefficients). 
Having determined a suitable set of coefficients $a_{\ell,k}$, it is then straightforward to expand
$h_\mu(\tau)$ to arbitrary order, and obtain a prediction for an infinite number of Abelian D4-D2-D0 invariants.

\subsection{$X_5$}
\label{subsec-X5}

\begin{figure}[h]
\begin{center}
\includegraphics[height=5.5cm]{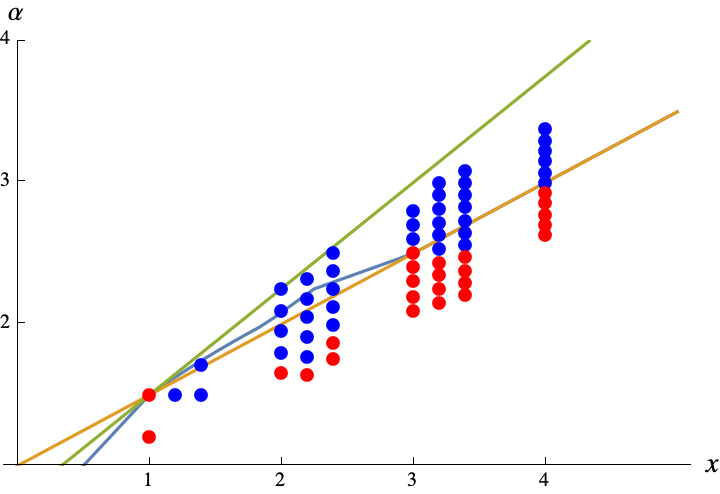}\hspace*{5mm}
\includegraphics[height=5.5cm]{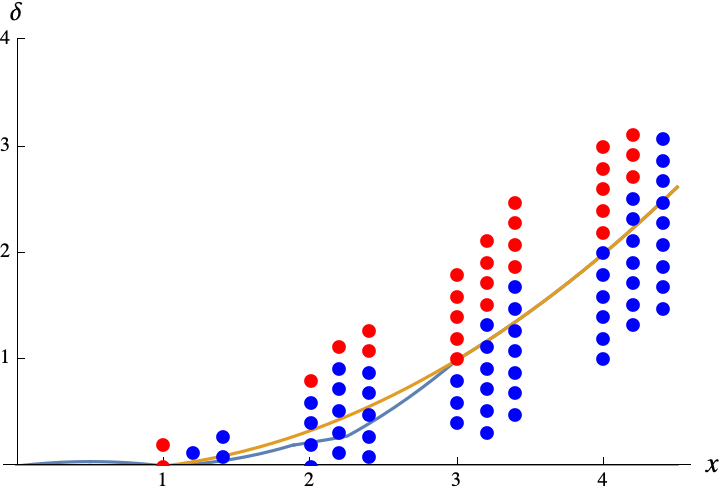}
\end{center}
\caption{Left: Blue dots indicate values of $(x,\alpha)$ for which the formula \eqref{thmS11inv} turns out to give the correct D4-D2-D0 indices for $X_5$. Red dots instead indicate values of $(x,\alpha)$ for which   \eqref{thmS11inv} fails to give the correct result. The values of $(x,\alpha)$ correspond to $(Q_k,m_k)$ with $k\in\{k_0,k_0-1,k_0-2\}$ where $k_0$ is the minimal value of $k$ such that \eqref{condfeps} holds. 
All red dots lie below the line $\alpha=f(x)$ shown in blue, and in fact they all lie below the line $\alpha=\frac{x}{2}+1$ shown in orange (see Remark \ref{remark stronger BG in}). Interestingly, there are also some blue dots lying below this line, which indicates that the condition $f(x)<\alpha$ for the validity of  \eqref{thmS11inv} can probably be weakened. Right: same diagram in the 
$(x,\delta)$ plane.}
\label{figX5}
\end{figure}

\begin{table}
\be
\begin{array}{|r|r|rrrrrr|}
\hline
Q & g_C & \delta=0 &  \delta=1 &  \delta=2 &  \delta=3 &  \delta=4 & \delta=5\\
\hline
 1 & 1 & 0 & 2875 & \text{} & \text{} & \text{} & \text{} \\
 2 & 2 & 0 & 0 & 609250 & \text{} & \text{} & \text{} \\
 3 & 3 & 0 & 0 & 609250 & 317206375 & \text{} & \text{} \\
 4 & 4 & 0 & 8625 & 534750 & 3721431625 & 242467530000 & \text{} \\
 5 & 6 & 10 & 1100 & 49250 & -15663750 & 75478987900 & 12129909700200 \\
 6 & 7 & 0 & -34500 & -3079125 & -7529331750 & 3156446162875 & 871708139638250 \\
 7 & 9 & 0 & 0 & 4874000 & 1300955250 & -1917984531500 & 245477430615250 \\
 8 & 11 & 0 & 0 & -6092500 & -1670397000 & 2876330661125 & -471852100909500 \\
 9 & 13 & 0 & 60375 & 5502750 & 18763368375 & -12735865055000 & 1937652290971125 \\
 10 & 16 & -50 & -5700 & -286650 & 50530375 & -454092663150 & 150444095741780 \\
 11 & 18 & 0 & -86250 & -7357125 & -29938013250 & 22562306494375 & -4041708780324500 \\
 12 & 21 & 0 & 0 & -13403500 & -3937166500 & 8725919269125 & -2017472506595500 \\
 13 & 24 & 0 & 0 & -15840500 & -4638330000 & 10690009494250 & -2578098061480250 \\
 14 & 27 & 0 & -138000 & -10177500 & -52227066000 & 42752384997625 & -8759526658670500 \\
 15 & 31 & -100 & -9200 & -342400 & 136695125 & -1214106563650 & 484402370601245 \\
 16 & 34 & 0 & 181125 & 11178000 & 70714095125 & -60120995398500 & 13182681427726625 \\
 17 & 38 & 0 & 0 & -28025500 & -7761538500 & 20623428936750 & -5693356905665000 \\
 18 & 42 & 0 & 0 & 31681000 & 8578113250 & -23636174920000 & 6726357908107750 \\
 19 & 46 & 0 & -258750 & -10246500 & -103897578000 & 92567501962875 & -22247603793898250 \\
 20 & 51 & 175 & 9700 & 113650 & -271460000 & 2362533313525 & -1059131220525950 \\
 21 & 55 & 0 & 319125 & 7158750 & 129691149375 & -118821918509250 & 30276261813046500 \\
 22 & 60 & 0 & 0 & 48740000 & 11680440750 & -37863219131500 & 12130764520281750 \\
 23 & 65 & 0 & 0 & 53614000 & 12356541750 & -41972283930000 & 13849264699781000 \\
 24 & 70 & 0 & 422625 & -2829000 & 174040666500 & -165847969399750 & 46048552308175750 \\
 25 & 76 & 275 & 1950 & -261225 & -437171250 & 3908290893900 & -1955377337896550 \\
 26 & 81 & 0 & -500250 & 14145000 & -207540563250 & 202764143836375 & -59568660504287750 \\
 27 & 87 & 0 & 0 & 75547000 & 14268228250 & -60651049880500 & 22529431755767500 \\
 28 & 93 & 0 & 0 & -81639500 & -14474860500 & 65883050745250 & -25213918522757500 \\
 29 & 99 & 0 & 629625 & -40175250 & 264127092375 & -267394402192000 & 85401556513695875 \\
 30 & 106 & -400 & 20800 & -111400 & 631692625 & -5861793912900 & 3278134921975475 \\
 \hline
\end{array}
\nonumber
\ee
\caption{GV invariants $\GVg{g_C(Q)-\delta}$ for $X_5$, assuming modularity.}
\vspace{-0.5cm}
\label{table_GVX5}
\end{table}

\begin{table}
\be
\begin{array}{|r|r|rrrrrr|}
 \hline
Q & m_C & \delta=0 &  \delta=1 &  \delta=2 &  \delta=3 &  \delta=4 & \delta=5\\
\hline
 1& 0 & 0 & 2875 & -5750 & 8625 & -11500 & 14375 \\
 2&-1 & 0 & 0 & 609250 & 2912875 & -14703500 & 38888250 \\
 3&-2 & 0 & 0 & 609250 & 317206375 & 1117181000 & -2098275750 \\
 4&-3 & 0 & 8625 & 569250 & 3722552875 & 244219693000 & 609122565875 \\
 5& -5 & 10 & 1200 & 58500 & -15336250 & 75441932225 & 12282361758020 \\
 6& -6 & 0 & -34500 & -3395375 & -7552124750 & 3111341190625 & 884181641560000 \\
 7& -8 & 0 & 0 & 4874000 & 1359443250 & -1904746390000 & 230184283873875 \\
 8& -10 & 0 & 0 & -6092500 & -1767877000 & 2852214003125 & -437477532060500 \\
 9 & -12 & 0 & 60375 & 6831000 & 18887370000 & -12396985924250 & 1736738444379375 \\
 10 &-15 & -50 & -7200 & -468000 & 40719875 & -452993138850 & 140467307991350 \\
 11& -17 & 0 & -86250 & -10117125 & -30201650750 & 21720393561500 & -3466439656488000 \\
 12&-20 & 0 & 0 & -13403500 & -4419692500 & 8583611403125 & -1740547789348750 \\
 13&-23 & 0 & 0 & -15840500 & -5303631000 & 10490837623750 & -2175677447038750 \\
 14&-26 & 0 & -138000 & -17077500 & -52884636000 & 40335754941625 & -6932684543525000 \\
 15&-30 & -100 & -15200 & -1053000 & 98891125 & -1207584961600 & 421450499252120 \\
 16&-33 & 0 & 181125 & 22770000 & 71772279125 & -55849465988500 & 9821365434297875 \\
 17&-37 & 0 & 0 & -28025500 & -9723323500 & 20027962736250 & -4351425496412500 \\
 18&-41 & 0 & 0 & 31681000 & 11049231250 & -22889100270000 & 5004138750546250 \\
 19&-45 & 0 & -258750 & -33016500 & -105769272000 & 83774260263375 & -15020907593198000 \\
 20&-50 & 175 & 27200 & 1930500 & -186148000 & 2339695863100 & -842900254597650 \\
 21&- 54 & 0 & 319125 & 40986000 & 132211590000 & -105493513413000 & 19064998024136500 \\
 22& -59 & 0 & 0 & 48740000 & 17236800750 & -36241075427500 & 8050136250878750 \\
 23 & -64 & 0 & 0 & 53614000 & 19004677750 & -40055925472500 & 8920421250973750 \\
 24 & -69 & 0 & 422625 & 54648000 & 177541278000 & -142726517485750 & 25997724680535000 \\
 25 & -75 & 275 & 43200 & 3100500 & -302490500 & 3849177065100 & -1404833757662750 \\
 26 & -80 & 0 & -500250 & -64894500 & -211538544000 & 170651270906875 & -31197269617418250 \\
 27 & -86 & 0 & 0 & 75547000 & 26960124250 & -57222750675000 & 12836703751401250 \\
28 & - 92 & 0 & 0 & -81639500 & -29169970500 & 61991313231250 & -13924560001520000 \\
29 & - 98 & 0 & 629625 & 81972000 & 268200654000 & -217192526608750 & 39863177843487000 \\
30 & -105 & -400 & -63200 & -4563000 & 447918625 & -5736028567600 & 2107250636494125 \\
  \hline
\end{array}
\nonumber
\ee
\caption{Stable pair invariants $\PT(Q,m_C(Q)+\delta)$ for $X_5$, assuming modularity.}
\vspace{-0.5cm}
\label{table_PTX5}
\end{table}

Abelian D4-D2-D0 invariants for the quintic threefold were first studied in \cite{Gaiotto:2006wm}, using a different basis of modular forms and an ingenuous but non-rigorous method for computing the polar terms. In this case,
$\kappa=5$, $n_1^p=7$ and $n_1^c=0$ so the vector-valued modular form is uniquely determined by computing 7 of its coefficients.
Using the overcomplete basis of the previous subsection, the result of \cite{Gaiotto:2006wm} can be written as 
\be
\begin{split}
h_{\mu}=&\, \frac{1}{\eta^{70}} \[-\frac{222887 E_4^8+1093010 E_4^5 E_6^2+177095 E_4^2 E_6^4}{35831808}
\right.
\\
&\,
+\frac{25 \(458287 E_4^6 E_6+967810 E_4^3 E_6^3+66895 E_6^5\)}{53747712}\, D
\\
&\, \left.
+\frac{25 \(155587 E_4^7+1054810 E_4^4 E_6^2+282595 E_4 E_6^4\)}{8957952}\, D^2\] \vths{5}_{\mu},
\end{split}
\ee
In view of the symmetry properties \eqref{thsym}, there are only three distinct components, with the following expansion: 
\be
\begin{split}
h_{0} =&\, \q^{-\frac{55}{24}}\, \Bigl(
\underline{5 - 800 \q + 58500 \q^2} + 5817125 \q^3 + 75474060100 \q^4 
\\
&\, 
+28096675153255 \q^5+3756542229485475 \q^6+277591744202815875 \q^7  
\\
&+ 13610985014709888750 \q^8 +490353109065219393125 \q^9
+ \dots \Bigr),
\\
h_{1} =&\, \q^{-\frac{55}{24}+\frac35} \, \Bigl(  \underline{0+8625 \q}
- \dotuline{ 1138500 \q^2} + 3777474000 \q^3
+ 3102750380125 \q^4 
\\
&\, 
+577727215123000 \q^5 +52559194851824125 \q^6+2990604504777589125 \q^7
+ \dots\Bigr),
\\
h_{2} =&\, \q^{-\frac{55}{24}+\frac25}\,
\Bigl(  \underline{0+0 \q}  -\dotuline{1218500 \q^2} + 441969250 \q^3 + 953712511250 \q^4 
\\
&\,  
+217571250023750 \q^5+22258695264509625 \q^6 
+1374043315791020500 \q^7
+ \dots\Bigr).
\end{split}
\ee
Here and elsewhere, the polar coefficients are underlined.
Using Eq. \eqref{thmS11inv} and GV invariants up to genus 53, we have reproduced all terms up to (and including) orders $\q^9$, $\q^5$ and $\q^6$ in these expansions, respectively.\footnote{For the coefficients up to $\q^5$ in $h_0$, $\q^2$ in $h_1$ and $\q^3$ in $h_2$, the relevant value of $(Q,m)$ is optimal and the formula \eqref{thmS11} has only one non-vanishing contribution (or none when the coefficient is zero). For the terms of order $\q^6, \q^7,\q^8, \q^9$ in in $h_0$, there are contributions from 2,3,4,5 walls, respectively. For the order $\q^3,\q^4,\q^5,\q^6,\q^7$ in $h_1$, there are contributions from 2,3,4,4, 5 walls, respectively. For the terms of order $\q^4, \q^5,\q^6, \q^7$ in in $h_2$, there are contributions from 2,4,3,3, 5 walls, respectively.} In many cases, we find that \eqref{thmS11inv} holds even though the assumption $f(x)<\alpha$ is not obeyed (see Figure \ref{figX5}), in particular we can also reproduce the coefficients of $\q^6$ in $h_{1}$ and $\q^7$ in $h_2$ using \eqref{thmS11inv} with $k=k_0-1$, where $k_0$ is the minimal value of $k$ for which \eqref{condfeps} is satisfied.

As already noted in \cite{Alexandrov:2022pgd}, the naive Ansatz \eqref{naive} with $r=1$ gives the correct polar terms in this case. In addition, it also correctly predicts the $\cO(\q^2)$ terms in $h_1$ and $h_2$, as indicated with dotted underline. The coefficient of the order $\cO(\q^3)$ term in $h_0$ can be 
understood as 
\be
5817125 = 
-2 \DT(0, 3)+\DT(0,2)+\DT(1,1)^2 \ ,
\ee 
where the first term is the naive ansatz prediction, the second is a correction from the locus where the 3 D0-branes are aligned, and the last term corresponds to a bound state of D6-D2
and $\overline{\rm D6}$-$\overline{\rm D2}$-branes \cite{Collinucci:2008ht}. It would be interesting to have a similar bound state interpretation for other non-polar coefficients. 

Using modularity we can also predict GV and PT invariants of arbitrary degree, provided they are close enough to the Castelnuovo bound. In Table \ref{table_GVX5}, we list the GV invariants with $\delta=g_C(Q)-g\leq 5$, and similarly
in Table \ref{table_PTX5} we list the PT invariants with $\delta=m-m_C(Q)\leq 5$,
where $g_C(Q)$ and $m_C(Q)$ were defined in \eqref{gCast} and \eqref{mCast}, respectively. Using these GV invariants, we have in principle sufficiently many boundary conditions to fix the holomorphic ambiguity up to genus 69. Due to limited computer resources, we have currently pushed up the direct integration to genus 64, and confirmed the predictions of modularity up to that order.

\subsection{$X_{10}$}
\label{subsec-X10}

We now turn to the decantic in $\IP^4_{5,2,1,1,1}$, which was first studied in \cite{Gaiotto:2007cd} and revisited in \cite{VanHerck:2009ww}. In this case, $\kappa=1$,  $n_1^p=2$ and $n_1^c=0$ so the scalar modular form $h:=h_0$ is
uniquely fixed by 2 coefficients. In \cite{Gaiotto:2007cd},
it was suggested that 
\be
\begin{split}
h\stackrel{?}{=}&\, \frac{541 E_4^4+1187 E_4 E_6^2}{576\, \eta^{35}}
\\
=&\, \q^{-\frac{35}{24}}\,\Bigl(
\underline{3-576 \q}+271704 \q^2+206401533 \q^3+ \dots \Bigr).
\end{split}
\label{resfunX10old}
\ee
The same result was found in \cite{Alexandrov:2022pgd}
using the naive Ansatz \eqref{naive}. Instead, assuming that the BMT inequality is satisfied, Eq. \eqref{thmS11inv}
predicts that the coefficient of the subleading polar term should be $-575$, as suggested in \cite{VanHerck:2009ww}.
In fact, using \eqref{thmS11inv} and GV invariants up to genus 47, we can check all the terms up to order $\q^{11}$ in the resulting expansion, 
\be
\begin{split}
h=&\, \frac{203 E_4^4+445 E_4 E_6^2}{216\, \eta^{35}}
\\
=&\, \q^{-\frac{35}{24}}\,\Bigl(
\underline{3-\underline{575} \q}+271955 \q^2+206406410 \q^3+21593817025 \q^4
\\
&\,
+1054724115956 \q^5
+32284130488575 \q^6+712354737460415 \q^7
\\
&\, +12285858824682770 \q^8 + 174458903522212025 \q^9 + 
 2114022561929255740 \q^{10}\\
 &\,  + 22434520426025264925 \q^{11} 
 +212611407819858981640 \q^{12}
+\cdots\Bigr).
\end{split}
\label{resfunX10}
\ee
Here and in Appendix \ref{sec_gen}, the double-underline underscores the fact that the polar coefficient deviates from the naive ansatz \eqref{naive}.
Interestingly, applying Eq.\eqref{thmS11inv} with $k=k_0-1$, where $k_0$ is the minimal value for which \eqref{condfeps} is satisfied, one can reproduce the expansion \eqref{resfunX10} 
to even higher order $\q^{14}$.

\begin{table}
\be
\begin{array}{|r|r|rrrrrr|}
\hline
Q & g_C & \delta=0 &  \delta=1 &  \delta=2 &  \delta=3 &  \delta=4 & \delta=5\\
\hline
  1 & 2 & 3 & 280 & 231200 & \text{} & \text{} & \text{} \\
 2 & 4 & -12 & -1656 & -537976 & 207680960 & 12215785600 & \text{} \\
 3 & 7 & -18 & -2646 & -1057570 & 630052679 & -46669244594 & 1264588024791 \\
 4 & 11 & 27 & 4060 & 1825541 & -1268283512 & 125509540304 & -5611087226688 \\
 5 & 16 & 39 & 5730 & 2814100 & -2139555052 & 244759232792 & -13239429980228 \\
 6 & 22 & -54 & -7507 & -4004506 & 3254742758 & -416588796648 & 25859458639950 \\
 7 & 29 & -72 & -9193 & -5375708 & 4629222449 & -655954806090 & 45976776864607 \\
 8 & 37 & 93 & 10554 & 6910207 & -6280307986 & 981118531775 & -77100442475920 \\
 9 & 46 & 117 & 11320 & 8597590 & -8227101620 & 1413894771755 & -124031731398850 \\
 10 & 56 & -144 & -11185 & -10438670 & 10490492480 & -1979933144970 & 193210634123311 \\
 11 & 67 & -174 & -9807 & -12450166 & 13093396333 & -2709028151150 & 293124778727973 \\
 12 & 79 & 207 & 6808 & 14669923 & -16061324744 & 3635467145440 & -434786567257064 \\
 13 & 92 & 243 & 1774 & 17162672 & -19423381916 & 4798433274180 & -632285283576376 \\
 14 & 106 & -282 & 5745 & -20026330 & 23213797570 & -6242490557180 & 903422424012068 \\
 15 & 121 & -324 & 16235 & -23398840 & 27474114305 & -8018190890070 & 1270440806044980 \\
 \hline
\end{array}
\nonumber
\ee
 \caption{GV invariants $\GVg{g_C(Q)-\delta}$ 
 for $X_{10}$, assuming modularity.}
 \vspace{-0.5cm}
\label{table_GVX10}
\end{table}

As discussed in~\cite{VanHerck:2009ww}, the deviation from the naive ansatz arises because the moduli space of D4-D0 bound states is in general not a bundle over the moduli space of the D0-brane, which is $\CY$ itself. When the 
D0-brane is at a generic position, the requirement that it should belong to the divisor imposes one condition on the defining equation of the divisor.
Since the divisor is the vanishing locus of a linear polynomial in the three homogeneous coordinates of weight one, the moduli space of divisors containing a given generic point on $\CY$ is $\mathbb{P}^1$. However, when the D0-brane lies at the special point where  all homogeneous coordinates of weight one vanish, 
it  no longer imposes any condition on the divisor, whose moduli space is then enhanced to $\mathbb{P}^2$. Hence, the index for a D4-brane bound to a single D0-brane should be \cite{VanHerck:2009ww}
\begin{align}
\label{X10resolve}
    \chi(\mathbb{P}^1)\times(\chi_{\CY}-\chi(\text{pt}))+\chi(\mathbb{P}^2)\times\chi(\text{pt})=-575\,,
\end{align}
in agreement with \eqref{resfunX10}. Ignoring the effect of the special point, one would instead find 
 $\chi(\mathbb{P}^1)\times\chi_{\CY}=-576$,
as predicted by the naive ansatz.

While the maximal genus attainable by the standard direct integration method is 50, using modularity, we can predict
GV invariants close to the Castelnuovo bound to arbitrary genus (see Table \ref{table_GVX10}), and provide sufficiently many boundary conditions
to push the direct integration method, in principle, up to genus 70.

\subsection{$X_{4,2}$}
\label{subsec-X42}

\begin{table}
\be
\begin{array}{|r|r|rrrrrr|}
\hline
Q & g_C & \delta=0 &  \delta=1 &  \delta=2 &  \delta=3 &  \delta=4 & \delta=5\\
\hline
 1 & 1 & 0 & 1280 &  &  &  &  \\
 2 & 2 & 0 & 0 & 92288 &  &  &  \\
 3 & 3 & 0 & 0 & 2560 & 15655168 &  &  \\
 4 & 4 & 0 & -8 & -672 & 17407072 & 3883902528 &  \\
 5 & 5 & 0 & 0 & 7680 & 16069888 & 24834612736 & 1190923282176 \\
 6 & 6 & 0 & 0 & 276864 & 12679552 & 174937485184 & 23689021709568 \\
 7 & 7 & 0 & 7680 & 591360 & -285585152 & 2016330670592 & 494602061689344 \\
 8 & 9 & 15 & 1520 & 67208 & -8285120 & -46434384200 & 37334304102560 \\
 9 & 10 & 0 & -25600 & -2270720 & 370290688 & -4031209095680 & 1103462757073920 \\
 10 & 12 & 0 & 0 & 1384320 & 117390080 & 528559731712 & -344741538150784 \\
 11 & 14 & 0 & 0 & -46080 & -160005120 & -109083434240 & 163217721434624 \\
 12 & 16 & 0 & -96 & -12096 & 208486080 & 49221875968 & -145360041245120 \\
 13 & 18 & 0 & 0 & -61440 & -223475712 & -160179161088 & 272915443716096 \\
 14 & 20 & 0 & 0 & 2491776 & 175162624 & 1228486889728 & -1047846937829632 \\
 15 & 22 & 0 & -56320 & -4428800 & 1220514304 & -16165844458240 & 7742999973263360 \\
 16 & 25 & 84 & 7408 & 286784 & -30323216 & -231113426452 & 341194684288608 \\
 17 & 27 & 0 & 71680 & 5002240 & -1685727232 & 22238429571584 & -11254527777976576 \\
 18 & 30 & 0 & 0 & 4152960 & 230535424 & 2276356656640 & -2136509421094912 \\
 19 & 33 & 0 & 0 & 130560 & 507426816 & 386536492032 & -775389100867584 \\
 20 & 36 & 0 & -240 & -23808 & 626523936 & 144150871104 & -587797370270104 \\
 21 & 39 & 0 & 0 & 161280 & 632463360 & 482500187136 & -1002354648247296 \\
 22 & 42 & 0 & 0 & 6367872 & 217984256 & 3670593912832 & -3643185915136000 \\
 23 & 45 & 0 & 133120 & 5255680 & -3616804864 & 46617513355264 & -26329088088999936 \\
 24 & 49 & 180 & 8240 & 90016 & -91088144 & -597217698472 & 974876677046816 \\
 25 & 52 & 0 & -158720 & -4172800 & 4443311104 & -56810684083200 & 33070947498452480 \\
 26 & 56 & 0 & 0 & 9136512 & 71061760 & 5412945197824 & -5609947543679488 \\
 27 & 60 & 0 & 0 & -276480 & -1095613440 & -818893387776 & 1871635810564608 \\
 28 & 64 & 0 & -480 & -22848 & 1324638144 & 242234826816 & -1337626038427488 \\
 29 & 68 & 0 & 0 & -322560 & -1278422016 & -943519552512 & 2226739820757504 \\
 30 & 72 & 0 & 0 & 12458880 & -289968896 & 7510718536448 & -8087071198417408 \\
 31 & 76 & 0 & -250880 & 5058560 & 7256038912 & -93671359907840 & 59478452149884928 \\
 32 & 81 & 324 & -5200 & -169696 & -165887120 & -1137829570120 & 1982002329031968 \\
 33 & 85 & 0 & 286720 & -10846720 & -8197510144 & 108071476324864 & -70629413377719296 \\
 34 & 90 & 0 & 0 & 16334976 & -958134016 & 9978853510144 & -11133687621246976 \\
 \hline
\end{array}
\nonumber
\ee
 \caption{GV invariants $\GVg{g_C(Q)-\delta}$ 
 for $X_{4,2}$, assuming modularity.}
 \vspace{-0.5cm}
\label{table_GVX42}
\end{table}

Finally, we turn to the degree $(4,2)$ complete intersection in $\IP^5$. For this model $\kappa=8$ and there are 15 polar coefficients with one modular constraint. In  \cite{Alexandrov:2022pgd}
it was found that the naive Ansatz \eqref{naive} is incompatible with modularity. Using \eqref{thmS11inv} and GV invariants up to genus 50, we find that the first terms of the generating function are given by 
\be
\begin{split}
\hpol_{0} =&\, \q^{-\frac{8}{3}}\Bigl(\underline{-6 + 880 \q -60192\q^2}
-780416\q^3 +23205244196\q^4 +36880172393344\q^5 
\\
&\, +10924546660884800\q^6
+1454816640629235200\q^7
+\cdots \Bigr),
\\
\hpol_{1} =&\, \q^{-\frac{8}{3}+\frac{9}{16}}\Bigl( \underline{0-5120 \q -\tfrac{1}{30} \PT(25, -49) \q^2}
+\tfrac{1}{29}\, (222720 + \PT(25, -48))\q^3 +\cdots \Bigr),
\\
\hpol_{2} =&\, \q^{-\frac{8}{3}+\frac14} \Bigl( \underline{0+ 0 \q + \tfrac{1}{33} \PT(26, -53) \q^2}
-\tfrac{1}{32} \PT(26, -52)\q^3+\cdots \Bigr),
\\
\hpol_{3} =&\, \q^{-\frac{8}{3}+\frac1{16}} \Bigl( \underline{0+0\q+7680 \q^2}
+\tfrac{1}{35} \PT(27, -56)\q^3+\cdots\Bigr),
\\
\hpol_{4} =&\, \q^{-\frac{8}{3}} \Bigl( \underline{0+12\q-2112\q^2} -34689216\q^3
+\tfrac{1}{37} \PT(28, -59)\q^4
+\cdots\Bigr).
\end{split}
\label{resX42}
\ee
Although not all polar terms are found in this way, the result \eqref{resX42} provides an overdetermined set 
of coefficients which are compatible with the modular constraint and sufficient to fix uniquely 
the corresponding modular form. It is found to be
\bea
h_{\mu} &=& \frac{1}{\eta^{88}} \[ -\frac{827243  E_6^7}{13060694016 } 
\right.
\nn\\
&&  -\frac{E_4 \left(-71601885840 E_4^9-69248772786 E_4^6 E_6^2
+131750318292 E_4^3 E_6^4+14988448525 E_6^6\right)}{2190387225600}D 
\nn\\
&& +\frac{ \left(-7850108795 E_4^8 E_6-3026319343 E_4^5 E_6^3
+15844024271 E_4^2 E_6^5\right)}{30422044800}D^2 
\nn\\
&& +\frac{\left(41784458605 E_4^9+14762282727 E_4^6E_6^2
-68049440469 E_4^3 E_6^4-1016731100 E_6^6\right)}{19013778000}D^3 
\label{hmuX42}\\
&&-\frac{4  \left(173171 E_4^7 E_6+342266 E_4^4 E_6^3+44435 E_4 E_6^5\right)}{229635 }D^4 
\nn\\
&& \left.
+\frac{16\left(-93844535 E_4^8-89437029 E_4^5 E_6^2
+93510063 E_4^2 E_6^4\right)}{132040125} D^5 
\] \vths{8}_{\mu}\, ,
\nn
\eea
and produces the following expansions
\be
\begin{split}
h_{0} =&\, \q^{-\frac{8}{3}} \,\Bigl(\underline{ -6 + 880 \q - 60192 \q^2} - 780416 \q^3 + 23205244196 \q^4 +
 36880172393344 \q^5 
\\ 
& + 10924546660884800 \q^6 + 1454816640629235200 \q^7
+\dots \Bigr),
\\
h_{1} =&\, \q^{-\frac{8}{3}+\frac{9}{16}}\,\Bigl(
\underline{ 0-5120 q + 668160 \q^2 }+ 112032256 \q^3 + 2015342615552 \q^4 
\\ 
& + 1027768507417600 \q^5 
+ 184583137843579904 \q^6 + 17979440506308718592 \q^7
+\dots  \Bigr),
\\
h_{2} =&\, \q^{-\frac{8}{3}+\frac14}\,\Bigl(
\underline{ 0 + 0 \q + 276864 \q^2} - 32485376 \q^3 + 176489687424 \q^4 + 168522803580928 \q^5 
\\ 
&
+ 39373360484128256 \q^6 + 4527688807584194560 \q^7
+\dots \Bigr),
\\
h_{3} =&\, \q^{-\frac{8}{3}+\frac1{16}}\,\Bigl(
\underline{0 +0 \q+7680 \q^2} - 32203776 \q^3 + 27746555904 \q^4 + 53778203675136 \q^5 
\\ 
&  +
 15108125739695104 \q^6 + 1937976067726382592 \q^7
+\dots \Bigr),
\\
h_{4} =&\, \q^{-\frac{8}{3}}\,\Bigl(
\underline{ 0 + \underline{12} \q - \underline{2112} \q^2} - 34689216 \q^3 + 10834429824 \q^4 +
 36099879476640 \q^5 
\\ 
&
+ 10900431340916352 \q^6 + 1454331023779312896 \q^7
+\dots \Bigr).
\end{split}
\label{resX42full}
\ee
Furthermore, coefficients of $\q^2$, $\q^3$ in $h_0$, $\q^2$, $\q^3$, $\q^4$ in $h_1$, up to $\q^5$ in $h_2$, $\q^3$, $\q^4$, $\q^5$ in $h_3$, and $\q^4$ in $h_4$ are also reproduced by \eqref{thmS11inv} with $k=k_0-1$. Thus, there is overwhelming evidence that \eqref{hmuX42} is correct. 
While the maximal genus attainable by the standard direct integration method is 50, using modularity we can predict
GV invariants close to the Castelnuovo bound to arbitrary genus (see Table \ref{table_GVX42}), and provide sufficiently many boundary conditions 
to push the
direct integration method up to genus 64.

\section{Discussion}
\label{sec-disc}

In this work, we have exploited a triangle of relations between Gopakumar-Vafa invariants $\GVg{g}$, which determine the topological string partition function on a CY threefold $\CY$, Pandharipande-Thomas invariants $\PT(Q,m)$ which count bound states of a single anti-D6-brane with $Q$ D2 and $m$ D0 branes, and  D4-D2-D0 invariants $\bOm_{r,\mu}(\hq_0)$, which count BPS black holes with $r$ units of D4-brane charge along an ample divisor $\cD$, and D2-D0 brane charge determined by $\mu$ and $\hq_0$. Mathematically, these invariants count embedded curves, stable pairs and Gieseker-stable coherent sheaves supported on $\cD$, respectively. 
While the relation between GV and PT invariants is standard \cite{gw-dt}, and relations between GV invariants and D4-D2-D0 indices were first proposed in \cite{Ooguri:2004zv}, we used a novel explicit formula \eqref{thmS11inv} proven in Appendix \ref{sec_appS}, which applies for one-parameter CY threefolds with $\Pic\CY=\IZ H$ (or more generally, CY threefolds satisfying Assumption (*) in Prop. \ref{prop.relax}) and for $r=1$ unit of D4-brane charge. We applied this formula for the 13 CY threefolds of hypergeometric type, for which we have computed GV invariants (and therefore PT invariants) to relatively high genus using the direct integration method of \cite{Huang:2006hq}. 

For most models, we could rigorously compute the first few coefficients in the generating series of Abelian D4-D2-D0 invariants, including both polar and non-polar terms, and find a unique vector-valued modular form which reproduces all of them, providing a 
striking confirmation of the modularity properties predicted by string theory. These results also provide indirect support for the BMT inequality which is assumed in the derivation of \eqref{thmS11inv}, in cases where it is not yet known to hold.  
For $X_{3,3}$, $X_{4,3}$, $X_{3,2,2}$ and $X_{2,2,2,2}$, we could not compute sufficiently many terms rigorously to uniquely fix the vector-valued modular form, but in the first two cases we could determine a unique candidate which agrees with the formula \eqref{thmS11inv} for many coefficients, albeit outside the known regime of validity for this formula. 
For $X_{3,2,2}$ and $X_{2,2,2,2}$, our current knowledge of GV invariants is not sufficient to identify the modular form with sufficient confidence.
Conversely, in cases where the vector-valued modular form could be identified,  
we used these modular predictions
to determine GV invariants close to the Castelnuovo bound for arbitrarily high genus. These results provide new boundary conditions for the direct integration method, which in principle allow to reach higher genus than hitherto possible (in practice, some computational limitations need to be overcome in order to reach such high genera). 
The case of  $X_{4,3}$ is particularly noteworthy, as it requires combining information from direct integration, modularity and wall-crossing to go beyond the  restrictions imposed by each of these methods separately.

These results raise several natural questions. First, it is intriguing that the naive Ansatz \eqref{naive}, which was proposed as an educated guess in \cite{Alexandrov:2022pgd}, so often manages to produce the correct polar terms. 
As discussed in \S\ref{sec_wcr0}, a similar result \eqref{conseqTh1} arises by studying the chamber structure of rank 0 DT-invariants in the space of weak stability conditions.
Unfortunately, the walls can only be controlled under the very restrictive assumption \eqref{condFeyz} which in practice limits its use to the most polar terms. It would be very interesting to relax the condition \eqref{condFeyz}, but this seems to require a stronger bound on $\ch_3$ than provided by the standard BMT inequality. Physically, such a result would give a clear physical origin of the polar coefficients in terms of bound states of D6-$\overline{\rm D6}$-branes, as proposed in \cite{Denef:2007vg}. Instead, the formula \eqref{thmS11} expresses the spectrum of anti-D6-branes as a sum of bound states of D4-D2-D0 branes and anti-D6-branes with lower D4-D2-D0 brane charge, and does not provide any insight on the micro-structure of D4-D2-D0 bound states by themselves. 

A second question is, why on earth should the generating series of Abelian D4-D2-D0 invariants be modular. Of course, physics gives a clear reason, by identifying them with the elliptic genus of the $(0,4)$-superconformal field theory obtained by wrapping an M5-brane on the divisor $\cD$. From the mathematics viewpoint however, modularity is still largely mysterious. For non-compact CY threefolds of the form $\CY=K_S$ where $S$ is a complex projective surface, the generating series of Abelian D4-D2-D0 invariants supported on the divisor $S$ is given by G\"ottsche's formula for the Euler number of the Hilbert scheme of points on $S$ \cite{Gottsche:1999ix}, which is manifestly modular. For K3-fibered CY threefolds, the modularity of vertical D4-D2-D0 indices (counting D4-branes supported on a K3-fiber) can be shown to follow from  G\"ottsche's formula for the Hilbert scheme of points on $K3$ and from the modularity of the generating series of Noether-Lefschetz numbers determined by the fibration \cite{Bouchard:2016lfg}. In our generating series \eqref{decomp-modform} of Abelian D4-D2-D0 invariants, it is tempting to identify the factor $1/\eta^{\kappa+c_2}$ as coming from the Hilbert scheme of points on the divisor $\cD$, and the remainder as the generating series of some Noether-Lefschetz-type numbers taking into account the moduli of the divisor $\cD\subset\CY$ equipped with a line bundle \cite{Feyzbakhsh:2020wvm}. Eventually, one would hope that modularity can be derived from the existence of an underlying vertex operator algebra acting on the cohomology of the moduli space of semi-stable sheaves, similar to the case of Hilbert scheme of points on surfaces \cite{nakajima1999lectures}.

Third, it would be very interesting to generalize this approach to the case
of non-Abelian D4-D2-D0 indices, where the generating series are expected to be mock modular. While the relation between rank 1 and rank 0 DT invariants 
from \cite[Thm 1.2]{Feyzbakhsh:2022ydn} 
covers this case, it requires taking 
the spectral flow parameter to be large enough, with unspecified lower bound.
Nonetheless, we expect that the approach in Appendix \ref{sec_appS} can be generalized and used to compute polar coefficients for $r>1$ as well. The strategy outlined in \cite{Alexandrov:2022pgd} can then be used to construct a suitable mock modular series (using the generating series of Hurwitz class numbers to cancel the modular anomaly in the $r=2$ case). In a subsequent work \cite{Alexandrov:2023ltz}, we apply this strategy 
 for the models $X_{8}$ and $X_{10}$ at rank $r=2$ and verify the 
mock modular properties predicted in \cite{Alexandrov:2016tnf,Alexandrov:2017qhn,Alexandrov:2018lgp,Alexandrov:2019rth}. 
It would also be desirable to generalize this approach to other classes of one-parameter CY threefolds (such as freely acting orbifolds of hypergeometric models, or complete intersections in Grassmannians or Pfaffians), and to the 
CY threefolds with 2 or more K\"ahler parameters. In particular, we expect an interesting interplay between the modularity of D4-D2-D0 invariants and the modularity associated to elliptic fibrations~\cite{Klemm:2012sx,Alim:2012ss,Huang:2015sta}. 

Finally, having found the generating series of Abelian D4-D2-D0 invariants, it is now straightforward to extract the asymptotic growth of the Fourier coefficients, and produce a Rademacher-type series which computes them explicitly~\cite{Dijkgraaf:2000fq,Dabholkar:2005by,Dabholkar:2005dt,Manschot:2007ha}. It would be very interesting to reproduce these microstate degeneracies from localization in supergravity, in analogy to cases with $\cN=8$ or $\cN=4$ supersymmetry \cite{Dabholkar:2011ec,Dabholkar:2014ema,Iliesiu:2022kny,Ferrari:2017msn,LopesCardoso:2022hvc} (see \cite{Murthy:2013xpa,Murthy:2015yfa,Gomes:2019vgy} for some progress in this direction).  
We hope to return to these issues in near future.

\medskip
\medskip
\newpage

\appendix

\section{New explicit formulae, by S. Feyzbakhsh\label{sec_appS}}

Let $(\CY, H)$ be a smooth polarised Calabi-Yau threefold (i.e. $K_ \CY \cong \cO_\CY$ and $H^1(\CY, \cO_\CY) =0$) with $\Pic(\CY) =\Z.H$ satisfying the BMT conjecture. In \S\ref{sec_Thm1}, we improve the result of \cite[Theorem 1.2]{Feyzbakhsh:2022ydn} for rank zero classes with minimal D4-brane charge $\ch_1=H$ and obtain, under some assumptions,  an explicit formula for the stable pair invariants $\PT_{m, \beta}$  in terms of rank zero DT invariants and stable pair invariants $\PT_{m', \beta'}$ for $\beta'. H < \beta. H$. In \S\ref{sec_appI}, we explain how this result can be inverted to determine rank 0 DT invariants with minimal D4-brane charge from stable pair invariants. In \S\ref{sec_appII}, we apply the wall-crossing argument for Theorem \ref{thm-main} to establish a Castelnuovo-type bound
for PT invariants, and compute the PT invariants saturating this bound explicitly for the quintic threefold. In \S\ref{sec_nonprimwc}, we extend Theorem \ref{thm-main} to a special case where non-primitive wall-crossing occurs. 
In \S\ref{sec_relaxPic}, we state a generalisation of Theorem \ref{thm-main} under the weaker assumption $(\star)$ that $H^3$ divides $H'.H^2$ for all $H'\in\Pic(\CY)$. Finally, in \S\ref{sec_rk0fromDTPT}, we strengthen \cite[Theorem 1.1]{Feyzbakhsh:2022ydn} under 
assumption $(\star)$. 
The proofs of these results are collected in \S\ref{sec_proof}. 

\noindent {\it Notation:} In this section, we label the charges by the Chern character (rather than the Mukai vector), so for instance $\bOmH(\v)$ counts $H$-Gieseker semi-stable sheaves of Chern character $\v$. Furthermore, for readability we decompose $\v$ into
its components in $H^0(\CY,\IZ)$, $H^2(\CY,\IZ)$, $H^4(\CY,\IQ)$ and $H^6(\CY,\IQ)$. For example,
$\v=(-1,0,\beta,-m)$ stands for $\v=-1+\beta-m$. We also use
the notation $\ch_{\leq 2}E$ for the projection of $\ch E$ on $H^0(\CY,\IZ)\oplus H^2(\CY,\IZ)\oplus H^4(\CY,\IQ)$.

\subsection{Stable pair invariants from rank 0 DT invariants with minimal D4-brane charge}
\label{sec_Thm1}
For any $\beta \in H_2(\CY, \Z)$ (which can also be regarded as a class in 
$H^4(\CY, \Z)$ by Poincar\'e duality), we define the integer
\begin{equation}\label{functionf}
\cC(\beta)\ \coloneqq \ \left\{\!\!\!\begin{array}{cc} \ \ \left\lfloor \frac{2}{3H^3} (\beta. H)^2 + \frac{\beta. H}{3}   \right\rfloor & \ \ \  \text{if $\beta. H < H^3$,} 
\vspace{.3 cm}\\
\ \ \left\lfloor \frac{1}{2H^3} (\beta. H)^2 + \frac{\beta. H}{2}  \right\rfloor & \ \ \ \text{if $H^3 \leq \beta. H$,}
\end{array}\right.
\end{equation}
which determines the Castelnuovo bound as we explain below. 
Consider the function 
\begin{equation}\label{function f}
f(x)\ \coloneqq \ \left\{\!\!\!\begin{array}{cl} 
\vspace{.1 cm}
x+\frac12 & \text{if $0 < x < 1$,}\\
\vspace{.1 cm}
\sqrt{2x+\frac{1}{4}} & \text{if $1 < x < \frac{15}{8}$,}\\
\vspace{.1 cm}
\frac{2}{3}x+ \frac{3}{4} & \text{if $\frac{15}{8} \leq x< \frac{9}{4}$,}\\
\vspace{.1 cm}
\frac{1}{3}x+ \frac{3}{2} & \text{if $\frac{9}{4} \leq x< 3$,}\\
\vspace{.1 cm}
\frac{1}{2}x+ 1 & \text{if $3 \leq x$.}
\end{array}\right.
\end{equation}
\begin{theorem}\label{thm-main}
	Fix $m \in \Z$ and $\beta \in H_2(\CY, \Z)$ such that $\beta.H>0$ and 
	\begin{equation}\label{condition-f}
	f\left(\frac{\beta. H}{H^3}\right) < -\frac{3m}{2\beta. H}\, .
	\end{equation}
	Then 
\be	
\begin{split}\label{pt-thm}
	& \PTi_{m, \beta} = 
	\\ 
	&  \sum_{(m',\, \beta') \,\in\, M_{m,\, \beta}} (-1)^{\chi_{m',\beta'}}\chi_{m',\beta'}\ \PTi_{m', \beta'}\ \bOmH\left(0,\ H, \ \frac{1}{2}H^2 -\beta' +\beta\ , \ \frac{1}{6}H^3 +m'-m -\beta'. H \right),
\end{split}
\ee
	where 
	\begin{equation}
	\chi_{m',\beta'} = \beta. H + \beta'. H +m -m' - \frac{H^3}{6} - \frac{1}{12}c_2(\CY).H\,.
	\end{equation}
	The set $M_{m,\, \beta}$ consists of pairs $(m', \,\beta') \in H_0(\CY, \Z) \oplus H_2(\CY, \Z)$  such that
	\begin{equation}\label{all}
	0 \leq \beta'. H  \leq \frac{H^3}{2} + \frac{3mH^3}{2\beta. H} + \beta. H
	\end{equation} 
	and 
	\begin{align}\label{final}
	- \cC(\beta')
	\leq \ m' \ \leq  \frac{1}{2H^3} (\beta. H -\beta'. H)^2 + \frac{1}{2}(\beta. H +\beta'. H) +m. 
	\end{align}
\end{theorem}
Since $f(x) > \frac{1}{2}$, \eqref{condition-f} implies that $\frac{1}{2} < -\frac{3m}{2\beta. H}$, thus $\beta'. H < \beta. H$ in \eqref{all}. 

\begin{remark}
Here are three comments regarding Theorem \ref{thm-main}: a) to prove Theorem \ref{thm-main}, we only need a weaker version of BMT conjecture explained in Remark \ref{remark weaker BMT}, b) one can strengthen Theorem \ref{thm-main} for specific CY threefolds (e.g. quintic threefolds) where a stronger version of Bogomolov-Gieseker inequality holds, see Remark \ref{remark stronger BG in} for more details, and c) a generalisation of Theorem \ref{thm-main}, when $\CY$ is not of Picard rank one but satisfies assumption $(\star)$, is proved in Proposition \ref{prop.relax}.
\end{remark}

For any $(m,\beta) \in H_0(\CY, \Z) \oplus H_2(\CY, \Z)$, consider the function $\Psi(x, x', \alpha)$ defined in \eqref{defPsi} for $x = \frac{\beta.H}{H^3}$ and $\alpha= - \frac{3m}{2\beta.H}$. We define the function $\Psi_{m, \beta} \colon \mathbb{R} \to \mathbb{R}$ as
\begin{align}
    \Psi_{m,\,  \beta}(x') \coloneqq \Psi\left(\frac{\beta.H}{H^3},\ x',\ -\frac{3m}{2\beta.H}\right) =  \frac12\,\left(\frac{\beta.H}{H^3}-x'\right)^2 + \frac12\,\left(\frac{\beta.H}{H^3}+x'\right)+ \frac{m}{H^3} + \frac12 \,x'^2 + \frac12\, x'\ .
\end{align}
Note that $\Psi_{m, \beta}(x')$ is the difference of the right and the left hand side of \eqref{final}, up to a factor of $H^3$, for $x'=\beta'.H/H^3$. 
\begin{definition}\label{def-optimal} 
	A value of $(m,\beta) \in H_0(\CY, \Z) \oplus H_2(\CY, \Z)$ is called optimal if 
	\begin{itemize} 
		\item the inequality \eqref{condition-f} is satisfied, and 
		\item 	$\Psi_{m, \beta}(\frac{1}{H^3}) < 0$ and $\Psi_{m, \beta}\left(\frac{1}{2} + \frac{3m}{2\beta. H} + \frac{\beta. H}{H^3}\right) < 0$ (i.e. $\Psi_{m, \beta}$ is negative for all possible positive values of $\beta'. H$ in \eqref{all}).
		\end{itemize}
\end{definition}
For an optimal value of $(m,\beta)$, the only possible value of $\beta'$ in the wall-crossing formula \eqref{pt-thm} is $\beta' = 0$. Since $\PTi_{m',0} = 1$ if $m' =0$ and otherwise vanishes, we get the following:
\begin{corollary}\label{cor-explicit formula}
For any optimal value $(m,\beta) \in H_0(\CY, \Z) \oplus H_2(\CY, \Z)$, 
we have
\begin{equation}\label{claim.1}
\PTi_{m, \beta} = (-1)^{1+\chi(\cO_\CY(H), \w)}\chi(\cO_\CY(H), \w)\ \bOmH(\w),
\end{equation}
where $\w = \left(0, H, \beta+ \frac{H^2}{2}, -m + \frac{H^3}{6}\right)$. 
\end{corollary}

\subsection{Application I: An explicit formula for minimal rank zero DT invariants}
\label{sec_appI}

Corollary \ref{cor-explicit formula} gives us a way to write arbitrary minimal rank zero DT invariants in terms of PT invariants. Pick $\beta_0 \in H^4(\CY, \mathbb{Q})$ and $m_0 \in \Z$. 
After tensoring by the line bundle $\cO_\CY(k H)$ with $k\in\IZ$,
the rank 0 DT invariant is unchanged,
\begin{equation}
\bOmH(0, H, \beta_0, m_0) = \bOmH\left(0, H, \beta_0+k H^2 , m_0+k\beta_0.H+\frac{k^2}{2}H^3 \right). 
\end{equation} 
Then as a consequence of Corollary \ref{cor-explicit formula}, one gets the following. 
\begin{corollary}\label{Cor-explicit}
There exists $k(\beta_0, m_0) \in \mathbb{Z}_{>0}$ such that for any integer $k \geq k(\beta_0, m_0)$, the class $(m,\beta)$ for
\begin{equation}
\beta  \coloneqq \beta_0+k H^2 -\frac{H^2}{2} \ , \qquad m \coloneqq \frac{H^3}{6} -m_0- k\beta_0.H-\frac{k^2}{2}H^3
\end{equation}
is optimal (see Definition \ref{def-optimal}), thus
\begin{align}
\bOmH(\w) = (-1)^{1+\chi\big(\cO_\CY((1-k)H), \,\w\big)}
\frac{\PTi_{m, \beta} }{\chi\big(\cO_\CY((1-k)H), \,\w\big)}\, ,
\end{align}
where $\w = \left(0, H, \beta_0, m_0 \right)$.
\end{corollary}

\subsection{Application II: Castelnuovo bound for PT invariants} 
\label{sec_appII}
As a result of wall-crossing for rank $-1$ classes and induction on $\beta. H$, we can prove the following Castelnuovo-type bound for stable pairs.  
\begin{theorem}\label{thm-Cast}
	Fix $\beta \in H_2(\CY, \Z)$ and $m \in \Z$. Let $\cO_\CY \xrightarrow{s} F$ be a stable pair such that $\ch_2(F) =\beta$ and $\ch_3(F) =m$, then 
	\begin{equation}\label{claim-2}
	m \geq -\cC(\beta).  
	\end{equation}
In particular, the invariant $\PTi_{m, \beta}$ vanishes unless the inequality \eqref{claim-2} is satisfied.
\end{theorem}
The nature of the proof of Theorem \ref{thm-Cast} is similar to \cite[Proposition 1.3]{macri-genus-bound} where the same result is proved under the extra assumption that for any $E \in \Coh(\CY)$, we have $\ch_2(E) \in \frac{H^2}{2}\Z$ and $\ch_3(E) \in \frac{H^3}{6}\Z$. Theorem \ref{thm-Cast} for quintic threefolds, as well as the first part of Theorem \ref{thm-quintic} below, is also proved in the recent paper \cite{Liu:2022agh} via different arguments.   

\begin{remark}\label{rem-cast}
Note that when $m$ is close to the Castelnuovo bound $-\cC(\beta)$, then one can apply Theorem \ref{thm-main} to find an explicit formula for $\PTi_{m, \beta}$ in terms of rank zero invariants $\bOmH(\w)$ with $\ch_1(\w) =H$. Since $\bOmH(\w) = \bOmH(\w \otimes \cO_\CY(k H))$ for any $k \in \Z$, the knowledge of a few rank zero DT invariants determines PT invariants along the boundary of Castelnuovo bound.  
\end{remark}

Let us spell out Remark \ref{rem-cast} for the case of the quintic threefold $X_5$. Based on physical arguments in \S\ref{subsec-X5}, the following vanishings are expected (here $\kappa = H^3= 5$): 
\begin{conj}\label{con}
	\begin{enumerate}
		\item [(i)] $\bOmH\Bigl(0,\, H,\, \pm(\frac{1}{2} + \frac{1}{\kappa})H^2 ,\, \frac{H^3}{6}\Bigr) = 0$,
		\item [(ii)] $\bOmH\Bigl(0,\, H,\, \pm(\frac{1}{2} + \frac{2}{\kappa})H^2 ,\, (-\frac{m}{\kappa} + \frac{1}{6})H^3\Bigr) = 0$ for $m=0, -1$. 
	\end{enumerate} 
\end{conj}   
\begin{theorem}\label{thm-quintic}
	Take an integer $\mu \geq 13$ or $\mu =10$. If $\mu \overset{\kappa}{\equiv} 0$ and $m= -\,\cC(\frac{\mu}{\kappa}H^2)$, then 
		\begin{equation}
		\PTi_{m ,\, \frac{\mu}{\kappa}H^2} = (-1)^{m + \mu}(5-m-\mu) \times 5.
		\end{equation}
		If $\mu \overset{\kappa}{\equiv} \pm p$ where $p =  1\ (\text{resp. } 2)$ and Conjecture \ref{con}(i) (resp. Conjecture \ref{con}(ii)) holds, then $\PTi_{m,\, \frac{\mu}{\kappa}H^2}$ vanishes if $m <  -\,\cC(\frac{\mu}{\kappa}H^2) +p$; moreover, for $m= -\,\cC(\frac{\mu}{\kappa}H^2) +p$, then 
	\begin{equation}
	\PTi_{m ,\ \frac{\mu}{\kappa}H^2 } = (-1)^{1+\chi(\cO_\CY(H), \w)}\chi(\cO_\CY(H), \w)\ \bOmH(\w),
	\end{equation}
	where $\w= \Bigl(0,\, H,\, \big(\frac{1}{2} + \frac{\mu}{\kappa}\big)H^2 ,\,  \big(-\frac{m}{\kappa} + \frac{1}{6}\big)H^3 \Bigr)$. 
\end{theorem}

\subsection{An example with non-primitive wall-crossing}
\label{sec_nonprimwc}
In this part, we generalize Theorem \ref{thm-main} for some examples of pairs $(m,\beta)$ where the inequality \eqref{condition-f} is saturated. 
Similar results can be obtained when the inequality is mildly violated.
\begin{proposition}\label{lem-x=4}
	If $\beta.H>4 H^3$ and $\alpha = -\frac{3m}{2\beta.H} = f(x)$, where
	$(x,\alpha)$ are defined in \eqref{defxa},
	then Theorem \ref{thm-main} holds true. However if $\beta.H = 4H^3$ and $m=-8H^3$, then there is an additional wall with non-primitive wall-crossing, leading to
	\be
	\begin{split}
	\PTi_{m, \beta}  = &
	 \sum_{(m',\, \beta') \,\in\, M_{m,\, \beta}} (-1)^{\chi_{m',\beta'}}\chi_{m',\, \beta'} \PTi_{m', \beta'}\ \bOmH\left(0,\ H, \ \frac{1}{2}H^2 -\beta' +\beta\ , \ \frac{1}{6}H^3 +m'-m -\beta'. H \right) 
 \\ 
	& + \frac{1}{2}\,\big(\chi(\cO_\CY, \cO_\CY(2H))\big)^2 - \frac{1}{2}\,\chi(\cO_\CY, \cO_\CY(2H))\ ,
	\end{split}
	\ee
	where we recall that $\chi(\cO_\CY, \cO_\CY(2H))=\chi_{2\cD}=\frac43\kappa+\frac16 c_2$.
\end{proposition}

\subsection{Relaxing the Picard rank one assumption}\label{sec_relaxPic} 
In this subsection, we relax the Picard rank one assumption to
\begin{description}
\item[Assumption $(\star)$] $H^3$ divides $H'.H^{2}$ for all $H' \in \Pic(\CY)$.
\end{description}
The following proposition generalises 
Theorem \ref{thm-main}. 
\begin{proposition}\label{prop.relax}
    Let $(\CY, H)$ be a smooth polarised Calabi-Yau threefold $\CY$ satisfying the BMT conjecture and assumption $(\star)$.
Fix $m \in \Z$ and $\beta \in H_2(\CY, \Z)$ such that the condition \eqref{condition-f} is satisfied. Then
\be
\begin{split}
	& 
	\PTi_{m, \beta} = 
	\\ & \sum_{\begin{subarray}{l} 
	(H', \, m',\, \beta') \,\in\, \widetilde{M}_{m,\, \beta}\end{subarray}}
	(-1)^{\chi_{H', m',\beta'}}\chi_{H', m',\beta'} \ \PTi_{m', \beta'}\ \bOm_H\left(0,\ H', \ \frac{1}{2}H'^2 -\beta' +\beta\ , \ \frac{1}{6}H'^3 +m'-m -\beta'. H' \right), 
	\end{split} 
	\ee
	where 
	\begin{equation}
	\chi_{H', m',\beta'} = \beta. H' + \beta'. H' +m -m' - \frac{H'^3}{6} - \frac{1}{12}c_2(\CY).H'\,,
	\end{equation}
	and the set $\widetilde{M}_{m,\, \beta}$ consists of triples $(H', m', \,\beta') \in H^2(X, \mathbb{Q}) \oplus H_0(\CY, \Z) \oplus H_2(\CY, \Z)$  such that
	\begin{enumerate}
	    \item[(i)] $H'.H^2 = H^3$, 
	    \item [(ii)] $0 \leq \beta'. H  \leq \frac{H'^2.H}{2} + \frac{3mH^3}{2\beta. H} + \beta. H$,  
	    \item [(iii)]$m'\ \geq\ -\frac{2}{3}\beta'.H\left(\frac{\beta'.H}{H^3}+ \frac{1}{2}  \right)$, and
	    \item[(iv)] $ m' \leq \frac{1}{2H^3}\left( \frac{1}{2}H'^2H-\beta'. H +\beta. H\right)^2 + \frac{H^3}{24} -\frac{H'^3}{6} +m + \beta'.H'$.
	\end{enumerate}
\end{proposition}

\subsection{Rank 0 DT invariants from rank one DT and PT invariants}
\label{sec_rk0fromDTPT}
Finally, we provide a strengthening of the wall-crossing formula \cite[Theorem 1.1]{Feyzbakhsh:2022ydn} for rank zero classes, which holds for CY threefolds satisfying the BMT inequality and assumption $(\star)$.  

Given a polarization $H\in H^2(\CY,\IZ)$ and a rank 0 Chern class\footnote{Note that
the sign of $m$ is flipped compared to \cite[Theorem 1.1]{Feyzbakhsh:2022ydn}.}
\be
\label{vrank0}
\v= [0,D,\beta,-m] \in H^2(\CY, \mathbb{Q}) \oplus H^4(\CY, \mathbb{Q}) \oplus H^6(\CY, \mathbb{Q})\,,
\ee 
with $D\neq 0$,
let us define
\be
\label{defQSoh}
Q_H(\v) = \frac12 \left( \frac{ D. H^2}{H^3 }\right)^2 +
6 \left(\frac{\beta. H}{D. H^2}\right)^2 +\frac{12 m}{D. H^2}\, .
\ee
\begin{theorem}\label{thm.rk0}
    Let $(\CY, H)$ be a smooth polarised CY threefold $\CY$ satisfying the BMT conjecture and assumption $(\star)$. If a rank zero class $\v = [0, D, \beta, -m]$ satisfies $Q_H(\v) < \min\{\frac{r^2}{2}-\frac{1}{8},\, r-\frac{1}{2}\}$ for $r = \frac{D.H^2}{H^3}$ then 
\begin{equation}\label{sumprop3}
		\bOmH(\v) = \big(\#H^2(X,\Z)_{\mathrm{tors}}\big)^2 \sum_{\substack{\v_1\, =\, -e^{D_1}(1, 0, -\beta_1, -m_1)\\ \v_2 \,=\, e^{D_2}(1, 0, - \beta_2, -m_2)\\
				\v_1+\v_2\, =\, \v \\
				(D_i,\ \beta_i,\ m_i) \,\in\, M_i\left(\v \right)
		}}  (-1)^{\chi(\v_2, \v_1) -1}\;\chi(\v_2, \v_1)\, {\mathrm P}_{-m_1, \beta_1}\ {\mathrm I}_{m_2, \beta_2}\,.
		\end{equation}
Here $M_i(\v)$ for $i=1, 2$ is the set all classes $(D_i,\, \beta_i,\, m_i) \in H^2(X, \mathbb{Q}) \oplus H^4(X, \mathbb{Q}) \oplus H^6(X, \mathbb{Q})$ such that
\begin{equation}\label{bound on Di}
    \left| \frac{D_i.H}{H^3} - \theta_i \right| <1 \qquad \text{for} \qquad \theta_i \coloneqq  \frac{\beta.H}{D.H^2} + (-1)^i \sqrt{\frac{r^2}{4} - \frac{1}{2}Q_H(\v) }\ ,
\end{equation}
\begin{equation}\label{cc.1}
	-\frac{D_i^2.H}{2H^3} +\frac{\beta_i.H}{H^3} + \frac{\beta.H}{D.H^2} \frac{D_i.H^2}{H^3}  \leq  -\frac{1}{8} \left(\frac{D.H^2}{H^3}\right) + \frac{1}{2}\left(\frac{\beta.H}{D.H^2}\right)^2 + \frac{1}{4}Q_H(\v),
\end{equation}
and 
\begin{equation}\label{cc.2}
	(-1)^{i+1}m_i \leq\  \frac{2}{3}\beta_i.H \left( \frac{\beta_i.H}{H^3} + \frac{1}{2}\right).
\end{equation}
\end{theorem}   	
Note that if $Q_H(\v)$ is as small as required in \cite[Theorem 1.1]{Feyzbakhsh:2022ydn}, then one can apply \cite[Proposition 3.5 \& 3.6]{Feyzbakhsh:2022ydn} to show that the above sets $M_i(\v)$ reduce to $M(\v)$ so that the final wall-crossing formulae agree.

\subsection{Proofs}
\label{sec_proof}
In this subsection, we collect the proofs of the various claims up to now.

To prove Theorem \ref{thm-main}, we use weak stability conditions $\nu_{b,w}$ for $(b,w) \in U$ where $U = \{(b,w) \in \mathbb{R}^2 \colon w > \frac{b^2}{2}\}$, see Fig. \ref{fig-curves} and \cite[\S 1]{Feyzbakhsh:2021rcv} for all details. We study wall-crossing for the fixed class
\begin{equation}
\v= (-1, 0, \beta, -m) \in \textstyle{\bigoplus_{i=0\,}^3}H^{2i}(\CY,\mathbb{Q}). 
\end{equation}  
We start in the large volume limit $b>0$ and $w \gg 1$ where an object $E \in \cA_b$ of class $\v$ is $\nu_{b,w}$-semistable if and only if $E^{\vee} \otimes \det(E)^{-1}[1]$ is a stable pair, see 
\cite[\S 3]{Toda:2011aa}. Then we move down and investigate all walls of instability for objects of class $\v$. 


By the conjectural BMT inequality \eqref{BMTineq}, if there is a $\nu_{b,w}$-semistable object in $\cA_b$ of class $\v$, then $L_{b, w}(\v) = w(2H^3\beta. H) + 3b(-H^3)(-m) +2(\beta. H)^2 \geq 0$, i.e.
\begin{equation}
w \geq -b \frac{3m}{2\beta. H} - \frac{\beta. H}{H^3}\,.  
\end{equation}	
Hence any wall for class $\v$ lies above or on the line $\ell_f$ of equation $w = \alpha b -x$ where
\begin{equation}\label{alpha, x}
    \alpha = -\frac{3m}{2\beta. H} \qquad \text{and} \qquad x= \frac{\beta. H}{H^3}\, . 
\end{equation}
 Let $b_1< b_2$ be the values of $b$ at the intersection points of the line $\ell_f$ with the boundary $\partial U$. 
 
Assume there is a wall $\ell$ for class $\v$. Then there is an object $E \in \cA_b$ of class $\v$ which is strictly $\nu_{b,w}$-semistable for all $(b,w) \in \ell \cap U$. Let $E' \rightarrow E \rightarrow E''$ be a destabilising sequence along the wall $\ell$.  
\begin{lemma}\label{lem.destabilising objects}
	Suppose 
 	\begin{enumerate}
	    \item[(i)] $0 < b_1 < 2$, 
	    \item[(ii)] $1 < b_2-b_1$, and 
	    \item[(iii)] if $b_2 r \leq c \leq b_1(r+1)$ for some $(r, c) \in \mathbb{Z}_{\geq 0} \oplus \mathbb{Z}_{\geq 0}$, then either $(r, c) = (0, 1)$ or $(0, 0)$.
	\end{enumerate}
   Then $b_1 \geq 1$ and there is an ordering $E_0,E_1$ of $E',E''$ such that
	\begin{itemize}
		\item $E_0$ is a rank zero Gieseker-stable sheaf with $\ch_1(E_0) =H$.   
		\item $E_1$ is a rank $-1$ object and $\nu_{b,w}$-stable for $b> -1$ and $w \gg 1$. 
	\end{itemize}
	Moreover, there is no $\nu_{b,w}$-semistable object of class $\v$ for $(b,w) \in U$ below $\ell_f$. 
\end{lemma}
\begin{proof}
	Since $\rk(E)=-1$, one of the objects $E',E''$ has rank$\,<0$; call it $E_1$. The other $E_0$ has rank$\,\ge0$. Let 
	\begin{equation}
	\ch_{\leq 1}(E_0) = (r, cH) \qquad \text{and} \qquad \ch_{\leq 1}(E_1) = (-1-r, -cH).
	\end{equation}
	By \cite[Remark 1.2]{Feyzbakhsh:2021rcv} for any point $(b,w) \in \ell \cap U$, we have $\ch_1(E_i)H^2 -b \ch_0(E_i)H^3 \geq 0$ for $i=1, 2$. Since $\ell$ lies above $\ell_f$, this in particular holds for $b =b_1, b_2$ which implies 
	\begin{equation}\label{A1}
	b_2 r \leq c \leq b_1(1+r).
	\end{equation}
	By assumption (iii), it follows that either $(r, c) = (0, 0)$ or $(r, c) = (0, 1)$. In the first case, $E_0$ is of $\nu_{b,w}$-slope $+\infty$ for all $(b,w) \in U$, thus it cannot have the same $\nu_{b,w}$-slope as $E$ along the wall $\ell$. Thus $(r, c) = (0, 1)$ and $1 \leq b_1$. 
	
Hence $E_1$ is a rank $-1$ object with $\ch_1(E_1) = -H$. By \cite[Lemma 3.5]{feyz:effective-restriction-theorem}, there is no wall for $E_1$ crossing the vertical line $b =2$, so if $E_1$ is $\nu_{b,w}$-semistable for some $b=2$ and $w > 2$, then it is $\nu_{b, w}$-stable for $b=2$ and any $w >2$. By conditions (i) and (ii), the wall $\ell_f$ and so the wall $\ell$ intersects the vertical line $b=2$ at a point inside $U$, thus $\nu_{b,w}$-semistability of $E_1$ along the wall implies that $E_1$ is $\nu_{b, w}$-stable for $b=2$ and $w \gg 1$. Then the wall and chamber structure for the class $\ch(E_1)$ implies that $E_1$ is $\nu_{b,w}$-stable for any $b> -1$ and $w \gg 1$.   
	
We claim that there is no wall for $E_0$ when we move up from the wall $\ell$ to the large volume limit. Suppose for a contradiction that there was a wall $\ell_0$ with the destabilising sequence $F_1 \hookrightarrow E_0 \twoheadrightarrow F_2$. Set $r_i:=\rk\!\(\cH^{-1}(F_i)\)\ge0$. By definition of the heart $\cA_b$, for $(b,w) \in \ell_0 \cap U$, 
\begin{equation}
    \ch_1\(\cH^{-1}(F_i)\).H^2 \le b\, r_iH^3 \qquad \text{and} \qquad  \ch_1\(\cH^0(F_i)\).H^2 \ge b\,(\ch_0(F_i)+r_i)H^3. 
\end{equation}
Since $\ell_0$ lies above or on $\ell_f$, we may apply the left hand inequality for $b= b_1$ and the right hand for $b= b_2$. Then subtracting gives
$\ch_1(F_i).H^2\ \ge\ b_2\ch_0(F_i)H^3+ (b_2-b_1)r_iH^3$. 
Adding over $i=0,1$ gives 
\begin{equation}\label{B4}
H^3 = \ch_1(E).H^2\ \ge\ (b_2-b_1)(r_0+r_1)H^3.
\end{equation}
Since $b_2-b_1 >1$ by assumption (ii), we get $r_1=r_2 =0$. Thus $\cH^{-1}(F_i) =0$ as they are torsion-free sheaves by definition of the heart $\cA_b$. Thus $F_i$'s are sheaves with $\ch_0(F_1) =\ch_0(F_2) = 0$. Hence they have the same $\nu_{b,w}$-slope as $E_0$ with respect to any $(b, w) \in U$, so they cannot induce a wall. Thus $E_0$ is $\nu_{b,w}$-semistable for $w \gg 1$, hence is a torsion sheaf by \cite[Lemma 2.7(c)]{bayer2016space}. 

By definition of the heart $\cA_b$, any torsion sheaf $F$ lies in $\cA_b$ and $\nu_{b,w}(F) = \frac{\ch_2(F).H}{\ch_1(F).H^2}$ if $\ch_1(F) \neq 0$, otherwise $\nu_{b,w}(F) = +\infty$. In our case, since $\ch_1(E_0).H^2 = H^3$ is minimal, the $\nu_{b,w \gg 1}$-semistability of $E_0$ immediately implies that $E_0$ is a Gieseker-stable sheaf.    

\end{proof}

Recall that the equation of $\ell_f$ is $w = \alpha b -x$ for the values of $\alpha, x$ in \eqref{alpha, x}. The $b$-values of the intersection point of $\ell_f$ with $\partial U$, which is the parabola of equation $w = \frac{1}{2}b^2$, are 
\begin{equation}
    b_1=\alpha -\sqrt{\alpha^2-2x}, \quad 
    b_2 = \alpha + \sqrt{\alpha^2-2x} \,. 
\end{equation}
The condition (i) that $0 < b_1 < 2$ is equivalent to 
\begin{equation}
\alpha - \sqrt{\alpha^2-2x} < 2 \qquad \text{i.e.} \qquad 0< \alpha <2 \qquad \text{or} \qquad  1 + \frac{x}{2} < \alpha\,. 
\end{equation}
Also the condition (ii) that $b_2-b_1 >1$ is equivalent to 
\begin{equation}
2\sqrt{\alpha^2-2x} > 1 \qquad \text{i.e.} \qquad \sqrt{2x +\frac{1}{4}} < \alpha\,.
\end{equation}
Hence, a simple case by case analysis verifies the following: 
\begin{lemma}\label{lem-check conditions}
    Consider the function $f(x)$ defined in \eqref{function f}. If $x > 1$ and $\alpha > f(x)$, then the conditions (i), (ii) and (iii) in Lemma \ref{lem.destabilising objects} hold. 
\end{lemma}
By Lemma \ref{lem.destabilising objects}, the destabilising objects are of Chern character 
	\begin{equation}
	\v^0_{m', \beta'} \coloneqq \ch(E_0) = \left(0,\ H,\ \frac{1}{2}H^2-\beta' +\beta ,\ \frac{1}{6}H^3 +m'-m-\beta'. H\right)
	\end{equation}
	and 
	\begin{equation}
	\v^1_{m', \beta'} \coloneqq \ch(E_1) = \left(-1,\ -H,\ -\frac{H^2}{2}+\beta', \ -\frac{1}{6}H^3-m' + \beta'. H  \right).
	\end{equation}
We know that the point $\varpi(\ch(E_1)) = (1\,, \, -\frac{\beta'. H}{H^3} + \frac{1}{2})$ (defined in \eqref{defPi}) lies above or on $\ell_f$, so
	\begin{equation}\label{location of the wall}
	\frac{\beta'. H}{H^3}\,  \leq\,  \frac{3m}{2\beta. H} + \frac{\beta. H}{H^3} + \frac{1}{2}\,. 
	\end{equation} 	
Moreover, applying \cite[Lemma B.3]{Feyzbakhsh:2021rcv} for $E_0$ implies
\begin{align}\label{rankzerobound}
   \frac{1}{6}H^3 +m'-m-\beta'. H \leq \frac{1}{2H^3}\left( \frac{1}{2}H^3-\beta'. H +\beta. H\right)^2 + \frac{H^3}{24}
\end{align}
which is equivalent to 
\begin{equation}\label{bound for ch3(E_0)}
    m' \ \leq  \frac{1}{2H^3} (\beta. H -\beta'. H)^2 + \frac{1}{2}(\beta. H +\beta'. H) +m\, .
\end{equation}

\begin{proof}[Proof of Theorem \ref{thm-Cast}]
    Suppose $\PTi_{m, \beta} \neq 0$, i.e. there is a $\nu_{b,w}$-stable object of class $\v = (-1, 0, \beta, -m)$ for $b>0$ and $w \gg 1$. By \cite[Lemma 3.5]{feyz:effective-restriction-theorem}, there is no wall for class $\v$ crossing the vertical line $b=1$. Then the conjectural BMT inequality \eqref{BMTineq} at the boundary point $(b,w) = (1, \frac{1}{2})$ implies that 
	\begin{equation}\label{rank -1 bound}
	-\frac{m}{H^3}\ \leq\ \frac{2}{3}\(\frac{\beta. H}{H^3}\)^2 + \frac{\beta. H}{3H^3}
	\end{equation}
	which proves the claim \eqref{claim-2} if $\frac{\beta. H}{H^3} \leq 1$. Hence we may assume $\beta. H^3 > H^3$. If \eqref{claim-2} does not hold, then 
	\begin{equation}\label{assum}
	-\frac{3m}{2\beta. H} > \frac{3\beta. H}{4H^3} + \frac{3}{4} > f\(\frac{\beta. H}{H^3}\) .
	\end{equation}
	Then by combining Lemma \ref{lem.destabilising objects} and Lemma \ref{lem-check conditions}, it follows that as we move down from the large volume limit, any large volume limit stable object $E$ of class $\v$ gets destabilised along a wall with the destabilising objects $E_0$ and $E_1$ as described in Lemma \ref{lem.destabilising objects}. By \cite[Corollary 3.10]{bayer2016space}, we know that the discriminants $\Delta_H$ of the destabilising factors are less than $\Delta_H(\v)$, so $\beta'. H < \beta. H$. Thus by applying induction on $\beta. H$, we may assume the object $E_1$, which is large volume limit stable of rank $-1$, satisfies the claim. Combining it with \eqref{bound for ch3(E_0)} implies that   
	\begin{align}\label{final-new}
	-\frac{1}{2H^3} (\beta'. H)^2 - \frac{\beta'. H}{2}  
	\leq \ m' \ \leq  \frac{1}{2H^3} (\beta. H -\beta'. H)^2 + \frac{1}{2}(\beta. H +\beta'. H) +m. 
	\end{align}   
	This in particular implies that 
	\begin{equation*}
	0 \leq \frac{1}{2H^3} (\beta. H -\beta'. H)^2 + \frac{1}{2}(\beta. H +\beta'. H) +  \frac{1}{2H^3} (\beta'. H)^2 + \frac{\beta'. H}{2} +m.
	\end{equation*}
	If the claim \eqref{claim-2} does not hold, then 
	\begin{align*}
	    0 <\  & \frac{1}{2H^3} (\beta. H -\beta'. H)^2 + \frac{1}{2}(\beta. H +\beta'. H) +  \frac{1}{2H^3} (\beta'. H)^2 + \frac{\beta'. H}{2} + \left(-\frac{1}{2H^3}(\beta. H)^2 -\frac{1}{2}\beta. H\right) \nn\\
	    =\ & \beta'. H \left(\frac{\beta'. H}{H^3} -\frac{\beta. H}{H^3} +1  \right).
	\end{align*}
	Since $\beta'. H \geq 0$, we get $\beta. H -H^3 < \beta'. H$, then \eqref{location of the wall} gives $\beta. H -H^3 < \frac{3mH^3}{2\beta. H} + \beta. H + \frac{H^3}{2}$, i.e. $-\frac{3m}{2\beta. H} < \frac{3}{2}$ which is not possible by \eqref{assum} for $\beta. H > H^3$.  
\end{proof}

\begin{proof}[Proof of Theorem \ref{thm-main}]
   By \cite[\S 3]{Toda:2011aa}, any $\nu_{b,w}$-semistable object $E \in \cA_b$ of class $\v$ for $b>0$ and $w \gg 1$ is derived-dual of a stable pair up to tensoring by a line bundle with torsion $c_1$, so 
   \begin{equation}
    \bOm_{\infty}(\v) \coloneqq \bOm_{b\,> 0,\ w \,\rightarrow\, +\infty}(\v) = \big(\#H^2(\CY,\Z)_{\mathrm{tors}}\big)\PTi_{m, \beta}. 
\end{equation}  
On the other hand, we know that there is no $\nu_{b,w}$-semistable object of class $\v$ when $(b,w) \in U$ lies below $\ell_f$, so $\bOm_{b,w}(\v) = 0$. Between, this point and large volume limit, there are several walls $\ell$. 

If $0 < x \leq 1$ and $\alpha > f(x)$, then $b_1 < 1$, implying $\PT_{m, \beta} = 0$ and no non-trivial wall exists, so the claim follows. Therefore, we may assume $x > 1$. Subsequently, by applying Lemma \ref{lem-check conditions}, we can utilize Lemma \ref{lem.destabilising objects}, which describes the destabilizing factors along any wall $\ell$. We know that the first factor $E_0$ is Gieseker-stable, and any Gieseker-stable sheaf of class $\ch(E_0)$ is $\nu_{b,w}$-stable for $(b,w) \in \ell$ as there is no wall for $\ch(E_0)$ between the large volume limit and $\ell$. Thus for points $(b,w^{\pm})$ above and below the wall $\ell$, we have
\begin{equation}
    \bOm_{b,w^{\pm}}(\v^0_{m', \beta'}) = \bOmH(\v^0_{m', \beta'}).
\end{equation}
We also know that there is no wall for the other factor $E_1$ between $\ell$ and the large volume limit, thus 
\begin{equation}
    \bOm_{b,w^{\pm}}(\v^1_{m', \beta'}) =\bOm_{\,b>1,\ w \,\rightarrow\, +\infty}(\v^1_{m', \beta'}) =   \big(\#H^2(\CY,\Z)_{\mathrm{tors}}\big)\PTi_{m', \beta'}.  
\end{equation}
Combining \eqref{location of the wall}, \eqref{bound for ch3(E_0)} and Theorem \ref{thm-Cast} implies that $(\beta',m') \in M_{m, \beta}$ as defined in Theorem \ref{thm-main}.  

    Then summing up the wall crossing formulae \cite[Equation (5.13)]{Joyce:2008pc} over all walls for class $\v$ between the large volume limit and $\ell_f$ gives 
    \begin{align}\label{wc.1}
        0 & \ =  \bOm_{\,b,\, w \,< \,\alpha b -x}(\v) \nn\\ 
        & \ = \bOm_{\infty}(\v) \ + \\
        & \ \sum_{(m', \beta') \,\in\, M_{m, \beta}} (-1)^{ \chi(\v^1_{m', \beta'} , \v^0_{m',\beta'})+1}
        \chi(\v^1_{m', \beta'} , \v^0_{m',\beta'})
        \big(\#H^2(\CY,\Z)_{\mathrm{tors}}\big)\ \PTi_{m', \beta'}\ \bOmH\left(\v^0_{m', \beta'} \right). \nn
    \end{align}
    This implies 
    \begin{align}
        \PTi_{m, \beta} = \sum_{(m', \beta') \,\in\, M_{m, \beta}} (-1)^{ \chi(\v^1_{m', \beta'} , \v^0_{m',\beta'})}
        \chi(\v^1_{m', \beta'} , \v^0_{m',\beta'})\
        \PTi_{m', \beta'}\ \bOmH\left(\v^0_{m', \beta'} \right),
    \end{align}
    where $\chi(\v^1_{m', \beta'} , \v^0_{m',\beta'}) = \beta. H + \beta'. H +m -m' - \frac{H^3}{6} - \frac{1}{12}c_2(\CY).H = \chi_{m',\beta'}$ as claimed. 
\end{proof}

\begin{proof}[Proof of Theorem \ref{thm-quintic}]
Suppose $\mu = n \kappa + p$ where $n \in \mathbb{Z}_{\geq 0}$ and $p=0, \pm 1, \pm 2$. One can easily check that if $\mu \geq 13$ or $\mu =10$, the classes (i) $\left( \cC(\frac{\mu}{\kappa}H^2),\frac{\mu}{\kappa}H^2\right)$ and (ii) $\left( \cC(\frac{\mu}{\kappa}H^2)-1,\frac{\mu}{\kappa}H^2\right)$ for $p=\pm 1, \pm 2$, and if $\mu \geq 18$ the class (iii) $\left( \cC(\frac{\mu}{\kappa}H^2)-2,\frac{\mu}{\kappa}H^2\right)$ for $p=\pm 2$ are optimal in the sense of Definition \ref{def-optimal}. Thus combining Corollary \ref{cor-explicit formula} and Conjecture \ref{con} implies the claim. Note that when $\mu = n \kappa$, \cite[Theorem 1.1]{Feyzbakhsh:2022ydn} implies that
\begin{align}
	& \bOmH \left(0, H, \left(n+ \frac{1}{2}\right)H^2, \left(\frac{1}{6} + \frac{n(n+1)}{2}\right)H^3\right) 
\\
	=&\, (-1)^{ \chi(\cO((n+1)H), \cO(nH)) +1} \chi(\cO((n+1)H), \cO(nH)[1])  = 5.
\nn
\end{align}  
	If $\mu= 13$ or $17$, one can directly apply the wall-crossing formula \ref{pt-thm} in Theorem \ref{thm-main} to show again $\beta' = 0$ which implies the claim. 
\end{proof}

\begin{remark}\label{remark weaker BMT}
To prove Theorem \ref{thm-main} and Theorem \ref{thm-Cast} (and so Corollary \ref{cor-explicit formula} and \ref{Cor-explicit}, and Theorem \ref{thm-quintic}), we applied the conjectural BMT inequality only for (i) the rank zero classes $\v^0_{m', \beta'}$ to get \eqref{rank zero bound} using \cite[Lemma B.3]{Feyzbakhsh:2021rcv}, and (ii) rank $-1$ class $(-1, 0, \beta, -m)$ with $\beta.H \leq H^3$ to get \eqref{rank -1 bound}. Thus we only need the conjectural BMT inequality for the following two cases: 
\begin{enumerate}
    \item[(i)] Rank zero Gieseker-stable sheaves of class $(0, H, \beta_1, m_1)$ for the values $(b,w) \in U$ lying along the line $\ell_1$ which is of slope $\frac{\beta_1.H}{H^3}$ and intersects $\partial U$ at two points with $b$-values $b'< b''$ so that $b'' -b' =1$.
    \item[(ii)] Rank one torsion-free sheaves of class $(1, 0, -\beta_2, m_2)$ for $\beta_2.H \leq H^3$ and the point $(b,w) = (-1, \frac{1}{2})$ along the boundary $\partial U$. 
\end{enumerate}
This, in particular, shows that the weaker version of BMT conjecture proved in \cite{li2019stability} for quintic $X_5$ and in \cite{liu2021stability} for $X_{4,2}$ is sufficient for our result.  
\end{remark}

\begin{remark}\label{remark stronger BG in}
Suppose our CY threefold $\CY$ satisfies a stronger version of classical Bogolomov-Gieseker inequality \eqref{BGineq}, i.e. there is a function $G \colon \mathbb{R} \rightarrow \mathbb{R}$ such that any slope-semistable sheaf $E$ satisfies $\frac{\ch_2(E).H}{\ch_0(E)H^3} \leq G\left(\frac{\ch_1(E).H^2}{\ch_0(E)H^3}\right)$ and $G(b) \leq \frac{b^2}{2}$ for all $b \in \mathbb{R}$. Then one can enlarge the space of weak stability conditions $U$ to $U_G \coloneqq \{ (b,w) \in \mathbb{R}^2 \colon w > G(b)\}$, and apply all the arguments in this section within the enlarged space $U_{G}$ instead of $U$. This, in particular, shows that the intersection of the line $\ell_f$ with $\partial U_{G}$ has $b$-values $b_1^{G}<b_2^{G}$ so that $b_2^{G}-b_1^G \geq b_2-b_1$, thus we can improve the function $f$ in Theorem \ref{thm-main}. For instance, for a quintic threefold, one can apply Li's version of stronger Bogolomov-Gieseker inequality \cite[Theorem 1.1]{li2019stability} to show that equation of the function $f$ can be improved to $\frac{x}{2} +1$ for any $x > 0$ (see Figures \ref{fig-Li} and \ref{figX5}).      
\end{remark}

\begin{proof}[Proof of Proposition \ref{lem-x=4}]
	Since $\alpha = \frac{x}{2}+1$, the line $\ell_f$ intersects $\partial U$ at two points with $b$-values $b_1=2 <b_2$ such that $b_2-b_1 \geq b_1$, where the inequality is strict if $x>4$. Then, using the same notations as in Lemma \ref{lem.destabilising objects}, it implies that $b_2r \leq c \leq 2(r+1)$. Thus if $x>4$, we can have (i) $(r, c) =(0, 1)$, or (ii) $(r, c) = (0, 2)$. If $x=4$, there is a third possibility (iii) $r=1$ and $c=4$. 
	
	First consider a wall $\ell$ of type (i). We know $\ell$ lies above or on $\ell_f$, so there is $w^+>2$ such that the point $(2, w^+) \in U$ lies just above the wall $\ell$. Since no wall for $E_1$ can cross the vertical line $b=2$, we get $ \bOm_{b =2,w^{+}}(\ch(E_1)) =\bOm_{\,b=2,\ w \,\rightarrow\, +\infty}(\ch(E_1))$. Thus the proof of Theorem \ref{thm-main} goes through. 
	
	
	In case (ii), we know that $\ch_{\leq 1}(E_1) = (-1, -2H)$. Thus $\varpi(E_1)$ lies on $\partial U$, so $\ch_2(E_1) = -2H^2$. Then \cite[Proposition 4.20(ii)]{feyz-cubic-des} implies that $E_1 = \cO_{\CY}(2H)[1]$, so $\frac{\ch_2(E_0).H^2}{H^3} = x+2$ and $\ch_3(E_0) = \frac{2}{3}\alpha x H^3 +\frac{4}{3}H^3 = \frac{H^3}{3}x^2 + \frac{2H^3}{3}x +\frac{4}{3}H^3$. Applying \cite[Lemma B.3]{Feyzbakhsh:2021rcv} for $E_3$ gives 
	\begin{equation*}
	    \frac{\ch_3(E_0)}{H^3} = \frac{1}{3}x^2 + \frac{2}{3}x + \frac{4}{3} \leq \frac{(x+2)^2}{4} +\frac{1}{3}
	\end{equation*}
	which holds only if $x=4$. If $\beta = 4H^2$ and $m=-8H^3$, then $\ch(E_0) = (0, 2H, 6H^2 , \frac{28}{3}H^3)$. Applying the conjectural BMT inequality \eqref{BMTineq} implies that the final wall for $E_0$ coincides with the line $\ell_f$ and this is the only wall that can happen for $E_0$ by \cite[Theorem 1.1]{Feyzbakhsh:2022ydn} where the destabilising factors are $\cO_{\CY}(4H)$ and $\cO_{\CY}(2H)[1]$.     
	
	In case (iii) when $x=4$, we know that $\varpi(E_0)$ lies along the wall, so $E_1$ is 
	of class $\ch_{\leq 2}(E_0) = (1, 4H, 8H^2)$. Since $\Delta_H(E_0) = 0$, there is no wall for $E_0$ up to the large volume limit by \cite[Corollary 3.10]{bayer2016space}, so $E_0$ is a slope-stable sheaf and $\ch_3(E_0) \leq \frac{32}{3}H^3$. The other factor is of class $\ch_{\leq 2}(E_1) = (-2, -4H, -4H^2)$. Applying the BMT inequality \eqref{BMTineq} at the point $(b = 2+\eps, w)$ on the wall, where $0 < \eps \ll 1$, implies that $\ch_3(E_1) \leq -\frac{8}{3}H^3$. Given that $\ch_3(E_0) +\ch_3(E_1) = 8H^3$, it follows that $\ch(E_0) = \ch(\cO_{\CY}(4H))$ and $\ch(E_1) = -\ch(\cO_{\CY}(2H)^{\oplus 2})$.  

    To summarise, case (i) only contributes to the walls 
    which are of the same form as described in Lemma \ref{lem.destabilising objects}. If $x>4$, this is the only case that we need to consider and so Theorem \ref{thm-main} is valid. But if $\beta = 4H^2$, cases (ii) and (iii)  contribute to the last wall $\ell_f$. Define $\v_1 \coloneqq \ch(\cO_\CY(2H)[1])$ and $\v_2 \coloneqq \ch(\cO_\CY(4H))$, and let $(b,w^{^{\pm}})$ be points in $U$ just above and below the final wall $\ell_f$. Then applying the wall-crossing formula \cite[Equation (16)]{Feyzbakhsh:2022ydn} shows that we have three contributions along this wall: 
	\begin{enumerate}
		\item $\{\v_1,\v_1+\v_2\}$ contribution (corresponding to case (ii)) to $\bOm_{\,b,w^{-}}(\v)$ is
		\begin{equation}\label{con1}
		    (-1)^{\chi(\v_1, \v_1+\v_2)+1}\chi(\v_1, \v_1+\v_2) \, \bOm_{\,b,w^{+}}(\v_1) \, \bOm_{\,b,w^{+}}(\v_1+\v_2). 
		\end{equation}
	 As explained in case (ii), we know that $\bOm_{\,b,w^{+}}(\v_i) = \bOmH(\v_i)=1$ for $i=1, 2$ and 
		\begin{equation*}
		\begin{split}
		\bOm_{\,b,w^{+}}(\v_1+\v_2) = & (-1)^{\chi(\v_1, \v_2)}\chi(\v_1, \v_2) \,	\bOm_{\,b,w^{+}}(\v_1) \,	\bOm_{\,b,w^{+}}(\v_2)
		\end{split} 
		\end{equation*}
       thus the contribution \eqref{con1} is $-\big(\chi(\cO_\CY, \cO_\CY(2H))\big)^2$.
		
		\item $\{2\v_1,\v_2\}$ contribution (corresponding to case (iii)) to $\bOm_{\,b,w^{-}}(\v)$ is 
	 $$2\chi(\cO_\CY(2H), \cO_\CY(4H)) \times \bOmH(2\v_1) \bOmH(\v_2), $$ 
	 where $\bOmH(2\v_1)=\frac14\bOmH(\v_1)=\frac14$ by \cite[Example 6.2]{Joyce:2008pc};
	 
		\item $\{\v_1,\v_1,\v_2\}$ contribution to $\bOm_{\,b,w^{-}}(\v)$ is 
		\begin{equation}
		\frac{1}{2} \big(\chi(\v_1,\v_2)\big)^2 \, \bOmH(\v_1)^2\, \bOmH(\v_2) = \frac{1}{2}  \big(\chi(\cO_\CY, \cO_\CY(2H))\big)^2.
		\end{equation}
		
	\end{enumerate} 
	Hence the overall contribution of the wall $\ell_f$ to $\bOm_{\,b,w^{-}}(\v)$ is 
	\begin{equation}
	-\frac{1}{2}\big(\chi(\cO_\CY, \cO_\CY(2H))\big)^2 + \frac{1}{2}\chi(\cO_\CY, \cO_\CY(2H)).
	\end{equation}
	Combining this with the wall-crossing formula \eqref{wc.1} implies the claim. 
\end{proof}

\begin{proof}[Proof of Proposition \ref{prop.relax}]
    As before, we do wall-crossing for the class $\v = (-1, 0, \beta, -m)$. The same argument as in Lemma \ref{lem.destabilising objects} implies that the destabilising factors have Chern class 
	\begin{equation}
	\v^0_{m', \beta', H'} = \ch(E_0) = \left(0,\ H',\ \frac{1}{2}H'^2-\beta' +\beta ,\ \frac{1}{6}H'^3 +m'-m-\beta'. H'\right)
	\end{equation}
	and $\v^1_{m', \beta', H'}$ where
	\begin{equation}
	\v^1_{m', \beta', H'} \otimes \cO_{\CY}(-H') = \ch(E_1(-H')) = \left(-1,\ 0,\ \beta', \ -m'   \right).
	\end{equation}
	where $\frac{H'.H^2}{H^3} = 1$. Moreover $E_0$ is a Gieseker-stable sheaf as $\ch_1(E_0).H^2$ is still minimal and $E_1(-H')$ is $\nu_{b,w}$-stable for $b>0$ and $w \gg 1$, thus $\beta'.H \geq 0$. We know that the point $\varpi(\ch(E_1)) = (1\,, \, -\frac{\beta'. H}{H^3} + \frac{H.H'^2}{2H^3})$ (defined in \eqref{defPi}) lies above or on $\ell_f$, so
	\begin{equation}
	\frac{\beta'. H}{H^3}\,  \leq\,  \frac{3m}{2\beta. H} + \frac{\beta. H}{H^3} + \frac{H.H'^2}{2H^3}\,. 
	\end{equation} 	
 Moreover, we know that $E_1$ is $\nu_{b,w}$-stable for $b> -1$ and $w \gg 1$, so $(E_1 \otimes \det(E_1))^{\vee}[1]$ is a stable pair, thus \cite[Proposition 2.6]{Feyzbakhsh:2020wvm} implies that   
\begin{equation}
	-m'\ \leq\ \frac{2}{3}\beta'.H
	\left(\frac{\beta'.H}{H^3}+ \frac{1}{2}  \right). 
\end{equation}
On the other hand, applying \cite[Lemma B.3]{Feyzbakhsh:2021rcv} for $E_0$ gives
\begin{align}\label{rank zero bound}
   \frac{1}{6}H'^3 +m'-m-\beta'. H' \leq \frac{1}{2H^3}\left( \frac{1}{2}H'^2H-\beta'. H +\beta. H\right)^2 + \frac{H^3}{24}.
\end{align}
Then the claim follows by a similar argument as in the wall-crossing formula \eqref{wc.1}.
\end{proof}

\begin{proof}[Proof of Theorem \ref{thm.rk0}]
    The argument is similar to \cite{Feyzbakhsh:2022ydn}, we include it for completeness. Define $r = \frac{D.H^2}{H^3}$ and $s = \frac{\beta.H}{H^3}$, 
    then by our assumption $r \in \mathbb{Z}$. The conjectural BMT inequality implies that any wall $\ell$ for class $\v$ lies above or on the line $\ell_{f}$ with equation 
\begin{equation}\label{line ell-f}
	w = \frac{s}{r} b + \frac{r^2}{8} - \frac{s^2}{2r^2} - \frac{1}{4}Q_H(\v)
\end{equation}
which intersects $\partial U$ at two points with $b$-values  
\begin{equation}
	b_1= \frac{s}{r} - \sqrt{\frac{r^2}{4} - \frac{1}{2}Q_H(\v) },\quad 
	b_2= \frac{s}{r} + \sqrt{\frac{r^2}{4} - \frac{1}{2}Q_H(\v) } \,.
\end{equation}
Our assumption on $Q_H(\v)$ implies $b_2-b_1 > \max\{r-1,\ \frac{1}{2}\}$. Let $E_1 \rightarrow E \rightarrow E_2$ be a destabilising sequence along a wall $\ell$ for class $v$. By definition of the heart $\cA_b$,  
	\begin{equation}\label{b}
	\mu_H^+(\cH^{-1}(E_i)) \leq b_1 \qquad \text{and} \qquad b_2 \leq \mu_H^-(\cH^0(E_i)).
	\end{equation}
Summing up over $E_1$ and $E_2$ and using $\rk(E) = \rk(E_1) +\rk(E_2) = 0$ imply 
\begin{equation}
    r \geq (b_2-b_1)\big(\ch_0(\cH^{0}(E_1)) + \ch_0(\cH^{0}(E_2)) \big).
\end{equation}
Since $b_2-b_1 > \frac{r}{2}$, we get
	\begin{equation}
	\ch_0(\cH^{-1}(E_1)) + \ch_0(\cH^{-1}(E_2)) = \ch_0(\cH^{0}(E_1)) + \ch_0(\cH^{0}(E_2)) \leq 1.
	\end{equation}
Therefore, one of the factors $E_1$ is of rank $-1$ with $\cH^{-1}(E_1)$ of rank one and $\cH^0(E_1)$ of rank zero; and the other factor $E_2$ is a sheaf of rank one. 	We claim 
	\begin{equation}\label{lb}
	\mu_H(E_2) -b_2 < 1.
	\end{equation}
	Otherwise, \eqref{b} gives   
	\begin{equation}
	b_2-b_1 \leq \mu_H(E_2) - 1 - \mu_H(\cH^{-1}(E_1)).
	\end{equation}   	
	Since $\cH^{0}(E_1)$ is of rank zero, we have $\mu_H(E_1) \leq \mu_H(\cH^{-1}(E_1))$, thus
	\begin{equation}
	b_2-b_1 \leq \mu_H(E_2) -1 -\mu_H(E_1) = r-1. 
	\end{equation}
	The last equality comes from $\rk(E_2) = -\rk(E_1) = 1$. But the above is not possible by our assumption on $b_2-b_1$. Combining it with \eqref{b} gives
	\begin{equation}
	    b_2 \leq \mu_H(E_2) = \frac{\ch_1(E_2).H^2}{H^3} < b_2 +1.
	\end{equation}
	We know that there is no wall for $E_2$ crossing the vertical lines $b = \mu(E) - \frac{1}{2}$ and $b = \mu(E) -1$ \cite[Lemma 3.5]{feyz:effective-restriction-theorem}. Since $b_2-b_1 > \frac{1}{2}$ at least one of these vertical lines intersects the wall $\ell$ at a point inside $U$. Thus $\nu_{b,w}$-semistability of $E_2$ along the wall implies that $E_2$ is $\nu_{b,w}$-stable for $b < \mu(E_2)$ and $w \gg 1$, so $E_2$ is a rank one torsion-free sheaf. A similar argument also shows that $E_1$ is stable in the large volume limit, so is the derived dual of a stable pair (up to tensoring by a line bundle).   
    
    Hence the destabilising factors are of classes $\v_i = (-1)^i e^{D_i} (1, 0, -\beta_i, -m_i)$ for $i=1, 2$. We know that the point $\varpi(E_i) = \left(\frac{D_i.H^2}{H^3},\ \frac{D_i^2.H}{2H^3} -\frac{\beta_i.H}{H^3}  \right)$ lies above or on $\ell_f$, i.e.
    \begin{equation}\label{line ell-f-2}
	\frac{D_i^2.H}{2H^3} -\frac{\beta_i.H}{H^3} - \frac{\beta.H}{D.H^2} \frac{D_i.H^2}{H^3}  \geq  \frac{1}{8} \left(\frac{D.H^2}{H^3}\right) - \frac{1}{2}\left(\frac{\beta.H}{D.H^2}\right)^2 - \frac{1}{4}Q_H(\v).
\end{equation}
Finally applying \cite[Proposition 2.5 \& 2.6]{Feyzbakhsh:2020wvm} to $F_i \otimes D_i^{-1}$ imply 
\begin{equation}
	(-1)^{i+1}m_i\ \leq\ \frac{2}{3}\beta_i.H\left(\frac{\beta_i.H}{H^3}+ \frac{1}{2}  \right). 
	\end{equation}

Conversely, take two classes $\v_i = (-1)^i e^{D_i} (1, 0, -\beta_i, -m_i)$ for $i=1, 2$ satisfying $\v_1+\v_2 = \v$ and conditions \eqref{bound on Di} and \eqref{cc.1}. Then we have 
\begin{equation}
\left|\frac{D_1.H^2}{H^3} -b_1 \right| < 1,  
    \qquad 
    \left|\frac{D_2.H^2}{H^3} -b_2 \right| < 1
\end{equation}
and $\frac{D_2.H^2}{H^3} - \frac{D_1.H^2}{H^3} =r$. The Hodge index theorem implies 
\begin{equation}
    \frac{D_i^2.H}{2H^3} -\frac{\beta_i.H}{H^3} \leq \frac{1}{2}\left(\frac{D_i.H^2}{H^3}\right)^2, 
\end{equation}
thus $\varpi(E_i)$ lies outside $U$ and above or on $\ell_f$ by \eqref{cc.1}. Since by our assumption on $Q_H(\v)$, we have $b_2-b_1 > \max\{r-1,\ \frac{1}{2}\}$, we get 
\begin{equation}
   b_1-1 < \frac{D_1.H^2}{H^3} < b_1 <b_2 < \frac{D_2.H^2}{H^3} < b_2+1\,. 
\end{equation}
Then the same argument as above shows that there is no wall for classes $\v_1$ and $\v_2$ above or on $\ell_f$. Hence large volume limit stable objects of classes $\v_i$ for $i=1, 2$ are $\nu_{b,w}$-stable of the same $\nu_{b,w}$-slope along the line $\ell$ passing through $\varpi(\v_i)$ for $i=1, 2$, which lies above or on $\ell_f$. Thus they make a wall for objects of class $\v$. 
This completes the proof of the claim.   
\end{proof}

\section{Other hypergeometric CY threefolds}
\label{sec_gen}
In this section, we extend the analysis of \S\ref{sec_test} to the other hypergeometric CY threefolds, with the exception of $X_{3,2,2}$ and $X_{2,2,2,2}$ for which the current knowledge of GV invariants is not sufficient yet to uniquely determine (or even guess) the generating series of Abelian D4-D2-D0 indices. In all cases, we assume that the BMT inequality is satisfied.

\subsection{$X_6$}

Here we consider the sextic in $\IP^4_{2,1,1,1,1}$, first studied in \cite{Gaiotto:2007cd}. In this case, $\kappa=3$, $n_1^p=4$ and $n_1^c=0$ so the modular form $h_\mu$ is
uniquely fixed by 4 coefficients. Using the basis \eqref{decomp-modform}, the 
generating function proposed in 
\cite{Gaiotto:2007cd} reads
\be
\label{hmuX6}
\begin{split}
h_{\mu}=&\, \frac{1}{\eta^{54}} \[\frac{7 E_4^6+58 E_4^3 E_6^2+7 E_6^4}{216}
+\frac{5 E_4^4 E_6+3 E_4 E_6^3}{2}\, D\] \vths{3}_{\mu},
\end{split}
\ee
and has the following expansion:
\be
\begin{split}
\hspace{-0.2cm}
h_{0} =&\,  \q^{-\frac{15}{8}}
\,\Bigl( \underline{-4 +612 \q }-  40392 \q^2 +146464860 \q^3 +66864926808 \q^4
\\ &\, 
+8105177463840 \q^5+503852503057596 \q^6+20190917119833144 \q^7+ \dots\Bigr),
\\
\hspace{-0.2cm}
h_{1} =&\,  \q^{-\frac{15}{8}+\frac{2}{3}} \,\Bigl( \underline{0 - 15768 \q}
+7621020 \q^2 + 10739279916 \q^3 +1794352963536 \q^4
\\ &\, 
+134622976939812 \q^5+6141990299963544 \q^6+196926747589177416 \q^7+\dots
\Bigr).
\end{split}
\label{exph1X6}
\ee
Using \eqref{thmS11inv}, we can rigorously compute and confirm the terms up to (and including) order $\q^9$ and $\q^6$ in these expansions.  The term of order $\q^6$ in $h_0$ can be further verified using Prop. \ref{lem-x=4}. Furthermore, the terms of order $\q^{10}$ and $\q^{11}$ in $h_0$ as well as $\q^7$
and $\q^8$ in $h_1$ are reproduced by \eqref{thmS11inv} with $k=k_0-1$. Thus, there is overwhelming evidence that \eqref{hmuX6} is correct.  While the maximal genus attainable by the standard direct integration method is 48, using modularity, we can predict
GV invariants close to the Castelnuovo bound to arbitrary genus (see Table \ref{table_GVX6}), and provide sufficiently many boundary conditions in principle 
to push the
direct integration method up to genus 63.

\begin{table}
\be
\hspace*{-6mm}
\begin{array}{|r|r|rrrrrr|}
\hline
Q & g_C(Q) & \delta=0 &  \delta=1 &  \delta=2 &  \delta=3 &  \delta=4 & \delta=5\\
\hline
 1 & 1 & 0 & 7884 & \text{} & \text{} & \text{} & \text{} \\
 2 & 2 & 0 & 7884 & 6028452 & \text{} & \text{} & \text{} \\
 3 & 4 & 6 & 576 & 17496 & 145114704 & 11900417220 & \text{} \\
 4 & 5 & 0 & -47304 & -14966100 & 10801446444 & 1773044322885 & 34600752005688 \\
 5 & 7 & 0 & 63072 & 22232340 & -21559102992 & 1985113680408 & 571861298748384 \\
 6 & 10 & -28 & -3168 & -146988 & -583398600 & 207237771936 & -18316495265688 \\
 7 & 12 & 0 & -110376 & -43329384 & 54521267292 & -8041642037676 & 513634614205788 \\
 8 & 15 & 0 & -141912 & -57278448 & 76595605884 & -12434437188576 & 904511824896888 \\
 9 & 19 & -52 & -5472 & -225504 & -1453991342 & 645551751060 & -82281995054250 \\
 10 & 22 & 0 & 220752 & 90243180 & -132472407960 & 24441320028348 & -2094555362224356 \\
 11 & 26 & 0 & -268056 & -109069632 & 166408768980 & -32325403958928 & 2952049189946940 \\
 12 & 31 & 88 & 7572 & 212904 & 2755381840 & -1352963727576 & 204189584421816 \\
 13 & 35 & 0 & 378432 & 150306948 & -246695539464 & 52656199163280 & -5391865451528568 \\
 14 & 40 & 0 & 441504 & 172213236 & -293223343680 & 65474719151724 & -7076432910134952 \\
 15 & 46 & 136 & 7956 & 47736 & 4489872516 & -2384492136120 & 414897391102896 \\
 16 & 51 & 0 & -583416 & -217181952 & 399497240700 & -97481656444968 & 11697806611060704 \\
 17 & 57 & 0 & 662256 & 239613660 & -459419696640 & 117150837604344 & -14795431515539352 \\
 18 & 64 & -196 & -4680 & 225396 & -6665394192 & 3810518530344 & -758652854479632 \\
 19 & 70 & 0 & -835704 & -282637296 & 593248436100 & -165165188729184 & 23060985834155292 \\
 20 & 77 & 0 & -930312 & -302472360 & 667301101092 & -194106551379768 & 28471201009767792 \\
 21 & 85 & -268 & 4632 & 365112 & -9289038760 & 5717547855792 & -1296313683456384 \\
 22 & 92 & 0 & 1135296 & 336739140 & -829978779600 & 263496783986604 & -42580355264714232 \\
 23 & 100 & 0 & -1245672 & -350287848 & 918685187964 & -304661265971256 & 51631322400126468 \\
 \hline
\end{array}
\nonumber
\ee
\caption{GV invariants $\GVg{g_C(Q)-\delta}$ 
 for $X_{6}$, assuming modularity.}
\vspace{-0.65cm}
\label{table_GVX6}
\end{table}

\subsection{$X_8$}

\begin{table}[]
\be
\hspace*{-7mm}
\begin{array}{|r|r|rrrrrr|}
\hline
Q & g_C
& \delta=0 &  \delta=1 &  \delta=2 &  \delta=3 &  \delta=4 & \delta=5\\
\hline
 1 & 1 & 0 & 29504 & \text{} & \text{} & \text{} & \text{} \\
 2 & 3 & 6 & 864 & 41312 & 128834912 & \text{} & \text{} \\
 3 & 4 & 0 & -177024 & -16551744 & 21464350592 & 1423720546880 & \text{} \\
 4 & 7 & 24 & 4152 & 301450 & 396215800 & -174859503824 & 12499667277744 \\
 5 & 9 & 0 & 354048 & 37529088 & -86307810432 & 12063928269056 & -674562224718848 \\
 6 & 13 & 40 & 7032 & 523434 & 918424384 & -537735889892 & 67237956960504 \\
 7 & 16 & 0 & -649088 & -67977216 & 194884427520 & -34549033260480 & 2730733623512576 \\
 8 & 21 & 64 & 10760 & 747160 & 1693127408 & -1100325268755 & 163574439433328 \\
 9 & 25 & 0 & 1062144 & 97599232 & -348278532864 & 70573905748736 & -6573094863849216 \\
 10 & 31 & 96 & 14664 & 874648 & 2715237856 & -1885455097488 & 317498157747448 \\
 11 & 36 & 0 & -1593216 & -115655680 & 547020195328 & -124368823627264 & 13265837355895808 \\
 12 & 43 & 136 & 17880 & 816224 & 3983508192 & -2920617786752 & 550836611504760 \\
 13 & 49 & 0 & 2242304 & 107984640 & -791226604800 & 201252013167104 & -24393882174586624 \\
 14 & 57 & 184 & 19352 & 517696 & 5502562160 & -4239252796968 & 892029516487568 \\
 15 & 64 & 0 & -3009408 & -57591808 & 1080060791808 & -307623836581376 & 42224741744709120 \\
 16 & 73 & 240 & 17832 & -7064 & 7283098000 & -5881013303280 & 1377169141402320 \\
 17 & 81 & 0 & 3894528 & -55349504 & -1411208698624 & 450913093594624 & -69905571017188608 \\
 18 & 91 & 304 & 11880 & -637720 & 9339141568 & -7892176820432 & 2051229771888392 \\
 19 & 100 & 0 & -4897664 & 253498368 & 1780270216704 & -639463383246336 & 111688644307754752 \\
 20 & 111 & 376 & -136 & -1103312 & 11680881536 & -10326002693808 & 2969507704650056 \\
 21 & 121 & 0 & 6018816 & -562346240 & -2180065252608 & 882336333453824 & -173179252180073216 \\
 22 & 133 & 456 & -20040 & -932336 & 14301420112 & -13242624843432 & 4199274885440864 \\
 23 & 144 & 0 & -7257984 & 1010216960 & 2599854822400 & -1189005528876544 & 261596095595733504 \\
 24 & 157 & 544 & -49848 & 602936 & 17155557680 & -16707780101408 & 5821612261875808 \\
 25 & 169 & 0 & 8615168 & -1628266752 & -3024477174528 & 1568911793583616 & -386029168134457600 \\
 \hline
\end{array}
\nonumber
\ee
\caption{GV invariants $\GVg{g_C(Q)-\delta}$ 
 for $X_{8}$, assuming modularity.}
\vspace{-0.5cm}
\label{table_GVX8}
\end{table}

We now consider the octic in $\IP^4_{4,1,1,1,1}$, first studied in \cite{Gaiotto:2007cd}. In this case, $\kappa=2$, $n_1^p=4$ and $n_1^c=0$ so the modular form $h_\mu$ is
uniquely fixed by 4 coefficients. Using the basis \eqref{decomp-modform}, the 
generating function proposed in 
\cite{Gaiotto:2007cd} reads
\be
\label{hmuX8}
\begin{split}
h_{\mu}=&\, \frac{1}{\eta^{52}} \[\frac{103 E_4^6+1472 E_4^3 E_6^2+153 E_6^4}{5184}
+\frac{503 E_4^4 E_6+361 E_4 E_6^3}{108}\, D\]\vths{2}_{\mu},
\end{split}
\ee
and has the following expansion:
\be
\begin{split}
\hspace{-0.2cm}
h_{0}=&\, \q^{-\frac{46}{24}}\,\Bigl( \underline{-4 + 888 \q} - 86140 \q^2 +132940136 \q^3 +86849300500 \q^4
\\
&\, 
+11756367847000 \q^5+787670811260144 \q^6+33531427162546608 \q^7+\dots\Bigr),
\\
\hspace{-0.2cm}
h_{1}=&\, \q^{-\frac{46}{24}+\frac34}\,\Bigl( \underline{0 - 59008  \q} + 8615168 \q^2 +21430302976 \q^3
+3736977423872 \q^4
\\
&\, 
+289181439668352 \q^5+13588569634434304 \q^6+448400041603851008 \q^7+\dots\Bigr).
\end{split}
\label{exph1X8}
\ee
Using \eqref{thmS11inv}, we can rigorously compute and confirm the terms up to (and including) order $\q^9$ and $\q^7$ in these expansions. The term of order $\q^4$ in $h_0$ can be further verified using Prop. \ref{lem-x=4}. The terms of order $\q^{10}$, $\q^{11}$ in $h_0$ 
as well as $\q^8$, $\q^9$ and $\q^{10}$ in $h_1$ are reproduced by \eqref{thmS11inv} with $k=k_0-1$. Thus, there is overwhelming evidence that \eqref{hmuX8} is correct.  While the maximal genus attainable by the standard direct integration method is 60, using modularity, we can predict
GV invariants close to the Castelnuovo bound to arbitrary genus (see Table \ref{table_GVX8}), and provide sufficiently many boundary conditions in principle 
to push the
direct integration method up to genus 80.

\subsection{$X_{4,3}$}

We now consider the complete intersection of degree $(4,3)$ in $\IP^5_{2,1,1,1,1,1}$, In this case, $\kappa=6$, $n_1^p=9$ and $n_1^c=0$ so the modular form $h_\mu$ is
uniquely fixed by 9 coefficients. This model was 
first considered in \cite{Alexandrov:2022pgd}, assuming the naive Ansatz \eqref{naive} for the polar terms. Unfortunately, with the GV invariants being known only up to genus 20 using direct integration, Eq. \eqref{thmS11inv} only allows to determine 3 polar coefficients:
\be
\begin{split}
h_{0}=&\, \q^{-\frac94} \,\Bigl(\underline{5 - 624 \q + \tfrac{1}{21}\PT(18, -34) \q^2}
- \tfrac{1}{20} \PT(18, -33) \q^3
+\dots \Bigr),
\\
h_{1} =&\, \q^{-\frac94+\frac{7}{12}} \,\Bigl(\underline{-\tfrac{1}{12} \PT(13, -20)-\tfrac{1}{24} \PT(19, -38) \q}
+\tfrac{1}{23} \PT(19, -37) \q^2
+\dots \Bigr),
\\
h_{2} =&\, \q^{-\frac94+\frac13} \,\Bigl( \underline{-\tfrac{1}{14} \PT(14, -23)+\tfrac{1}{13} \PT(14, -22)\q}
-\tfrac{1}{26} \PT(20, -41) \q^2
+\dots \Bigr),
\\
h_{3} =&\, \q^{-\frac94+\frac14} \,\Bigl( \underline{0+\tfrac{1}{15} \PT(15, -25)\q}
-\tfrac{1}{14} \PT(15, -24) \q^2
+\dots \Bigr).
\end{split}
\label{predictX43}
\ee

Despite this discouraging result, one can proceed assuming that 
for some coefficients 
Eq. \eqref{thmS11inv} still holds for $k=k_0-1$.
This assumption will be justified {\it a posteriori} by matching 
the predictions of \eqref{thmS11inv} and modularity for many more coefficients.
For this choice of the spectral flow parameter, one finds
\be
\begin{split}
h_{0}\stackrel{?}{=}&\, \q^{-\frac94} \,\Bigl(\underline{2 -234 \q +35415\q^2}
+19018272\q^3
+\tfrac{523497643503}{7}\, \q^4
+\dots \Bigr),
\\
h_{1}\stackrel{?}{=}&\, \q^{-\frac94+\frac{7}{12}} \,\Bigl(\underline{0+\(5832 + \tfrac{40}{11} \GVg[13]{21}\)\q}
- \(544320 + 78 \GVg[13]{21}\) \q^2
\\
&\,
+\(3919923072 + \tfrac{9880}{9}\,\GVg[13]{21} \)\q^3
+\(2506521907872 - \tfrac{45695}{4} \,\GVg[13]{21}\)\q^4
+\dots \Bigr),
\\
h_{2} \stackrel{?}{=}&\, \q^{-\frac94+\frac13} \,\Bigl( \underline{0+0\q}
-\tfrac{1}{12} \(\GVg[14]{22} + 44 \GVg[14]{23} + 1035 \GVg[14]{24}\) \q^2
+\dots \Bigr),
\\
h_{3}\stackrel{?}{=}&\, \q^{-\frac94+\frac14} \,\Bigl( \underline{0+0\q}
+0\q^2
+\tfrac{1}{13} \PT(15, -23) \q^3
+\dots \Bigr),
\end{split}
\label{predictX43eps}
\ee
where we expressed the result in terms of GV invariants and put question marks to emphasize 
that these expansions need not be correct. For example, the first two terms
in $h_0$ clearly disagree with the rigorous result \eqref{predictX43}.
Nonetheless, let us assume that all other polar terms, except the $\cO(\q)$ term in $h_{2}$,
are correctly computed by \eqref{predictX43eps}. 
Comparing with \eqref{predictX43}, this implies the vanishing of 
$\PT(13, -20)=\GVg[13]{21}$ and $\PT(14, -23)=\GVg[14]{24}$, which allows to further simplify \eqref{predictX43eps}. 
In particular, all polar terms, except the $\cO(\q)$ term in $h_{2}$, are now fixed.

To get a sufficient number of conditions to fix the modular form, let us further assume that  the $\cO(\q^3)$ term in $h_{0}$ is
also correctly computed by \eqref{predictX43eps}. This assumption provides the missing condition and allows to find a unique 
modular form matching all coefficients 
\be
\label{hmuX43}
\begin{split}
h_{\mu}=&\, \frac{1}{\eta^{72}} \[\frac{19161576 E_4^7 E_6 - 86969808 E_4^4 E_6^3 - 36701208 E_4 E_6^5}{17199267840}
\right.
\\
&\,
+\frac{29888136 E_4^8 + 147874032 E_4^5 E_6^2 + 16326792 E_4^2 E_6^4}{716636160}\, D
\\
&\,
-\frac{4751784 E_4^6 E_6 + 9532080 E_4^3 E_6^3 + 646056 E_6^5}{5971968}\, D^2
\\
&\, \left.
-\frac{1686312 E_4^7 + 10686384 E_4^4 E_6^2 + 2557224 E_4 E_6^4}{1244160}\, D^3\]\vths{6}_{\mu},
\end{split}
\ee
with the following expansion:
\be
\begin{split}
h_{0}=&\, \q^{-\frac94} \,\Bigl( \underline{5 - 624 \q + \underline{35415} \q^2}
+19018272 \q^3+74785378407 \q^4
\\
&\, 
+23744184704784 \q^5+2912626940217084 \q^6+201892603398250080 \q^7 +\dots \Bigr),
\\
h_{1} =&\, \q^{-\frac94+\frac{7}{12}} \,\Bigl(\underline{0+5832 \q}
-544320 \q^2+3919923072 \q^3+2506521907872 \q^4
\\
&\, 
+426826821029328 \q^5+36510169956413184 \q^6+1975570599744644544 \q^7+\dots\Bigr),
\\
h_{2} =&\, \q^{-\frac94+\frac13} \,\Bigl( \underline{0+81\q}
-\dotuline{455544 \q^2}+418794867 \q^3+589406293224 \q^4
\\
&\, 
+127700521014312 \q^5+12611391702441624 \q^6+754527616229888955 \q^7 +\dots \Bigr),
\\
h_{3} =&\, \q^{-\frac94+\frac14} \,\Bigl(\underline{0+0\q}
-\dotuline{322496 \q^2}+154768800 \q^3+356674019472 \q^4+84550767361152 \q^5
\\
&\, 
+8789804684886144 \q^6+544775594940872640 \q^7+\dots \Bigr).
\end{split}
\label{exph1X43}
\ee
Remarkably, we find that all terms up to $\q^4$ in $h_1$ turn out to coincide with those in 
\eqref{predictX43eps}, which provides strong support for the above assumptions leading to \eqref{hmuX43}.
Furthermore, only the $\cO(\q^2)$ term in $h_0$ differs from the value $35334$ given by the naive ansatz \eqref{naive}, 
while all other polar terms as well as the  $\cO(\q^2)$ coefficients in $h_2$ and $h_3$ perfectly match \eqref{naive}.

We can apply a similar procedure to provide additional constraints on GV invariants and additional checks
on the modular function \eqref{hmuX43}. First, taking into account that
$\PT(14,-22)=\GVg[14]{23}+46 \GVg[14]{24}=\GVg[14]{23}$ and matching the $\cO(\q)$ term in $h_2$ between \eqref{predictX43}
and \eqref{exph1X43}, one obtains $\GVg[14]{23}=1053$. 
To get a constraint at genus 22, we further assume that the $\cO(\q^2)$ term in $h_2$ is correctly captured by \eqref{predictX43eps}.
Comparing it with \eqref{exph1X43} and taking into account the previous findings for GV invariants, 
one gets $\GVg[14]{22}=5420196$. With all these constraints, it is possible to compute GV invariants up to genus 23 
and check that the coefficients of $\q^3$ and $\q^4$ in $h_2$ computed using \eqref{thmS11inv} 
with $k=k_0-1$ match those in \eqref{exph1X43}, which can be considered as a verification 
of the above assumption. 

\begin{table}
\be
\begin{array}{|r|r|rrrrrr|}
\hline
Q & g_C & \delta=0 &  \delta=1 &  \delta=2 &  \delta=3 &  \delta=4 & \delta=5\\
\hline
 1 & 1 & 0 & 1944 &  &  &  & \\
 2 & 2 & 0 & 27 & 223560 &  &  &  \\
 3 & 3 & 0 & 0 & 161248 & 64754568 &  &  \\
 4 & 4 & 0 & 81 & 227448 & 381704265 & 27482893704 &  \\
 5 & 5 & 0 & 5832 & 155520 & 3896917776 & 638555324400 & 14431471821504 \\
 6 & 7 & 10 & 816 & 26757 & -40006768 & 75047188236 & 20929151321496 \\
 7 & 8 & 0 & -23328 & -1358856 & -7825332240 & 2609489667744 & 1159250594105376 \\
 8 & 10 & 0 & 405 & 1815696 & 1246578255 & -1193106464964 & 169353267859971 \\
 9 & 12 & 0 & 0 & -1612480 & -590680416 & 1077388111920 & -185398224083488 \\
 10 & 14 & 0 & 567 & 2719656 & 2033988975 & -2396370890772 & 426751496255367 \\
 11 & 16 & 0 & -46656 & -2503872 & -23437746576 & 13091629897584 & -1992347003533392 \\
 12 & 19 & 55 & 4260 & 139245 & -159384576 & 528260763000 & -157181565397200 \\
 13 & 21 & 0 & 64152 & 2908224 & 35118682704 & -21249125934480 & 3648284023741704 \\
 14 & 24 & 0 & 1053 & 5420196 & 4378100382 & -6073970861304 & 1376630062962426 \\
 15 & 27 & 0 & 0 & 4514944 & ? & ? & ? \\
 16 & 30 & 0 & 1377 & 7211592 & 5890376457 & ? & ? \\
 17 & 33 & 0 & 110808 & 2927664 & 66253494456 & -43956428447664 & ? \\
 18 & 37 & 115 & 5448 & 68415 & -405033180 & 1447081995873 & -519150013281888 \\
 19 & 40 & 0 & -139968 & -1881792 & -85700360016 & 58822283187000 & -12434745915614736 \\
 20 & 44 & 0 & 2187 & 11660436 & 9506092041 & -14945171094720 & 4071759470600148 \\
 21 & 48 & 0 & 0 & -9352384 & -3491811840 & 9974934265584 & -3041698928528400 \\
 22 & 52 & 0 & 2673 & 14310108 & 11567018943 & -18861070782672 & 5421323164985343 \\
 23 & 56 & 0 & -209952 & 3623616 & -132448385088 & 96512552546400 & -22792919002464096 \\
 24 & 61 & 205 & 360 & -124995 & -768252196 & 2853693391443 & -1169466146662224 \\
 25 & 65 & 0 & 250776 & -8736336 & 159825290616 & -119774566448496 & 29839722776131176 \\
 26 & 70 & 0 & 3807 & 20437272 & 16072226307 & -28239088327452 & 9015178386188196 \\
 \hline
\end{array}
\nonumber
\ee
\caption{GV invariants $\GVg{g_C(Q)-\delta}$ 
for $X_{4,3}$, assuming modularity. A question mark indicates that the result 
depends on as yet unknown PT invariants.}
\vspace{-0.5cm}
\label{table_GVX43}
\end{table}

To go to even higher genus, the predictions of modularity based on the rigorous use of \eqref{thmS11inv} 
are again insufficient because GV invariants at genus 24 turn out to depend on unknown PT invariants.
In particular, it can be shown that $\GVg[15]{24}=\PT(15,-23)-216717312$.
Fortunately, we can apply the same trick as above: let us assume that the $\cO(\q^3)$ term 
in $h_3$ is computed correctly by \eqref{predictX43eps}. This fixes the required PT invariant and 
gives  $\GVg[15]{24}=1795277088$. As a result, the direct integration method can be pushed up to genus 24, while the maximal genus attainable by the standard direct integration method is only 20.
One can also check that the coefficients of $\q^4$ and $\q^5$ in $h_3$ computed using \eqref{thmS11inv} 
with $k=k_0-1$ match those in \eqref{exph1X43}, which supports the above assumption. 
Finally, using modularity, we can predict GV invariants close to the Castelnuovo bound to arbitrary genus (see Table \ref{table_GVX43}).

\subsection{$X_{6,4}$}

\begin{table}
\be
\hspace*{-3mm}
\begin{array}{|r|r|rrrrrr|}
\hline
Q & g_C & \delta=0 &  \delta=1 &  \delta=2 &  \delta=3 &  \delta=4 & \delta=5\\
\hline
 1 & 1 & 8 & 15552 &  &  &  &  \\
 2 & 3 & 3 & 128 & 258344 & 27904176 &  &  \\
 3 & 4 & -48 & -64432 & 36976576 & 5966034472 & 133884554688 &  \\
 4 & 7 & 15 & 1036 & 800065 & -272993052 & 15929894952 & 4502079839576 \\
 5 & 9 & 96 & 160128 & -148759496 & 14847229472 & -592538522344 & 42148996229312 \\
 6 & 13 & 27 & 1784 & 1846330 & -838903420 & 76751964798 & -3326821152316 \\
 7 & 16 & -176 & -318240 & 338189520 & -43591449792 & 2519386074032 & -86921827226312 \\
 8 & 21 & 45 & 2456 & 3387175 & -1727130716 & ? & ? \\
 9 & 25 & 288 & 536160 & -610236992 & 91763910544 & -6449197272904 & ? \\
 10 & 31 & 69 & 2548 & 5409137 & -2981186776 & 388162502583 & -26079491452172 \\
 11 & 36 & -432 & -810912 & 970636496 & -166948527648 & 13842057435472 & -721713847987144 \\
 12 & 43 & 99 & 1412 & 7922463 & -4655472528 & 697407486327 & -55249696746420 \\
 13 & 49 & 608 & 1138656 & -1426615872 & 278955967328 & -26973362355200 & 1666567885265984 \\
 14 & 57 & 135 & -1744 & 10981213 & -6816859292 & 1167120353936 & -107938440865312 \\
 15 & 64 & -816 & -1514784 & 1986583568 & -440060692768 & 49244924907392 & -3567835755931072 \\
 16 & 73 & 177 & -7856 & 14707727 & -9547461076 & 1857133031696 & -198887588738688 \\
 17 & 81 & 1056 & 1933920 & -2659986752 & 665252326368 & -85577571342976 & 7205538528304192 \\
 18 & 91 & 225 & -18004 & 19321425 & -12950149776 & 2843375759861 & -349939554154236 \\
 19 & 100 & -1328 & -2389920 & 3457146192 & -972514271520 & 142892550610016 & -13868257921375616 \\
 20 & 111 & 279 & -33412 & 25171927 & -17157902216 & 4221401587493 & -592607749008964 \\
 \hline
\end{array}
\nonumber
\ee
\caption{GV invariants $\GVg{g_C(Q)-\delta}$ for $X_{6,4}$, assuming modularity.}
\vspace{-0.5cm}
\label{table_GVX64}
\end{table}

We now consider the complete intersection of degree $(6,4)$ in $\IP^5_{3,2,2,1,1,1}$, In this case, $\kappa=2$,  $n_1^p=3$ and $n_1^c=0$ so the modular form $h_\mu$ is uniquely fixed by 3 coefficients. This model was 
first considered in \cite{Alexandrov:2022pgd}, assuming the naive Ansatz \eqref{naive} for the polar terms. Using GV invariants up to genus 14, Eq. \eqref{thmS11inv} predicts \be
\begin{split}
\hspace{-0.6cm}
h_{0}=&\,\q^{-\frac{34}{24}}\,\Bigl(\underline{3-\underline{304} \q}+\tfrac{1}{13} \PT(8, -18) \q^2
-\tfrac{1}{12} \PT(8, -17) \q^3
+\dots\Bigr),
\\
\hspace{-0.6cm}
h_{1} =&\, \q^{-\frac{34}{24}+\frac34}\,\Bigl( \underline{-16} -\tfrac{1}{10} \PT(7, -14) \q
+\tfrac{1}{9}\( \PT(7, -13)+192\) \q^2
+\dots\Bigr).
\end{split}
\ee
In particular, the polar part of $h_0$ differs from the value 
$3-312\q$ predicted by the naive Ansatz \eqref{naive}.
There is a unique modular form that matches these polar terms, namely
\be
\begin{split}
h_{\mu}=&\, \frac{1}{\eta^{40}} \[-\frac{85 E_4^3 E_6+23 E_6^3}{432}
-\frac{13 E_4^4+23 E_4 E_6^2}{6}\, D\]\vths{2}_{\mu},
\end{split}
\ee
with the following expansion:
\be
\begin{split}
\hspace{-0.6cm}
h_{0}=&\,\q^{-\frac{34}{24}}\,\Bigl(\underline{3-\underline{304} \q}+270431 \q^2+133585104 \q^3+12401092398 \q^4
+\dots\Bigr),
\\
\hspace{-0.6cm}
h_{1} =&\, \q^{-\frac{34}{24}+\frac34} \,\Bigl(\underline{-16} +32352 \q+36578048 \q^2+4364892672 \q^3+226014399392 \q^4+\dots\Bigr).
\end{split}
\label{exphX64}
\ee
The term of order $\q$ in $h_1$ is correctly reproduced by \eqref{thmS11inv} with $k=k_0-1$.
Assuming that \eqref{hmuX62} is correct, one can produce additional boundary conditions for the direct integration method (see Table \ref{table_GVX64}), allowing to reach genus 17, beyond the genus 14 available by standard methods.
Note that to get a boundary condition at genus 17, one uses the fact the $\cO(\q^4)$ coefficient in $h_0$ is subject to Prop. \ref{lem-x=4}.

\subsection{$X_{3,3}$}

Next, we consider the bicubic in $\IP^5$, first studied in \cite{Gaiotto:2007cd}. In this case, $\kappa=9$,  $n_1^p=14$ and $n_1^c=1$ (as first noted in \cite{Manschot:2008zb}) so the modular form $h_\mu$  is
uniquely fixed by 13 coefficients. Using the basis \eqref{decomp-modform}, the 
generating function proposed in 
\cite{Gaiotto:2007cd} reads   
 \be
\begin{split}
h_{\mu}=&\, \frac{1}{\eta^{90}} \[\frac{47723 E_4^9 E_6+25095 E_4^6 E_6^3-68943 E_4^3 E_6^5-3875 E_6^7}{107495424}
\right.
\\
&\,
+\frac{289326 E_4^{10}+415189 E_4^7 E_6^2-3458324 E_4^4 E_6^4-729839 E_4 E_6^6}{334430208}\, D
\\
&\,
+\frac{2261629 E_4^8 E_6+3219046 E_4^5 E_6^3-6371 E_4^2 E_6^5}{30965760}\, D^2
\\
&\,
-\frac{94271 E_4^9+1496733 E_4^6 E_6^2+1342665 E_4^3 E_6^4+52315 E_6^6}{5160960}\, D^3
\\
&\,\left.
-\frac{162167 E_4^7 E_6+300338 E_4^4 E_6^3+35159 E_4 E_6^5}{286720}\, D^4\] \vths{9}_{\mu},
\end{split}
\label{hmuX33}
\ee
and has the following expansion
\be
\begin{split}
\hspace{-0.4cm}
h_{0}=&\, \q^{-\frac{63}{24}} \,\Bigl( -\underline{6 + 720 \q - 40032 \q^2}
-678474 \q^3 + 30885198768 \q^4+ 35708825468142 \q^5 + \dots \Bigr),
\\
\hspace{-0.4cm}
h_{1} =&\, \q^{-\frac{63}{24}+\frac59} \,\Bigl( \underline{0 - 4212 \q +448578 \q^2}
+374980104 \q^3 +2020724648442 \q^4+\dots \Bigr),
\\
\hspace{-0.4cm}
h_{2} =&\, \q^{-\frac{63}{24}+\frac29} \,\Bigl( \underline{0 + 0 \q + 158436 \q^2}
-12471246 \q^3 +174600085086 \q^4+\dots\Bigr),
\\
\hspace{-0.4cm}
h_{3} =&\, \q^{-\frac{63}{24}} \,\Bigl(\underline{0 + 0 \q +10206 \q^2}
-\dotuline{13828428 \q^3}+24425287884 \q^4+\dots \Bigr),
\\
\hspace{-0.4cm}
h_{4} =&\, \q^{-\frac{63}{24}+\frac89} \,\Bigl( \underline{0 + 0 \q}-\dotuline{11040786 \q^2}
+6769752552 \q^3 + 17629606262268 \q^4 +\dots\Bigr).
\end{split}
\label{h1X33}
\ee
Unfortunately, with GV invariants being known up to genus 29, Eq. \eqref{thmS11inv} only allows to confirm the coefficients $-6+720\q$ in $h_0$ and $0\,\q^0$ in $h_3$. Applying \eqref{thmS11inv} with $k=k_0-1$, we find evidence that 
the coefficients of all terms up to $\q^5$ in $h_{0}$ and all the vanishing coefficients in other components are indeed correct. 
Moreover, the $\cO(\q^3)$ coefficient in $h_3$ and $\cO(\q^2)$ coefficient in $h_4$ as well as all polar terms turn out to agree with the ansatz \eqref{naive}. 
In addition, we observe that the coefficients $-4212$ in $h_1$ and 10206 in $h_3$ are given by $\frac15\PT(10,-9)$ and $-\frac16\PT(12,-12)$, even though the corresponding values of $(Q,m)$ do not satisfy the optimality conditions. 
Thus, there is strong evidence that \eqref{h1X33} is correct.

\begin{table} 
\be
\hspace*{-5mm}
\begin{array}{|r|r|rrrrrr|}
\hline
Q & g_C(Q) & \delta=0 &  \delta=1 &  \delta=2 &  \delta=3 &  \delta=4 & \delta=5\\
\hline
 1 & 1 & 0 & 1053 &  &  &  &  \\
 2 & 2 & 0 & 0 & 52812 &  &  &  \\
 3 & 3 & 0 & 0 & 3402 & 6424326 &  &  \\
 4 & 3 & 0 & 0 & 5520393 & 1139448384 &  &  \\
 5 & 4 & 0 & 0 & 5520393 & 4820744484 & 249787892583 &  \\
 6 & 6 & 0 & 0 & 10206 & 6852978 & 23395810338 & 3163476682080 \\
 7 & 7 & 0 & 0 & 158436 & -484542 & 174007524240 & 42200615912499 \\
 8 & 8 & 0 & 6318 & 372762 & -784819773 & 2028116431098 & 785786604262830 \\
 9 & 10 & 15 & 1170 & 39033 & -5412348 & -61753761036 & 36760497856020 \\
 10 & 11 & 0 & -21060 & -1421550 & 1150458714 & -4055688274977 & 1055748342473838 \\
 11 & 13 & 0 & 0 & 792180 & 42487254 & 523544632866 & -277740359622189 \\
 12 & 15 & 0 & 0 & -61236 & -67672476 & -96817818078 & 107933688748656 \\
 13 & 16 & 0 & 0 & 66244716 & 32180134734 & -71248361250798 & 17551409134469472 \\
 14 & 18 & 0 & 0 & -77285502 & -38299950252 & 89193730254030 & -23552769634742655 \\
 15 & 21 & 0 & 0 & 91854 & 107320680 & 167270244048 & -217376516354913 \\
 16 & 23 & 0 & 0 & -1584360 & -48866814 & -1393793916300 & 990222417035712 \\
 17 & 25 & 0 & 50544 & 2609334 & -3916924776 & 18349298486658 & -7855011831413205 \\
 18 & 28 & -90 & -5220 & -120186 & 23305068 & 338860808028 & -372702765685392 \\
 19 & 30 & 0 & -63180 & -2653560 & 5125104738 & -24509155811472 & 11014900785838314 \\
 20 & 33 & 0 & 0 & -2534976 & -34970130 & -2437828042176 & 1882564212119436 \\
 21 & 36 & 0 & 0 & -183708 & -222958548 & ? & ? \\
 22 & 38 & 0 & 0 & -209774934 & -107171300556 & ? & ? \\
 23 & 41 & 0 & 0 & -231856506 & -117773956584 & ? & ? \\
 24 & 45 & 0 & 0 & -244944 & -297478548 & -512109217728 & 784094829426108 \\
 25 & 48 & 0 & 0 & -4119336 & 58959090 & -4179092501304 & 3448552834527066 \\
 26 & 51 & 0 & -122148 & -589680 & 10761641532 & -53529168000492 & 27142164772551882 \\
 27 & 55 & -198 & -1656 & 78588 & 61819596 & 890324824482 & -1089181917906228 \\
 28 & 58 & 0 & 143208 & -1236222 & -12753357660 & 63994621219614 & -33427151297813844 \\
 29 & 62 & 0 & 0 & -5703696 & 236542518 & -5924469211524 & 5107223091368232 \\
 30 & 66 & 0 & 0 & 398034 & 475327980 & 841660464438 & -1378633833342540 \\
 31 & 69 & 0 & 0 & -452672226 & -211037368248 & 718903306166688 & -291661896939934680 \\
 32 & 73 & 0 & 0 & 485794584 & 223088117976 & -775956404598264 & 320699529577227510 \\
 33 & 78 & 0 & 0 & -489888 & -576452916 & -1032773712696 & 1747571748926544 \\
 34 & 82 & 0 & 0 & 8080236 & -653114988 & 8557070940234 & -7761488275449180 \\
 35 & 86 & 0 & -231660 & 14695668 & 20684191104 & -108491328034740 & 62714086906118814 \\
 \hline
\end{array}
\nonumber
\ee
\caption{GV invariants $\GVg{g_C(Q)-\delta}$ 
for $X_{3,3}$, assuming modularity. A question mark indicates that the result 
depends on as yet unknown PT invariants.}
\vspace{-0.5cm}
\label{table_GVX33}
\end{table}

Assuming that it is, one can produce additional boundary conditions for the direct integration method (see Table \ref{table_GVX33}), allowing to reach genus 33, beyond the genus 29 available using standard boundary conditions.

\subsection{$X_{4,4}$}

We now consider the complete intersection of degree $(4,4)$ in $\IP^5_{2,2,1,1,1}$, In this case, $\kappa=4$,  $n_1^p=6$ and $n_1^c=1$ so the modular form $h_\mu$ is uniquely fixed by 5 coefficients. This model was 
first considered in \cite{Alexandrov:2022pgd}, assuming the naive Ansatz \eqref{naive} for the polar terms. Using GV invariants up to genus 26, Eq. \eqref{thmS11inv} predicts \be
\begin{split}
\hspace{-0.2cm}
h_{0}=&\, \q^{-\frac{44}{24}} \,\Bigl( \underline{-4 + 432 \q}
-10032 \q^2 + 148611456 \q^3 -\tfrac{1}{24} \PT(16, -36)_4 \q^4
+\dots\Bigr),
\\
\hspace{-0.2cm}
h_{1} =&\, \q^{-\frac{44}{24}+\frac58} \,\Bigl( \underline{0 
+ \tfrac{1}{17} \PT(13,-26) \q} -\tfrac{1}{16} \PT(13,-25) \q^2 
+\dots\Bigr),
\\
\hspace{-0.2cm}
h_{2} =&\, \q^{-\frac{44}{24}+\frac12} \,\Bigl(\underline{0 -2816 \q}
+\tfrac{1}{19} \PT(14,-29) \q^2 +  \dots\Bigr). 
\end{split}
\ee
This is sufficient information to fully determine the generating series:
\be
\begin{split}
h_{\mu}=&\, \frac{1}{\eta^{56}} \[\frac{319 E_4^5 E_6+113 E_4^2 E_6^3}{11664 }
-\frac{146 E_4^6+1025 E_4^3 E_6^2+125 E_6^4}{972}\, D
\right.
\\
&\,\left.
-\frac{566 E_4^4 E_6+298 E_4 E_6^3}{81}\, D^2\]\vths{4}_{\mu},
\end{split}
\label{hmuX44}
\ee
with the following expansion:
\be
\begin{split}
\hspace{-0.2cm}
h_{0}=&\, \q^{-\frac{44}{24}}\,\Bigl(\underline{-4 + 432 \q} - 10032 \q^2 + 148611456 \q^3 +53495321332 \q^4
\\
&\, 
+5858228664240 \q^5+338470263518000 \q^6 +
+12773210724578176 \q^7 \\
&\, + 352882974651781356 \q^8 \dots\Bigr),
\\
\hspace{-0.2cm}
h_{1} =&\, \q^{-\frac{44}{24}+\frac58} \,\Bigl(\underline{0 - 7424  \q} +7488256 \q^2 +7149513728 \q^3
+1104027086592 \q^4
\\
&\, 
+78370863237632 \q^5+3411805769659904 \q^6 +\dots\Bigr),
\\
\hspace{-0.2cm}
h_{2} =&\, \q^{-\frac{44}{24}+\frac12} \,\Bigl( \underline{0 - 2816  \q} +2167680 \q^2 + 3503031296 \q^3
+619015800576 \q^4
\\
&\, 
+47430532268544 \q^5+2174342476769792 \q^6 +\dots\Bigr).
\end{split}
\label{exp-hX44}
\ee
In particular, the polar part agrees with the naive Ansatz \eqref{naive} that was assumed in \cite{Alexandrov:2022pgd}.
Assuming that \eqref{hmuX44} is correct, one can produce additional boundary conditions for the direct integration method (see Table \ref{table_GVX44}) allowing to reach genus 32. 
Furthermore, the term of order $\q^8$ in $h_0$ is subject to Prop. \ref{lem-x=4} and provides an additional boundary condition at genus 33 that allows to push the direct integration 
up to genus 34.
With this new data, we can further check
that \eqref{exp-hX44} is consistent up to orders $\q^3$, $\q^4$ and $\q^5$ with \eqref{thmS11inv}, and even reproduce the coefficients of order $\q^5$ in $h_1$ and $\q^6$ in $h_2$ by applying \eqref{thmS11inv} with $k=k_0-1$. Thus, there is overwhelming
evidence that \eqref{hmuX44} (first conjectured in \cite{Alexandrov:2022pgd}) is indeed  correct.

\begin{table}
\be
\hspace*{-2mm}
\begin{array}{|r|r|rrrrrr|}
\hline
Q & g_C & \delta=0 &  \delta=1 &  \delta=2 &  \delta=3 &  \delta=4 & \delta=5\\
\hline
 1 & 1 & 0 & 3712 &  &  &  &  \\
 2 & 2 & 0 & 1408 & 982464 &  &  &  \\
 3 & 3 & 0 & 3712 & 6953728 & 683478144 &  &  \\
 4 & 5 & 6 & 384 & -12432 & 148208928 & 26841854688 & 699999511744 \\
 5 & 6 & 0 & -22272 & -14802048 & 7282971392 & 2161190443904 & 88647278203648 \\
 6 & 8 & 0 & 11264 & 6367872 & -7046285440 & 773557598272 & 362668189458048 \\
 7 & 10 & 0 & -37120 & -29359104 & 21832649216 & -2470237776768 & 278617066306304 \\
 8 & 13 & 32 & 2256 & 1728 & 742436816 & -227235799678 & 21187753811008 \\
 9 & 15 & 0 & 59392 & 50769664 & -44144389120 & 6476935523072 & -435143766495232 \\
 10 & 18 & 0 & 28160 & 18608000 & -28596423936 & 5125410035840 & -407275256652416 \\
 11 & 21 & 0 & 89088 & 78985472 & -74401243136 & 12415279501056 & -983151655520000 \\
 12 & 25 & 64 & 3408 & -88512 & 1931209232 & -723035097878 & 91699351475728 \\
 13 & 28 & 0 & -126208 & -113249280 & 113072299008 & -20742352242176 & 1847677262046464 \\
 14 & 32 & 0 & 56320 & 37806720 & -65347039488 & 14118281042560 & -1421834838533888 \\
 15 & 36 & 0 & -170752 & -153130496 & 160635374080 & -32024183351808 & 3157850965939456 \\
 16 & 41 & 112 & 2704 & -304000 & 3721068368 & -1558763217664 & 236934426952368 \\
 17 & 45 & 0 & 222720 & 198005504 & -217632888320 & 46941587427584 & -5089550372194304 \\
 18 & 50 & 0 & 95744 & 62342016 & -118393743616 & 29795537375872 & -3617592752039168 \\
 19 & 55 & 0 & 282112 & 247161600 & -284657253888 & 66297765350656 & -7874526931335680 \\
 20 & 61 & 176 & -2544 & -535104 & 6123275152 & -2835063806944 & 506679656992912 \\
 21 & 66 & 0 & -348928 & -299796992 & 362338625536 & -91027164419584 & 11813373104231424 \\
 22 & 72 & 0 & 146432 & 90051200 & -189003413760 & 55119444471424 & -7968454696971008 \\
 23 & 78 & 0 & -423168 & -355020800 & 451332651008 & -122204523786752 & 17290687427825664 \\
 24 & 85 & 256 & -15792 & -347712 & 9139883728 & -4682106927504 & 976043266192272 \\
 25 & 91 & 0 & 504832 & 411853056 & -552306662400 & 161054029205248 & -24792539565154304 \\
 26 & 98 & 0 & 208384 & 118230912 & -278503055616 & 93943868100224 & -16042635201490176 \\
 \hline
\end{array}
\nonumber
\ee
\caption{GV invariants $\GVg{g_C(Q)-\delta}$ for $X_{4,4}$, assuming modularity.}
\vspace{-0.5cm}
\label{table_GVX44}
\end{table}

\subsection{$X_{6,6}$}

\begin{table}
\be
\begin{array}{|r|r|rrrrrr|}
\hline
Q & g_C & \delta=0 &  \delta=1 &  \delta=2 &  \delta=3 &  \delta=4 & \delta=5\\
\hline
 1 & 2 & 1 & 360 & 67104 &  &  &  \\
 2 & 4 & -6 & -928 & 291328 & 40692096 & 847288224 &  \\
 3 & 7 & -10 & -1807 & 867414 & -39992931 & 1253312442 & 254022248925 \\
 4 & 11 & 16 & 3054 & -1752454 & 111434794 & -3192574724 & 53221926192 \\
 5 & 16 & 24 & 4582 & -2962836 & 226181014 & -8162501599 & 181541450026 \\
 6 & 22 & -34 & -6284 & 4516784 & -401198640 & 17316022722 & -470838831620 \\
 7 & 29 & -46 & -8028 & 6434962 & -657358676 & 33294527348 & -1078394245876 \\
 8 & 37 & 60 & 9658 & -8736900 & 1020136914 & -59981343076 & 2281585927834 \\
 9 & 46 & 76 & 10994 & -11438612 & 1519838840 & -102904929012 & 4549405838854 \\
 10 & 56 & -94 & -11832 & 14549836 & -2191738688 & 169716778670 & -8644920617316 \\
 11 & 67 & -114 & -11944 & 18070914 & -3076075680 & 270744962214 & -15768167683888 \\
 12 & 79 & 136 & 11078 & -21989312 & 4217848666 & -419620144388 & 27750103729188 \\
 13 & 92 & 160 & 8958 & -26275780 & 5666343644 & -633965979716 & 47309465417064 \\
 14 & 106 & -186 & -5284 & 30880152 & -7474321920 & 936139362212 & -78385363446040 \\
 15 & 121 & -214 & 268 & 35726786 & -9696789948 & 1353997949560 & -126558933123332 \\
 \hline
\end{array}
\nonumber
\ee
\caption{GV invariants $\GVg{g_C(Q)-\delta}$ for $X_{6,6}$, assuming modularity.}
\vspace{-0.5cm}
\label{table_GVX66}
\end{table}

We now consider the complete intersection of degree $(6,6)$ in $\IP^5_{3,3,2,2,1,1}$, In this case, $\kappa=1$,  $n_1^p=1$ and $n_1^c=0$ so the scalar modular form $h=h_0$ is uniquely fixed by a single coefficient. Since the leading coefficient is known, the generating series is necessarily \cite{Alexandrov:2022pgd}
\be
\begin{split}
h=
-\frac{2E_4 E_6}{\eta^{23}}
=&\,  \q^{-\frac{23}{24}} \,\Bigl( - \underline{2}
+482 \q + 282410 \q^2 + 16775192 \q^3
+ 460175332 \q^4
\\
&\, 
+8112401426 \q^5+106227128612 \q^6+1118140132310 \q^7 +\dots \Bigr).
\end{split}
\label{hmuX66}
\ee
Using GV invariants up to genus 18, we can use Eq. \eqref{thmS11inv} to confirm all terms
up to (and including) $\q^3$. The $\cO(\q^2)$ coefficient
can also be verified independently using Prop. \ref{lem-x=4}. 
We note that the coefficient $J_1=482$ differs
from the naive prediction $\chi_\CY(\chi_D-1)=-120$, due to the singular curve where the two degree-one homogeneous coordinates vanish simultaneously.
Assuming \eqref{hmuX66} is correct, one can produce additional boundary conditions for the direct integration method (see Table \ref{table_GVX66}), allowing to reach genus 22, beyond the genus 18 available by standard methods.

\subsection{$X_{6,2}$}

\begin{table} 
\be
\hspace*{-1mm}
\begin{array}{|r|r|rrrrrr|}
\hline
Q & g_C & \delta=0 &  \delta=1 &  \delta=2 &  \delta=3 &  \delta=4 & \delta=5\\
\hline
 1 & 1 & 0 & 4992 &  &  &  &  \\
 2 & 2 & -4 & -504 & 2388768 &  &  &  \\
 3 & 3 & 0 & 14976 & 1228032 & 2732060032 &  &  \\
 4 & 5 & 10 & 1456 & 87376 & -13098688 & 79275664800 & 4599616564224 \\
 5 & 6 & 0 & -59904 & -7098624 & -5731751168 & 3921835430016 & 633074010435840 \\
 6 & 8 & -36 & -7176 & 18680344 & 1776341072 & -3978452463012 & 482407033529880 \\
 7 & 10 & 0 & -89856 & -11017344 & -11354017792 & 11762488063616 & -1739233315959552 \\
 8 & 13 & 45 & 7112 & 505904 & -26309632 & 397727244436 & -143734260919104 \\
 9 & 15 & 0 & 134784 & 16463616 & 19737387904 & -23804240518144 & 4517211164682496 \\
 10 & 18 & -72 & -14448 & 46979932 & 4192707384 & -16032046818880 & 3592452810930880 \\
 11 & 21 & 0 & 194688 & 22853376 & 30801879680 & -40000460268544 & 8544691377500288 \\
 12 & 25 & 85 & 12304 & 783986 & -82475800 & 1034830295100 & -453574307495648 \\
 13 & 28 & 0 & -269568 & -28838784 & -44399031040 & 60578209825920 & -14063973371548160 \\
 14 & 32 & -132 & -24072 & 94648652 & 6368703752 & -36273186323128 & 9662743440528720 \\
 15 & 36 & 0 & -359424 & -33541248 & -60489784576 & 85729492457728 & -21374438244874240 \\
 16 & 41 & 145 & 17072 & 812714 & -187672824 & 1995257447554 & -967713951224848 \\
 17 & 45 & 0 & 464256 & 35702784 & 79038018176 & -115677339047680 & 30834343468656896 \\
 18 & 50 & -216 & -32592 & 162178504 & 5450058896 & -64779929898136 & 19857403612132080 \\
 19 & 55 & 0 & 584064 & 33885696 & 100058488192 & -150674169631488 & 42863166970805504 \\
 20 & 61 & 225 & 18056 & 407662 & -336445208 & 3284718881588 & -1740135148941248 \\
 21 & 66 & 0 & -718848 & -26472576 & -123640960768 & 191002891152640 & -57945747882503680 \\
 22 & 72 & -324 & -35400 & 249881836 & -2285410296 & -101366274190332 & 35739934943323496 \\
 23 & 78 & 0 & -868608 & -11666304 & -149974972160 & 236979848907520 & -76636691924665088 \\
 24 & 85 & 325 & 10936 & -301906 & -502103864 & 4910704750656 & -2840815244088832 \\
 25 & 91 & 0 & 1033344 & -12509952 & 179377684096 & -288959785268736 & 99565363235000576 \\
 \hline
\end{array}
\nonumber
\ee
\caption{GV invariants $\GVg{g_C(Q)-\delta}$ for $X_{6,2}$, assuming modularity.}
\vspace{-0.5cm}
\label{table_GVX62}
\end{table}

Finally, we consider the complete intersection of degree $(6,2)$ in $\IP^5_{3,1,1,1,1}$, In this case, $\kappa=4$,  $n_1^p=7$ and $n_1^c=0$ so the modular form $h_\mu$  is uniquely fixed by 7 coefficients. This model was 
first considered in \cite{Alexandrov:2022pgd}, assuming the naive Ansatz \eqref{naive} for the polar terms. Using GV invariants up to genus 47, Eq. \eqref{thmS11inv} and Prop. \ref{lem-x=4} predict  \be
\begin{split}
\hspace{-0.4cm}
h_{0}=&\,\q^{-\frac{56}{24}} \,\Bigl(\underline{5-1024 \q+96390 \q^2}
+2412544 \q^3+79408559682 \q^4
\\
&\, 
+34353222823936 \q^5
+4968007484511900 \q^6
+389580600939126784 \q^7\\
&\,
+ 20087040094321343657 \q^8 
+\dots\Bigr),
\\
\hspace{-0.4cm}
h_{1} =&\, \q^{-\frac{56}{24}+\frac58}\,\Bigl( \underline{0+14976 \q}
-2520960 \q^2+2887376128 \q^3 +3893723178368 \q^4
\\
&\, 
+809149241398912 \q^5 + 78688042019771776 \q^6 
+ 4713543813612260224 \q^7 
\\
&\, 
+ 198770720341455440256 \q^8
+\dots\Bigr),
\\
\hspace{-0.4cm}
h_{2} =&\, \q^{-\frac{56}{24}+\frac12}\,\Bigl( 
\underline{\underline{6}-\underline{1536} \q}
-4647736 \q^2+621617152 \q^3 +1986721226130 \q^4
\\
&\, 
+ 453923870489088_3 \q^5
+\tfrac{1}{30}\, (-2242806300 - \PT(18, -43))\q^6
+\dots\Bigr).
\end{split}
\ee
In particular, the polar coefficients in $h_2$ differ from the values 
$16-4608\q$ predicted by the naive Ansatz \eqref{naive}. There is indeed 
a unique modular form which fits this vastly overdetermined set of coefficients:
\be
\begin{split}
h_{\mu}=&\, \frac{1}{\eta^{68}} \[-\frac{5(727 E_4^8+3322 E_4^5 E_6^2+1135 E_4^2 E_6^4)}{559872}
\right.
\\
&\,
+\frac{2409 E_4^6 E_6+5830 E_4^3 E_6^3+401 E_6^5}{5184}\, D
\\
&\, \left.
+\frac{2519 E_4^7+17978 E_4^4 E_6^2+5423 E_4 E_6^4}{1944}\, D^2\]\vths{4}_{\mu}.
\end{split}
\label{hmuX62}
\ee
Assuming that \eqref{hmuX62} is correct, one can produce additional boundary conditions for the direct integration method (see Table \ref{table_GVX62}), allowing in principle to reach genus 78, beyond the genus 63 available by standard methods.


\providecommand{\href}[2]{#2}\begingroup\raggedright\endgroup

\end{document}